\documentclass[11pt]{article}

\RequirePackage[OT1]{fontenc}
\usepackage[table]{xcolor}
\usepackage{enumitem} 
\usepackage{float}

\RequirePackage{amsmath,amssymb,amsthm,mathrsfs, makeidx,xcolor,multirow,pbox,tikz}
\RequirePackage{bm}

\RequirePackage[colorlinks,linkcolor=red,citecolor=blue,urlcolor=blue]{hyperref}

\RequirePackage{natbib}

\numberwithin{equation}{section}

\newcommand\blfootnote[1]{%
  \begingroup
  \renewcommand\thefootnote{}\footnote{#1}%
  \addtocounter{footnote}{-1}%
  \endgroup
}

\renewcommand{\epsilon}{\varepsilon}

\DeclareMathOperator*{\argmin}{argmin}

\DeclareMathOperator*{\Var}{Var}

\newcommand{\beq}{\begin{equation}}
\newcommand{\eeq}{\end{equation}}
\newcommand{\beqa}{\begin{equation} \begin{aligned}}
\newcommand{\eeqa}{\end{aligned} \end{equation}}
\newcommand{\beqas}{\begin{equation*} \begin{aligned}}
\newcommand{\eeqas}{\end{aligned} \end{equation*}}

\newcommand{\bit}{\begin{itemize}}
	\newcommand{\eit}{\end{itemize}}
\newcommand{\bmat}{\begin{bmatrix}}
	\newcommand{\emat}{\end{bmatrix}}

\makeindex

\theoremstyle{definition}\newtheorem{problem}{Problem}[section]
\theoremstyle{definition}
\theoremstyle{remark}\newtheorem{remark}{Remark}[section]
\theoremstyle{definition}
\theoremstyle{plain}\newtheorem{theorem}[problem]{Theorem}
\theoremstyle{plain}\newtheorem{lemma}[problem]{Lemma}
\theoremstyle{plain}\newtheorem{proposition}[problem]{Proposition}
\theoremstyle{plain}
\theoremstyle{plain}

\newtheorem*{assumption*}{\assumptionnumber}
\providecommand{\assumptionnumber}{}
\makeatletter 
\newenvironment{assumption}[1]
{%
  \renewcommand{\assumptionnumber}{Assumption #1}%
  \begin{assumption*}%
    \protected@edef\@currentlabel{#1}%
  }
  {%
  \end{assumption*}
}

\makeatletter 
\newenvironment{assumption2}[2]
{%
  \renewcommand{\assumptionnumber}{Assumption #1$_{#2}$}%
  \begin{assumption*}%
    \protected@edef\@currentlabel{#1}
  }
  {%
  \end{assumption*}
}

\newcommand{\GG}{{\mathbb G}}

\newcommand{\PP}{{\mathbb P}}
\newcommand{\QQ}{{\mathbb Q}}
\newcommand{\RR}{{\mathbb R}}

\newcommand{\EE}{\mathbb{E}}
\newcommand{\bs}[1]{\boldsymbol{#1}}
\newcommand{\mc}[1]{\mathcal{#1}}
\newcommand{\wh}{\widehat}

\usepackage{xparse} 
\usepackage{etoolbox} 
\NewDocumentCommand\DDnTwo{O{}}{
  \ifstrempty{#1}{
    \mathbb{D}_{n, t_1, t_2}
  }{
    \mathbb{D}_{n, {#1}}
  }
}
\NewDocumentCommand\twoArgs{O{}}{
  \ifstrempty{#1}{
    t_1, t_2
  }{
    {#1}
  }
}
\NewDocumentCommand\oneArg{O{}}{
  \ifstrempty{#1}{
    t
  }{
    {#1}
  }
}

\usepackage{versions}
\usepackage{mathtools}
\global\long\def\inv#1{\frac{1}{#1}}

\newcommand{\lp}{\left(} 
  \newcommand{\rp}{\right)}
\newcommand{\lb}{\left\{} 
  \newcommand{\rb}{\right\}}
\newcommand{\ls}{\left[} 
  \newcommand{\rs}{\right]}

\newcommand{\KR}[1]{{\noindent  \color{red}#1}}

\newcommand{\bb}[1]{\mathbb {#1}}

\newcommand\KXI{\mathcal{K}\kern-.09em\Xi_n} 
\newcommand\KKXI{\mathcal{K}\kern-.09em\mathcal{K}\kern-.09em\Xi_n} 
\newcommand\KF{\mathcal{K}\kern-.09em \mathcal{F}} 
\newcommand{\wt}{\widetilde}
\usepackage{bbm}
\newcommand{\one}{\mathbbm{1}}
\newcommand{\projf}{\Pi f}

\newenvironment{longform}{%
  \par \vspace{.2cm} \color{green} { \bf  Notes to self:}%
}{%
  $\blacktriangle$  \vspace{.1cm}
}

\excludeversion{mylongform}

\title{A nonparametric doubly robust test for a continuous treatment
  effect}

\author{Charles R.\ Doss$\dag$, Guangwei Weng$\dag$,  Lan Wang, \\
  Ira Moscovice, and Tongtan Chantarat}

\begin{document}

\maketitle

  \begin{abstract}

The vast majority of literature on evaluating the significance of a treatment effect based on observational data has been confined to discrete treatments.  These methods are not applicable to drawing inference for a continuous treatment, which arises in many important applications.  To adjust for confounders when evaluating a continuous treatment, existing inference methods often rely on discretizing the treatment or using (possibly misspecified) parametric models for the effect curve. Recently, Kennedy et al. (2017) proposed nonparametric doubly robust estimation for a continuous treatment effect in observational studies.  However, inference for the continuous treatment effect is a harder problem. To the best of our knowledge, a completely nonparametric doubly robust approach for inference in this setting is not yet available.  We develop such a nonparametric doubly robust procedure in this paper for making inference on the continuous treatment effect curve.  Using empirical process techniques for local U- and V-processes, we establish the test statistic's asymptotic distribution. Furthermore, we propose a wild bootstrap procedure for implementing the test in practice.  We illustrate the new method via simulations and a study of a constructed dataset relating the effect of nurse staffing hours on hospital performance.  We implement our doubly robust dose response test in the R package DRDRtest on CRAN.

  \end{abstract}

\section{Introduction}
\label{sec:intro}



\blfootnote{$\dag$Guangwei Weng and Charles R.\ Doss are jointly first authors on the paper.} We are interested in
hypothesis testing for a continuous (causal) treatment effect based on observational data.
The fundamental challenge of causal inference
with observational data
is 
to account for confounding variables,  which are
variables related to both the outcome and the treatment.
In the presence of confounding variables, it is well known that 
naive regression modeling does not lead to an unbiased 
estimate for the causal effect curve.
While continuous treatments are common in many important applications,
much of the existing literature on  
inference for a treatment effect from observational data 
has been focused on discrete treatments.
Relatively few methods are available for 
testing hypotheses about a continuous treatment effect curve.

Under the popular `no unmeasured
confounders' assumption,  
there are two broad directions to adjust for the confounding variables. 
A procedure can start by estimating
the outcome regression function, a function that relates the outcome to the treatment and the confounders, and then 
this can be weighted
appropriately to yield an estimate of the causal estimand. 
Alternatively, a procedure can start by
estimating the propensity score function, a function that relates the treatment
to confounders,
and then allows a variety of methods
to be implemented  to 
estimate the causal estimand.   
For instance, \cite{Imbens:2004cz} and \citet{Hill:2011bn} 
model only the outcome regression function; while
\citet{Hirano:2004in, Imai:2004gd, Galvao:2015ju}
model only the propensity score function.
Following the terminology of semiparametric
statistics, the outcome regression function and the propensity score function
are often referred to as nuisance parameters
(possibly infinite-dimensional).
In the aforementioned approaches, an incorrectly specified model 
for either nuisance parameter would lead to inconsistent estimates
of the treatment effect curve;
and, regardless, the conclusions are  susceptible to the curse of dimensionality: the rate of convergence of the estimator of
the treatment effect curve 
is the same as that of the estimator of the nuisance parameter, which may be high dimensional.

In the so-called doubly robust approach, widely used for estimating the 
average treatment effect in the
discrete treatment setting, one
estimates both nuisance parameters and then combines them.  
The term ``doubly
robust'' means that only one of the two nuisance parameters
needs to be estimated consistently to achieve consistent estimation 
of the causal treatment effect. Thus even if the model for one of the two nuisance parameters is misspecified, the causal estimand can still be estimated consistently if the other nuisance parameter model is correctly specified;
alternatively/similarly, 
the rate of the leading error term in estimating the causal estimand is determined by the
product of the error terms for estimating the two nuisance parameters,
allowing for efficient
estimation.

If one wishes to apply a doubly robust test in the continuous treatment setting, the simplest and likely the standard approach
would be to discretize the treatment and use the methodology for discrete
treatments \citep{robins2001comment,VanDerLaan:2003uu}.  Unfortunately, this
could result in misleading estimates, and can lead to possibly massive loss of power.
Also, in many applications maintaining the treatment as a continuous variable is important for post-analysis interpretation.
(As mentioned above,
\cite{Galvao:2015ju} develop inference procedures for the dose response curve but require good, possibly parametric, estimators for the propensity score (see their assumptions N.1, G.IV).)
If one
wishes to use doubly robust methods without discretization, then
\citet{Robins:2000gv} and \citet{Neugebauer:2007fg}
allow this, but requires specifying a  parametric model for the unknown causal treatment
effect curve. If the parametric model is not plausible, then the results can be unreliable. 
Recently, a nonparametric doubly robust estimation method has been proposed
(\citet{Kennedy:2017cq}),
allowing for greater flexibility in modeling the nuisance parameters.
Although the rates of nonparametrically estimating 
each nuisance function may be slower than $\sqrt{n}$,
the rate of
estimating the causal estimand 
may be much faster than that for estimating
either of the individual nuisance parameters (by virtue of the product rate discussed earlier),
while alleviating the difficulty of model specification
 for nuisance parameters.
 
 To the best of our knowledge, a completely nonparametric
doubly robust approach for inference for a continuous treatment effect 
is not yet available. 
It is worth noting that ``double robustness'' for  estimation does not automatically warrant ``double robustness'' for inference.
See, for instance, related discussions in \citet{van2014targeted} and \citet{Benkeser_Carone_Laan_Gilbert_2017}.





In this paper, we develop a doubly robust procedure for testing
the null hypothesis
that the treatment effect curve is constant. 
To do so, we introduce a test statistic based on comparing the
integrated squared distance from an estimate under the alternative to the null estimate.
We derive the limit distribution of the proposed test statistic.
In order to  implement the hypothesis test, the unknown parameters in the limit distribution must be estimated.
A natural approach is the bootstrap \citep{Efron:1993dc}.
Unfortunately, the naive bootstrap turns out to be inconsistent.
We propose a wild bootstrap procedure, which
provides provable guarantees for estimating
the limit distribution, and thus allows the test to be implemented.
Code that implements our doubly robust dose response test is available
in the R package `DRDRtest' on CRAN.

 Our main contribution is thus
a new doubly robust test procedure
which is consistent (in level and against fixed alternatives)
as long as at least one of the two nuisance parameters is specified correctly.
It requires only nonparametric assumptions on the nuisance parameters
unlike
\citet{Robins:2000gv} and \citet{Neugebauer:2007fg}.
The proposed test is doubly robust in the sense that 
the $p$-values we generate are reliable (uniformly distributed under the null hypothesis) even if one of the nuisance parameter models is misspecified. 
Some may
  argue that one may use machine learning to estimate the nuisance parameters to alleviate model misspecification.  Although this is true to some degree, popular machine learning
  methods such as random forests and neural networks
  are not immune from model misspecification;
  without structural assumptions (e.g., sparsity, additive structure) on the underlying model,
  they may have poor estimation accuracy (very slow rates of convergence).
  Their practical implementations also often require multiple tuning parameters.

Our statistic is inspired by the test of 
\cite{Hardle:1993ih} (also \citet{dette2001nonparametric}) which were developed in the non-causal setting.
Comparing with the non-causal setting,
our theory is significantly more complicated due to the
two infinite-dimensional nuisance parameters that are present in the causal inference setting. We introduce new empirical processes techniques for local U- and V-processes
to handle this complexity, which may be of independent interest for
nonparametric causal inference.




The rest of the paper is organized as follows. %
In Subsections~\ref{sec:nonp-hypoth-test} and \ref{sec:disc-recent-liter}, immediately following this one, we provide further discussions on related literature 
on nonparametric hypothesis testing and on
the continuous treatment effect setting, respectively. %
Section~\ref{sec:setup-method} introduces the setup, notation, the new testing procedure, and underlying assumptions. %
In Section~\ref{sec:main-results}, we present the main results. 
Section~\ref{sec:simulation} presents simulation studies and in Section~\ref{sec:real-analysis}
we present analysis of a dataset relating nurse staffing to hospital effectiveness.  

\subsection{Literature on nonparametric hypothesis testing}
\label{sec:nonp-hypoth-test}
The simpler, non-causal, problem of
hypothesis testing about a (non-causal) regression function when the
alternative is a large nonparametric class has a very large literature already.
There are many different approaches to this general problem;
to start, one must decide on the definition of the nonparametric alternative class.
The full, unrestricted, nonparametric alternative class is generally tested against by using a test based on a primitive of the function of interest: if $m(\cdot)$ is the regression function then $M(a):= \int_{-\infty}^a m(x)dx$ is the primitive.
This approach is possibly more familiar to readers in the density/distribution testing setting where $m$ and $M$ would be replaced by a density and cumulative distribution function, respectively.
Such tests based on $M$ are ``omnibus'' in the sense that in theory they have power approaching one against any fixed alternative.  However, for that theory to be relevant with certain fixed alternatives, extremely large sample sizes may be needed; or put another way, there are many alternatives that such tests for practical sample sizes  are not well powered.

Another set of procedures is based on taking the alternative class to be some sort of smoothness class (e.g., a H\"older, Sobolev, or Besov class \citep{Nickl_2016}).
Confusion may arise because some tests, e.g.\ those based on primitive functions, may have power against local alternatives converging at rate $n^{-1/2}$ whereas tests based on smoothness assumptions often require local alternatives to converge at a slower rate.  However, this is a case where the (local) rates of convergence can be misleading.
Rather than local rates, one can use global minimax rates of convergence over
a given class or classes to compare procedures.
\cite{Ingster:1993tv,Ingster:1993ub,Ingster:1993uo} studies minimax rates in nonparametric hypothesis testing problems (in a white noise model and in density estimation, both with simple null models).
We do not recount all the results here,
but note that
in general
tests based on primitives will not attain minimax optimal rates against smoothness-based alternatives 
(\citep[Section 2.5]{Ingster:1993tv}, \cite{Pouet_2001}). 
The minimax results are not just theoretical:
\cite{Eubank_LaRiccia_1992}  provide both theory and simulation results demonstrating that (in the context of density estimation) 
for any fixed sample size there are alternative sequences such that smoothness-based methods are more powerful than primitive-based methods (the Cram\'er-von Mises statistic in this case) even though the latter has a local $n^{-1/2}$ rate.

We develop our test statistic based on that of
\cite{Hardle:1993ih}
which is a smoothness-based test in the non-causal setting;
this allows us to develop a test that has power in all directions,
and to develop a test that is doubly robust.

One of the  difficulties in smoothness-based testing is
the issue of bias.
Nonparametric smoothness-based estimates generally have nontrivial bias which must be accounted for and the estimation of which entails complications.
In our particular testing setting, actually there is no bias under the null (since our null hyptohesis is the class of constant functions), so bias is not a major issue in the usual way. The bias will affect the estimator under the alternative and so will affect the power.
A related issue is that in nonparametric testing, the asymptotic distributions of many test statistics
have the property that their bias (mean) is of a larger order of magnitude than their variance.
One consequence of this for us is that it makes the asymptotically negligible error terms in the analysis of our
test statistic (in the causal setting of the current paper)  much more
complicated than they would otherwise be; it turns out that the large bias of
the main term gets multiplied by other error terms (after expanding a square)
and  this requires extra mathematical analysis. 


\subsection{Literature on continuous treatment effects}
\label{sec:disc-recent-liter}
In the last few years there has been significant and increasing interest in causal inference with continuous treatment effects;
this includes interest in the setting of optimal treatment regimes
\citep{Kallus:2018up},
and in specific scientific areas
(e.g., \cite{Kreif:2015cp} in the health sciences).

We briefly discuss here several statistics papers that have built theory and/or
methods related to causal effect estimation and/or inference in the presence
of a continuous treatment (based on observational data).  We start with
\citet{Kennedy:2017cq}, on which other works, including the present paper,
build.  \citet{Kennedy:2017cq} have developed a method for efficient doubly
robust estimation of the treatment effect curve.  
Denote the outcome
regression function by $\mu$ or $\mu_0$, and denote the propensity score function
by $\pi$ or $\pi_0$.
Their method is
  based on a pseudo-outcome $\xi \equiv \xi(\bs Z; \pi, \mu)$,
  which depends on the sample point $\bs Z$, and on the nuisance functions
  $\pi$, $\mu$.  The pseudo-outcome $\xi$ has the key double robustness property that if {\it either} $\pi = \pi_0$ or $\mu = \mu_0$, then $\EE( \xi( \bs Z; \pi, \mu) | A = a)$ is equal to the treatment effect curve (at the treatment value $a$).
  The estimation procedure of   \citet{Kennedy:2017cq} is then a natural two-step procedure: (1) estimate the nuisance functions $(\pi_0, \mu_0)$ by some estimators $(\wh \pi, \wh \mu)$ which the user can choose as they wish and
construct (observable) pseudo-outcomes $\wh \xi_i$
(which approximate $\xi_i$ and depend on $\wh \pi,$ $\wh \mu$), and (2)
regress the pseudo-outcomes on $A$ using some nonparametric method (e.g., local linear regression).  As we described above, the error term from the nuisance parameter estimation is given by the product of the error term for estimating $\pi_0$ and for estimating $\mu_0$, so is smaller than either, partially alleviating the curse of dimensionality.



Several works have now made use of the pseudo-outcome approach of
\citet{Kennedy:2017cq}, or similar approaches.
\citet{Westling_Gilbert_Carone_2020, Semenova:2017uc} use the pseudo-outcomes
of \citet{Kennedy:2017cq} with alternative estimation techniques, and \citet{Colangelo:2020tt, su2019non} use similar
pseudo-outcomes (and study particular nuisance estimators).
Like
\citet{Kennedy:2017cq},
\citet{Westling_Gilbert_Carone_2020} also develop a doubly robust estimator of a continuous treatment effect curve; they develop a different procedure,
based on the assumption that the true effect curve satisfies the shape constraint of monotonicity.
\citet{Colangelo:2020tt} provide an alternative motivation for a related pseudo-outcome 
\citet{Kennedy:2017cq}, study a sample-splitting variation of the
estimation methodology of
\citet{Kennedy:2017cq}, and also consider estimating the gradient of
the treatment curve.
These works do not consider teh global inference problem that we address here. 
\begin{mylongform}
  \begin{longform}
    The assumptions of \citet{Colangelo:2020tt} exclude double robustness (see their Assumption 3).
They find pointwise limit distributions, 
like \citet{Kennedy:2017cq}, but do not  address the global inference problem that we address here.

  \end{longform}
\end{mylongform}

Many works since \cite{Kennedy:2017cq} have considered doubly robust estimation of structural/causal functions based on nonparametric models.
Some include or focus on continuous treatment effects, while others
focus more on the problem of conditional average treatment effect (CATE) (or ``partially conditional average treatment effect'' (PCATE)
\citep{Wang_Wong_Yang_Chan_2021}) based on a binary treatment variable, or other related quantities.  The (P)CATE setting is of course different than the continuous treatment setting we consider, but may share some features with our setting when the covariates on which the treatment is  conditioned are continuous, so we discuss some of the recent literature briefly.
\cite{Chernozhukov_Chetverikov_Demirer_Duflo_Hansen_Newey_Robins_2018},
\cite{Semenova:2017uc},
\cite{Chernozhukov_Newey_Singh_2022}
develop general ``double/debiased'' machine learning approaches to estimating causal estimands that are continuous functions.
\cite{Chernozhukov_Newey_Singh_2022} 
develop a Dantzig-type estimator based on estimating equations for the nuisance parameters.
They consider four running example stimands, as well as 
 ``local'' and ``perfectly localized'' functionals, the latter including the continuous treatment effect at a fixed point.
They develop Gaussian approximations for the distributions of their 
estimators
at a fixed point.
They do not further develop inference methods, so their focus is distinct from our focus on a global testing problem.\footnote{In discussing works that focus on continuous treatment effect estimation,
they say ``These works develop inference on perfectly localized average potential outcomes with continuous treatment effects, using a different approach than what we develop here.  Our development is complementary as it covers a much broader collection of functionals.''}
\cite{Semenova:2017uc} develop a general theory for debiased machine learning and uniform confidence bands for inference for  different causal or missing data estimands, such as conditional average treatment effects, regression functions with partially missing outcomes, and conditional average partial derivatives.
They
also consider the causal effect curve with continuous treatments.
However, in the latter setting,
their assumptions are slightly too strong to allow  double robustness
for (pointwise) consistency,\footnote{``The only requirement we impose on the estimation of [the nuisance parameters] is that [they converge] to the true nuisance parameter $\eta_0$ at a fast enough rate [$o_p(n^{-1/4-\delta})$] for some $\delta \ge 0$'' \citep[page 271]{Semenova:2017uc};  see their Assumption~4.9.}
and their confidence band is centered at an approximation of the true function, rather than at the true function itself (i.e., there is an error term that is ignored; see their Theorem 4.7).
\begin{mylongform}
  \begin{longform}
    Note: I find their assumptions very hard to parse.  In assumption 4.9 If you take theri $d = N^{1/5}$ then i guess they have $r_n s_n = o_p(N^{-7/10})$?

    and I don't know what the right value for $\xi_d$ is on which the first rate statement in assumption 4.9 depends.

    (Separately, we also believe their Assumption 3.1 may rule out truly
    applying their results to the full continuous treatment causal effect
    curve rather than to a discrete approximation thereof.)

  \end{longform}
\end{mylongform}
In summary, the theory and methodology of
\cite{Semenova:2017uc,Chernozhukov_Newey_Singh_2022}
is built for a variety of settings and does not focus exclusively on the setting of continuous treatments, which we do focus on, and derive a powerful procedure for, here.

\cite{Wang_Wong_Yang_Chan_2021,luedtke2016super, lee2017doubly, zimmert2019nonparametric}
and \cite{foster2019orthogonal},
consider various doubly robust types of estimation procedures for (P)CATE estimation, generally based on pseudo-outcomes.
Here ``PCATE'' means the effect of a binary/discrete treatment conditional on some covariates  which may be a strict subset of the set of all confounding variables.
\cite{Nie_Wager_2021} and \cite{Kennedy_2020}
consider a different type of estimator
(dubbed the ``R-learner'' and ``lp-R-Learner'', respectively, after \cite{Robinson_1988}); they provide double-robust-type conditions
under which these two-step  estimators  can attain the oracle 
rate of convergence,
with \cite{Kennedy_2020} able to weaken the conditions on the nuisances given in 
\cite{Nie_Wager_2021} by a new cross-validation technique (inspired by \cite{Newey_Robins_2018})
and undersmoothing.
Very recently \cite{Kennedy_Balakrishnan_Wasserman_2022} (on the arxiv) find minimax lower and upper bounds for CATE estimation
(showing that the estimator/rates of \cite{Kennedy_2020} were optimal in some smoothness regimes but not others).
None of the above works consider hypothesis tests or questions of inference
except for \cite{lee2017doubly}. 
We do not know of an analog of the (lp-)R-Learner that has been directly studied for the case of estimation of continuous treatments.
An open question is whether the approach we develop here can be applied also in the setting of inference for a (P)CATE.  
In a different setup \citet{Luedtke_Carone_Laan_2019} develop a hypothesis test in a causal inference setting which could allow for continuous treatments.   However, their null hypothesis is different than ours and their conditions (Condition 3) rule out our setting. 

Very recently, \citet{westling2020nonparametric} considered a similar problem to the one we consider, but using a different test statistic, based on a ``primitive'' (or ``anti-derivative'') of the treatment effect curve.
The present paper and \citet{westling2020nonparametric} were developed entirely separately.
A strength of the primitive-based method of
\citet{westling2020nonparametric}
is that it naturally handles mixed discrete-continuous exposures, whereas our method would require further modifications (e.g., locally chosen bandwidths) to do so.  
\citet{westling2020nonparametric}'s method is not quite doubly robust
(requiring an $o_p(n^{-1/2})$ rather than $O_p(n^{-1/2})$ product-nuisance-estimation-rate), whereas our method
(which allows an $o_p( (n\sqrt{h})^{-1/2})$ rate)
is.  (See  \citet{westling2020nonparametric}'s Assumption A4 and Theorems 3 and 4, as well as Figure 1.)  Also, the two methods have noticeably different power in different scenarios.
\citet{westling2020nonparametric}'s test has a local $n^{-1/2}$ convergence rate, whereas our test statistic has a local convergence rate of $(n \sqrt{h})^{-1/2}$.  The implications are as discussed above: \citet{westling2020nonparametric}'s test will have power focused in ``one direction'' and decaying in directions away from that one direction, and our power will be  more uniformly spread over the alternatives.  For instance,
we will have noticeably higher power 
when the true (alternative) effect curve has significant peaks or valleys.




\section{Setup and method}
\label{sec:setup-method}

\subsection{Notation, data and setup}
We use the following notation throughout. Let $(\bs{Z}_1,\ldots,\bs{Z}_n)$ be
the observed sample where each observation is an independent copy of the tuple
$\bs{Z}=(\bs{L},A,Y)$ with support $\mc{Z}=\mc{L}\times\mc{A}\times\mc{Y}$. Here
$\mc{L}\subseteq\RR^d$ and $\bs{L}$ is the vector containing the  $d$ potential confounder variables (covariates);
$\mc{A}\subseteq \RR$ and  $A$ is the continuous treatment dosage
received; $\mc{Y}\subseteq\RR$ and  $Y$ is the observable outcome of interest. We let $P$ denote the distribution
of $\bs{Z}$ and $p_0(\bs{z})=p_0(y|\bs{l},a)p_0(a|\bs{l})p_0(\bs{l})$ denote the corresponding density function with respect
to some dominating  measure $\nu$.
Let $\mu_0(\bs{l},a):=\EE(Y|\bs{L}=\bs{l},A=a)$ denote
the outcome regression function. 
Similarly, let $\pi_0(a|\bs{l}):=\frac{\partial}{\partial
  a}P(A\le a|\bs{L}=\bs{l})$ denote the conditional density or propensity score
function of $A$ given $\bs{L}$ and $\varpi_0(a):=\frac{\partial}{\partial a}P(A\le
a)$ denote the marginal density function of $A$
(both of which densities are assumed to exist).
For a function $f$ on $\RR$, 
we let $\PP\{f(\bs{Z})\}:=\int_{\mc{Z}}f(\bs{z})\,dP(\bs{z})$.
And for $p\ge 1$,
we use $\|f\|_p:=\{\int f(\bs{z})^p\,d P(\bs{z})\}^{1/p}$ to denote the $L_p(P)$
norm and use $\|f\|_{\mc{X}}:=\sup_{x\in \mc{X}}|f(x)|$ to denote the uniform
norm over the range $\mc{X}$. We use $\PP_n$ to denote the empirical
distribution defined on the observed data so that $\PP_n\{f(\bs{Z})\}:=\int
f(\bs{z})\,d\PP_n(\bs{z})=n^{-1}\sum^n_{i=1}f(\bs{Z}_i)$. 

To characterize the problem, let $Y^a$ be the potential outcome \citep{Psychology:uk} when treatment
level $a$ is applied. Then the causal estimand that we are interested in learning about (developing a hypothesis test for) is
$\theta_0(a):=\EE(Y^a)$,
and we wish to test if this function is constant.
Specifically, we want to test
\begin{equation}
  \label{eq:11}
  H_0\colon
  \theta_0 \equiv c \in \RR
  \, \text{ versus } \,
  H_1\colon \theta_0 \text{ is nonconstant},
\end{equation}
where  we will assume 
that
$\theta_0(\cdot)$
satisfies some smoothness assumptions
if it is nonconstant. 

\subsection{Proposed method}
\label{sec:proposed-method}
In \citet{Kennedy:2017cq}, the authors derived a doubly robust mapping for
estimating the continuous treatment effect curve. Like doubly robust estimators
in binary treatment cases, the doubly robust mapping depends on both the outcome
regression function and the propensity score function, and can be written as
\begin{align}
  \label{eq:dr-mapping}
  \xi(\bs{Z};\pi,\mu) = \frac{Y-\mu(\bs{L},A)}{\pi(A|\bs{L})}
  \int_{\mc{L}}\pi(A|\bs{l})\,dP(\bs{l})+\int_{\mc{L}}\mu(\bs{l},A)\,dP(\bs{l}).
\end{align}
where $\pi(a| \bs l)$ and $\mu(\bs l, a)$ are some propensity score and outcome regression functions, respectively.
The above mapping has the desired property of double robustness in that
\begin{align}
  \label{eq:dr-property}
  \EE\lb\xi(\bs{Z};\pi,\mu) | A = a\rb
  =\theta_0(a),
\end{align}
provided either $\mu=\mu_0$ or $\pi=\pi_0$,
under Assumptions~\ref{assm:IA} below \citep{Kennedy:2017cq}.
Thus $\theta_0(\cdot)$ could be
estimated using standard nonparametric smoothing techniques if
either $\mu_0$ or $\pi_0$ were known.  Since we do not actually know $\mu_0,\pi_0$,  we plug in estimators
$\widehat{\mu}$, $\widehat{\pi}$ for $\mu_0$, $\pi_0$.
To compute $\xi$ we also need to know $dP(\bs l)$ in two places; since we do not, we plug
in
$\PP_n(\bs{l})$ for $P(\bs{l})$, and we denote this by $\wh \xi$.
Thus, our estimate of the pseudo-outcome  $\xi( \bs Z; \pi_0, \mu_0)$ is 
\begin{align}
  \label{eq:pseudo-outcome}  \widehat{\xi}(\bs{Z};\widehat{\pi},\widehat{\mu})=\frac{Y-\widehat{\mu}(\bs{L},A)}{\widehat{\pi}(A|\bs{L})}\int_{\mc{L}}\widehat{\pi}(A|\bs{l})\,d\PP_n(\bs{l})+\int_{\mc{L}}\widehat{\mu}(\bs{l},A)\,d\PP_n(\bs{l}).
\end{align}
We can then apply a nonparametric estimation procedure to the (observed) tuples
$\{(\widehat{\xi}(\bs{Z}_i;\widehat{\pi},\widehat{\mu}),A_i)\}_{i=1}^n$.
\citet{Kennedy:2017cq} show that when at least one of the estimators  is consistent then, 
under some assumptions about  complexity and boundedness conditions of  $\wh
\mu$ and $\wh \pi$
and the product of their convergence rates,  the  convergence rate of the nonparametric estimator is
the same as if we know the true $\mu_0$ or $\pi_0$.
\citet{Kennedy:2017cq} apply  a local linear estimator to the pseudo-outcome
to estimate $\theta_0(\cdot)$   and show  the above-stated property on the pointwise convergence
rate of the nonparametric estimator.

In this paper, we are
interested in a different problem: testing if $\theta_0(\cdot)$ is constant or not. 
As in the estimation problem, in the setting where we (unrealistically) know one of $\mu_0$ or $\pi_0$,
testing whether $\theta_0(\cdot)$ is constant becomes a standard
regression problem, and we can consider many possible nonparametric tests to the tuples $\{(\widehat{\xi}(\bs{Z}_i;\widehat{\pi},\widehat{\mu}),A_i)\}_{i=1}^n$.
As described in Section~\ref{sec:intro}, not all tests will be doubly robust for testing, though.

In
\citet{Hardle:1993ih}, the authors consider the problem of testing parametric
null linear models (not in a causal setting) and construct test statistics based on the integrated difference
between the nonparametric model estimated using the Nadaraya-Watson estimator and the
parametric null model. \citet{alcala1999goodness} extended the test to allow
using  a
local
polynomial estimator \citep{Fan:1996tk} for the nonparametric model.
We propose to test our hypothesis
\eqref{eq:11} of a constant treatment effect curve (i.e., no treatment effect)
using the following statistic
\begin{align}
  \label{eq:test-stat}
  T_n=n\sqrt{h}\int_{\mc{A}}\lp\widehat{\theta}_h(a)-\PP_n\widehat{\xi}(\bs Z) \rp^2w(a)\,da,
\end{align}
where $w(\cdot)$ is a user-specified weight function, $\widehat{\theta}_h(a)$ is the local linear estimator applied to
$\{(\widehat{\xi}(\bs{Z}_i;\widehat{\pi},\widehat{\mu}),A_i)\}_{i=1}^n$.  To define the local linear estimator, we let
$  \widehat{\beta}_h(a)=\argmin_{\beta\in \RR^2}\PP_n
\big[ K_{ha}(A)\lb\widehat{\xi}(\bs{Z};\widehat{\pi},\widehat{\mu})-g_{ha}(A)^T\beta\rb^2 \big],$
where
$g_{ha}(t)=(1,\frac{t-a}{h})^T$, $K_{ha}(t)=h^{-1}K\{(t-a)/h\}$, and $K(\cdot)$ is a kernel function, and then we let
$\widehat{\theta}_h(a)=g_{h,0}(0)^T\widehat{\beta}_h(a)$.
Note: we could define the test statistic $T_n$ by a summation over $A_i$ rather than as an integral against (Lebesgue measure) $da$; as mentioned in \citet{Horowitz_Spokoiny_2001}, \citet{dette2001nonparametric}, similar results as ours would hold although some constants would change.
Note that under the null hypothesis of no
treatment effect, $\PP_n\{\xi(\bs{Z};\pi_0,\mu_0)\}$ is an $\sqrt{n}$-consistent estimator of the null
model and so is $\PP_n\widehat{\xi}$ provided $\widehat{\mu}, \widehat{\pi}$ are
not converging too slowly (see the later discussion). We will see under some mild
conditions  on the convergence rates of $\widehat{\pi}$ and
$\widehat{\mu}$, that $T_n$ converges to a  normal distribution similar to the one in
\citet{alcala1999goodness} under the null model (given in \eqref{eq:11}). However, similar to \citet{Hardle:1993ih} and
\citet{alcala1999goodness}, due to the slow order of convergence of the
asymptotically negligible terms that arise in the proof, we do not suggest using the
target distribution  given in Theorem~\ref{thm:null-normality} to directly calculate the
critical values under the null hypothesis. Instead 
we advocate using the bootstrap \citep{Efron:1993dc} to
estimate the distribution of $T_n$ to improve the finite sample performance and,
more specifically, we use the so-called wild bootstrap \citep{davidson2008wild} as used
in  \citet{Hardle:1993ih},
\citet{alcala1999goodness} and  \citet{dette2001nonparametric}. In
\citet{Hardle:1993ih}, the authors show the theoretical properties of  three
different bootstrap methods: (1) the naive resampling method; (2) the adjusted
residual bootstrap; (3) the wild bootstrap and showed only the wild bootstrap
gives consistent estimation of the null distribution. These results are again true in our setting: the wild bootstrap is valid
whereas the other two are not. 
Here we provide a brief outline of our proposed test procedure.


\begin{enumerate}
  \label{item:1}
\item Estimate
  $(\pi_0,\mu_0)$ by
  (black-box estimators)
  $(\widehat{\pi},\widehat{\mu})$.
\item Calculate the pseudo-outcomes
  $\widehat{\xi}(\bs{Z};\widehat{\pi},\widehat{\mu})$  by \eqref{eq:pseudo-outcome}  and construct the local
  linear estimator $\widehat{\theta}_h(a)$ using
  $\{(\widehat{\xi}(\bs{Z}_i;\widehat{\pi},\widehat{\mu}),A_i)\}_{i=1}^n$.
\item To generate wild bootstrap samples to estimate the distribution of
  $T_n$ under the null hypothesis,
  \begin{enumerate}
  \item Calculate
    $\hat{\epsilon}_i=\widehat{\xi}(\bs{Z}_i;\widehat{\pi},\widehat{\mu})-\widehat{\theta}_h(A_i)$ 
    (we can also use $\hat{\epsilon}_i=\widehat{\xi}(\bs{Z}_i;\widehat{\pi},\widehat{\mu})-\sum_{i=1}^n\widehat{\xi}(\bs{Z}_i;\widehat{\pi},\widehat{\mu})/n$),
  \item Do the following $B$ times, where $B$ is the desired number of bootstrap resamplings: for each $i \in \{1,\ldots,n\}$, generate $\epsilon^{\ast}_i\sim \hat{F}_i$
    (defined just below, based on $\{\hat \epsilon_i \}$) 
    and use
    $(\xi_i^{\ast}=\PP_n\wh\xi+\epsilon^{\ast}_i,A_i)$ as
    bootstrap observations,
  \end{enumerate}

\item Use the wild bootstrap  samples to compute $T_{n,j}^*$, $j=1,\ldots, B$ (according to
  \eqref{eq:test-stat} but using the bootstrap samples) 
  and use $\{ T_{n,j}^*\}_{j=1}^B$  to  estimate the distribution of
  $T_n$ under the null hypothesis.
  Let $\widehat{t}^{\ast}_{n,1-\alpha}$
  denote the  $1-\alpha$ quantile of the estimated distribution, where
  $0 < \alpha < 1$ is the predetermined significance level. Reject the
  null hypothesis if $T_n>\widehat{t}^{\ast}_{n,1-\alpha}$.
\end{enumerate}

When generating bootstrap samples, we use $\hat{F}_i$ to estimate the
conditional distribution of $\xi(\bs{Z}_i;\bar{\pi},\bar{\mu})$ based on the
single residual $\hat{\epsilon}_i$. \citet{Hardle:1993ih} use a ``two point
distribution''  which matches the first three moments of $\hat{\epsilon}_i$ and is
defined as
\begin{align}
  \label{align:defn:Fhati}
  \epsilon_i^{\ast}=\lb
  \begin{array}{rr}
    -\hat{\epsilon}_i(\sqrt{5}-1)/2&\text{with probability }
                                     (\sqrt{5}+1)/(2\sqrt{5})\\
    \hat{\epsilon}_i(\sqrt{5}+1)/2&\text{with probability } (\sqrt{5}-1)/(2\sqrt{5}).
  \end{array}\right.
\end{align}
We also consider another common choice, a Rademacher type distribution, where $ \epsilon_i^{\ast}$ equals $ \hat{\epsilon}_i$ or $- \hat{\epsilon}_i$ with probability $1/2$ each.  (\citet{davidson2008wild}).  Unlike the two point distribution, the Rademacher distribution matches the first two and the fourth (and all even) moments of $\hat{\epsilon}_i$, but imposes symmetry on $\hat{F}_i$.


\begin{remark}
  In Section~\ref{sec:extens-test-finite}
  we also present an extension of the doubly robust
  pseudo-outcome to allow for possible (discrete or continuous) effect
  modifiers, and present the natural extension of the test to the case where
  the effect modifier is discrete.
\end{remark}

\subsection{Assumptions}
\label{sec:assumptions}

Here we introduce the assumptions needed for our theoretical results.
Our parameter of interest, $\theta_0(a)$, is defined on the potential outcome $Y^a$ which is not observable. Thus, we need the
following 
identifiability 
conditions on the observed data.
\begin{assumption}{I}
  \label{assm:IA}
  $\phantom{blah}$
  \begin{enumerate}
  \item \label{assm:IA-item1}Consistency: $A=a$ implies $Y=Y^a$.
  \item \label{assm:IA-item2}Positivity: $\pi_0(a|\bs{l})\ge \pi_{\min}>0$ for all $\bs{l}\in
    \mc{L}$  and all $a \in A$.
  \item \label{assm:IA-item3}Ignorability:
    $\EE(Y^a|\bs{L},A)=\EE(Y^a|\bs{L})$.
  \end{enumerate}
\end{assumption}

We need some further assumptions to regulate the distribution of the observed
data and the treatment effect curve $\theta_0(a)$.
\begin{assumption}{D}
  \label{assm:DA}
  $\phantom{blah}$
  \begin{enumerate}
  \item  \label{assm:DA-item1} The support of $A$ (i.e., $\mc{A}$), is a compact subset of $\RR$.
  \item  \label{assm:DA-item2} The treatment effect curve $\theta_0(a)$ and the marginal density function
    $\varpi_0(a)$ are twice continuously differentiable.
  \item \label{assm:DA-item3}The conditional density $\pi_0(a|\bs{l})$ and the outcome regression
    function $\mu_0(\bs{l},a)$ are uniformly bounded.
  \item  \label{assm:DA-item4}Let $\tau(\bs{l},a) := \Var(Y|\bs{L}=\bs{l},A=a)$ be the conditional
    variance of $Y$ given  covariates and treatment level. Assume there exist
    $\tau_{\max}>0$ such that
    $0<\tau(\bs{l},a)\le \tau_{\max}$ for all $\bs{l}\in\mc{L}$ and $a\in\mc{A}$. Moreover, define
    \begin{align*}
      &S_{\tau}
        :=\{\bs{l}\in \mc{L}:\tau(\bs{l},a)\text{ is a continuous function of
        }a\},\\
      &S_{\pi_0}
        :=\{\bs{l}\in \mc{L}:\pi_0(a|\bs{l})\text{ is a continuous function of
        }a\},\\
      &S_{\mu_0} :=\{\bs{l}\in \mc{L}:\mu_0(\bs{l},a)\text{ is a continuous
        function of }a\};
    \end{align*}
    assume     we have $P(S_{\tau} \cup S_{\pi} \cup
    S_{\mu})=1$. 
  \end{enumerate}
\end{assumption}
\noindent
Statement~\ref{assm:DA-item4} %
about the sets $S_\tau$, $S_{\pi_0}$, $S_\mu$ is
just a slight relaxation of the requirement that the given functions all be simultaneously almost surely continuous everywhere.
We also impose some conditions on the estimators of the nuisance parameters and
the local linear estimator of the treatment effect curve. We assume the
estimators for nuisance parameters fall in classes with finite uniform entropy
integrals. For a generic class of functions $\mc{F}$, let $F$ denote an envelope
function for $\mc{F}$, i.e., $\sup_{f\in \mc{F}}|f|\le
\mc{F}$. Let $N(\epsilon,\mc{F},\|\cdot\|)$ denote the covering number, i.e.,
the minimal number of $\epsilon$-balls (with distance defined on $\|\cdot\|$)
needed to cover $\mc{F}$.
Let
\begin{align}
  \label{eq:uniform-entropy-2-integral}
  J_m(\delta,\mc{F}, L_2 )
  :=
  \int^{\delta}_0\sup_Q \lp 1 + \log N(\epsilon\|F\|_{Q,2},\mc{F},L_2(Q)) \rp^{m/2}
  \,d\epsilon,
\end{align}
where the sup is over all probability measures $Q$ and $L_2(Q)\equiv\|\cdot\|_{2,Q}$ is the
$L_2$ semimetric under the distribution $Q$, i.e., $\|f\|_{2,Q}=(\int
f^2\,dQ)^{1/2}$.
If  $J_1(1, \mc{F},L_2) < \infty$ we say $\mc{F}$ has a finite uniform entropy integral.
We will at times require $J_m(1, \mc{F}, L_2) < \infty$ for differing values of $m \in \lb 1,2,3,4 \rb$.  Following standard convention, we sometimes let $J(\cdot,\cdot,\cdot)$ refer to $J_1(\cdot,\cdot,\cdot)$.
For the assumption on our kernel, we also need to define a
\emph{Vapnik-Chervonenkis (VC)} (Dudley, 1999) class.
If a class of functions $\mc{F}$ is a VC class, we have that
\begin{align}
  \label{eq:def-vc}
  \sup_{Q} N\lp\tau\|F\|_{2,Q},\mc{F},L_2(Q) \rp \le \lp\frac{C}{\tau}\rp^v
\end{align}
for some positive $C,v$ and all $\tau > 0$ (and again the sup is over all probability measures $Q$).
The assumptions we make on our estimators are as follows.
\begin{assumption}{E(A)}
  \label{assm:EA}
  $\phantom{blah}$
  \begin{enumerate}
  \item \label{assm:EA-item3} The bandwidth $h\equiv h_n$ fulfills $c^h_1n^{-1/5}\le \liminf
    h_n\le \limsup h_n\le c^h_2n^{-1/5}$ for some constants $0<c^h_1\le c^h_2<\infty$.
  \item  \label{assm:EA-item1}Let $\bar{\pi}$ and $\bar{\mu}$ denote the limits of the estimators
    $\widehat{\pi}$ and $\widehat{\mu}$ such that
    $\|\widehat{\pi}-\bar{\pi}\|_{\mc{Z}}=o_p(\sqrt{h})$ and
    $\|\widehat{\mu}-\bar{\mu}\|_{\mc{Z}}=o_p(\sqrt{h})$, where $h$ is the bandwidth
    used in local linear estimator. And we have either
    $\bar{\pi}=\pi_0$ or $\bar{\mu}=\mu_0$.
  \item  \label{assm:EA-item4}The kernel function $K$ for the local linear estimator is a
    continuous symmetric probability density function with support on
    $[-1,1]$. Moreover, we assume the 
    class of functions       $\{ K((\cdot-a)/h):a\in \RR,h>0
    \}$ satisfies
    condition \eqref{eq:def-vc}.
  \item  \label{assm:EA-item5}

    Let  $r^{\infty}_n$ and $s^{\infty}_n$ be such that
    \begin{align*}
      \sup_{a\in\mc{A}} \|\widehat{\pi}(a|\bs{L})-\pi_0(a|\bs{L})\|_2 &= O_p(r^{\infty}_n)\\
      \sup_{a\in\mc{A}}  \|\widehat{\mu}(\bs{L},a)-\mu_0(\bs{L},a)\|_2&=O_p(s^{\infty}_n).
    \end{align*}
    We  assume  $s^{\infty}_nr^{\infty}_n
    = o\{(n\sqrt{h})^{-1/2}\}$.
  \end{enumerate}
\end{assumption}

Assumption~\ref{assm:EA}.\ref{assm:EA-item3} requires that $h$ is of the order of magnitude for optimal estimation; such $h$ can be achieved a variety of ways (for instance, one can minimize a risk estimate, or perform cross validation \citep{Fan:1996tk}).  Assumption~\ref{assm:EA}.\ref{assm:EA-item1} is not a stringent assumption; by definition $\overline \pi, \overline \mu$ are the limits of the estimators $\wh \pi, \wh \mu$; here we require the rate of convergence (in $\| \cdot \|_{\mc Z}$) to these limits (which are not necessarily the truth) to be order $\sqrt{h} = o(n^{-1/10})$ which is quite slow.

Assumption~\ref{assm:EA}.\ref{assm:EA-item4} is a standard  assumption on the user-chosen kernel.
Assumption~\ref{assm:EA}.\ref{assm:EA-item5} is a somewhat non-standard
assumption, since it combines $L_\infty$ and $L_2$ norms.
The $L_2$ aspect arises in the local asymptotics of
\citet{Kennedy:2017cq}, and the $L_\infty$ aspect arises because we consider a global test.
Note that in
parametric settings, $r_n^\infty$ or $s_n^\infty$ may attain $\sqrt{n}$
rates.  For instance, in a linear regression, if the regression model is
$\mu(\bs l, a) = (\bs l^T, a)\bs{\beta}$, and $\wh{\bs{\beta}}$ is an
estimator converging to the true parameter $\bs{\beta}_0$ at $\sqrt{n}$
rate, then
$\PP\lp (\bs l^T,a)(\wh{\bs{\beta}} - \bs{\beta}_0) \rp^2
\le 2 O_p(n^{-1} ) (a^2 + \PP \| \bs l \|_2^2 ),$
which follows from 
using the inequality $(a+b)^2 \le 2a^2+2b^2$ and the fact that if
$(\wh{\bs{\beta}} - \bs{\beta}_0) = O_p(n^{-1/2})$ then
$(\wh{\bs{\beta}} - \bs{\beta}_0) (\wh{\bs{\beta}} - \bs{\beta}_0)^T$ has
eigenvalues of order $O_p(n^{-1})$.  Taking a square root and a supremum
over the bounded set $a \in \mc{A}$ shows that in this case $r_n^\infty$ is
order $\sqrt{n}$.  Similar results hold in other parametric models for $\mu$ or for $\pi$.   Thus indeed, the test is ``doubly robust'':  if one parametric model is well-specified and attains root-$n$ rates, the other may be misspecified.

Many nonparametric or semiparametric examples will also satisfy
Assumption~\ref{assm:EA}.\ref{assm:EA-item5}; if $r_n^\infty$ and $s_n^\infty$ are both, say, order $n^{-2/5}$ up to polylogarithmic factors, which is the rate one expects from for instance a generalized additive model under twice differentiability, then the assumption is satisfied.

In addition to the above
\ref{assm:EA}
(`Estimator assumption part A')
assumption, we make one more assumption (`Estimator assumption part B'). We subscript this next assumption by $m$, which corresponds to the requirement that $J_m(1, \mc{F}, L_2) < \infty$.  Differing results will require differing values of the integer $m$ in $\{1,2,3,4\}$.    We label the below assumption as ``\ref{assm:EA-item2}$_{m}$'' and when we want to assume that, for example, $J_3(1, \mc{F}, L_2) < \infty$, we  refer to the assumption as ``\ref{assm:EA-item2}$_3$''.
\begin{assumption2}{E(B)}{m}
  \label{assm:EA-item2} 
  The estimators $\widehat{\pi}$, $\widehat{\mu}$ and their limits
  $\bar{\pi}$, $\bar{\mu}$ are contained in uniformly bounded function
  classes
  $\mc{F}_\pi$, $\mc{F}_\mu$,
  which
  satisfy that
  $J_m(1, \mc{F}, L_2) < \infty$ for $\mc{F} = \mc{F}_\pi$ or
  $\mc F = \mc{F}_\mu$,
  with $1/\widehat{\pi}$  also uniformly bounded.
  Moreover,
  we assume    $P(S_{\bar{\pi}}\cup S_{\bar{\mu}})=1$, where we let
  \begin{align*}
    S_{\bar{\pi}}
    &   :=\{\bs{l}\in \mc{L}:\bar{\pi}(a|\bs{l})\text{ is a continuous function of }a\} \\
    S_{\bar{\mu}} & :=\{\bs{l}\in \mc{L}:\bar{\mu}(\bs{l},a)\text{ is a continuous
                    function of }a\}.
  \end{align*}
\end{assumption2}

\begin{remark}

  To control rates of convergence, we make assumptions on the complexity of the classes being considered in  \ref{assm:EA-item2}.  For our first main theorem (limit distribution of the test statistic) we require \ref{assm:EA-item2}$_3$ i.e.\ $m=3$ and for our bootstrap theorem we require \ref{assm:EA-item2}$_4$ i.e.\ $m=4$.   For instance, if $\mc {F}_\mu$ is a class of H{\"o}lder continuous functions with H\"older exponent $\beta > 0$ on $D = d+1$ dimensional Euclidean space,
  and we require     \ref{assm:EA-item2}$_m$ to hold, then the $\epsilon$-entropy is of order $\epsilon^{-D / \beta}$ so we require $m D / 2 \beta < 1$ or $\beta > mD/2$.  When $m=1$ the condition is the standard one and when $m=3$ or $4$ it is more restrictive.
  \end{remark}

  \begin{remark}
    It may be possible to weaken these assumptions.
    The assumption \ref{assm:EA-item2}$_m$ with $m > 1$ arises from certain (degenerate) U- or V-process terms in the analysis.  Analyzing such terms requires these more stringent entropy conditions.  On the other hand, an $m$th order (degenerate) U-process comes with a faster decay to $0$, of order $n^{-m/2}$.  When the class $\mc{F}$ (i.e., $\mc{F}_\mu, \mc{F}_\pi$) does not depend on $n$ this does not help us.  But if we allow a sieve-type approach where the class $\mc{F} \equiv \mc{F}_n$ depends on $n$, then we need $J_1(1, \mc{F}_n)$ to be $O(1)$ but $J_m(1, \mc{F}_n)$ for $m >1$ can be allowed to grow with $n$; if we allow such sieve classes $\mc{F}_n$, then the only entropy required to stay finite/bounded in $n$ is $J_1$, and so in this sense we can recover/require the more classical condition.  At present we have phrased the conditions only in terms of independent-of-$n$ classes.
  \end{remark}

\section{Main results}
\label{sec:main-results}


Now we present the asymptotic distribution of our test statistic $T_n$ under the
null hypothesis.
To metrize weak convergence, we use
the Dudley metric \citep[Chapter 14, Section
2,][]{shorack2000probability} (although any topologically equivalent
metric would work), which is defined as 
\begin{equation}
  \label{eq:dudley}
  d(\mu, \nu) := \sup \Big\{ \int g\,d\mu-\int g\,d\nu:\|g\|_{BL}\le 1 \Big\},  
\end{equation}
where $X$ and $Y$ are random variables with      %
probability distributions/laws $\mu$ and $\nu$, respectively, 
and where
$ \|g\|_{BL}:= \sup_{x\in \RR}|g(x)|+ \sup_{x\ne y}|g(x)-g(y)|/|x-y|$.
For a kernel function $K$, we use $K^{(s)}$ to denote the $s$-times convolution
product of $K$, that is $  K^{(s)}(x) = \int K^{(s-1)}(y)K(x-y)\,dy$,
with $K^{(1)}=K$.
And we let $K^{(s)}_h(x):= K^{(s)}(x/h)$.
Let
\begin{align}
  \label{eq:condition-var}    \sigma^2(a)=\EE\ls\frac{\tau(\bs{L},a)+\{\mu_0(\bs{L},a)-\bar{\mu}(\bs{L},a)\}^2}{\{\bar{\pi}(a|\bs{L})/\bar{\varpi}(a)\}^2/\{\pi_0(a|\bs{L})/\varpi_0(a)\}}\rs
  -\lb\theta_0(a)-\bar{m}(a))\rb^2,
\end{align}
where $\bar{\varpi}(a):=\int\bar{\pi}(a|\bs{l})\,dP(\bs{l})$ and
$\bar{m}(a):=\int \bar{\mu}(\bs{l},a)\,dP(\bs{l})$.
We can now state our main theorem, which gives the limit distribution of our
test statistic under the null hypothesis.
Let $\mc{L}(X)$ denote the probability law of the random variable $X$.

\begin{theorem}
  \label{thm:null-normality}
  Let Assumptions \ref{assm:IA}\ref{assm:IA-item1}--\ref{assm:IA}\ref{assm:IA-item3}, \ref{assm:DA}\ref{assm:DA-item1}--\ref{assm:DA}\ref{assm:DA-item4}, \ref{assm:EA}.\ref{assm:EA-item3}--\ref{assm:EA}.\ref{assm:EA-item5},
  and \ref{assm:EA-item2}$_3$
  hold and let $w(\cdot)$ be a continuously
  differentiable weight function on $\mc A$.
  Let $\sigma^2(\cdot)$ and $T_n$ be as given in 
  \eqref{eq:condition-var}
  and
  \eqref{eq:test-stat}.
  Then  under $H_0$ (from \eqref{eq:11}), we have 
  \begin{align}
    \label{eq:null-limit-dist}
    d\lb \mc{L}( T_n), \mc{L}(N(b_{0h},V))
    \rb\to 0
  \end{align}
  as $n\to \infty$, where
  \begin{align}
    b_{0h}=h^{-1/2}K^{(2)}(0)\int_{\mc{A}}\frac{\sigma^2(a)w(a)}{\varpi_0(a)}\,da,\quad V=2K^{(4)}(0)\int_{\mc{A}}\ls\frac{\sigma^2(a)w(a)}{\varpi_0(a)}\rs^{2} \,da.
  \end{align}
\end{theorem}

\smallskip \noindent The full proof is given in the appendices. 
We provide an outline of the proof here.

\smallskip \noindent {\it Proof outline for Theorem~\ref{thm:null-normality}.} We can decompose the statistic $T_n$  as
\begin{equation}
  \label{eq:6}
  T_n=n\sqrt{h}\int_{\mc{A}} \lp D_1(a)+D_2(a)+D_3 \rp^2 w(a)\,da, 
\end{equation}
where
\begin{equation*}
  D_1(a):=\widehat{\theta}_h(a)-\tilde{\theta}_h(a),
  \quad
  D_2(a):=\tilde{\theta}_h(a)-\PP_n\bar{\xi}(\bs Z), 
  \quad
  D_3:=\PP_n\bar{\xi}(\bs Z)-\PP_n\widehat{\xi}(\bs Z),  
\end{equation*}
and where $\tilde{\theta}_h(a)$ is the local linear estimator regressing the oracle doubly robust mappings $\bar{\xi}_i:=\xi(\bs{Z}_i;\bar{\pi},\bar{\mu})$ on $A_i$.  Expanding the square leads to 6 terms to be analyzed, so we break the proof up into 6 main steps corresponding to each of those terms.  In step 1, we verify that $ n\sqrt{h} \int_{\mc A} D_2(a)^2w(a) \, da$ is distributed approximately as the given $N(b_{0h}, V)$ limit distribution in the theorem.
This follows essentially from \citet{alcala1999goodness},
which extends the results of \citet{Hardle:1993ih} to allow local polynomial estimators.
In steps 2--6, we show that the remainder terms 
(which are %
the terms  %
$n\sqrt{h} \int_{\mc A} D_3^2 w(a)\,da$, $n\sqrt{h}\int_{\mc A} D_1(a)^2w(a) \, da$, $n\sqrt{h} D_3 \int_{\mc A} D_2(a)w(a) \, da$,
$n \sqrt{h} \int_{\mc A} D_1(a) D_3 w(a)\, da$, 
and $n \sqrt{h} \int_{\mc A} D_1(a) D_2(a) w(a)\, da$) are all $o_p(1)$ as $n \to \infty$.  We now will discuss those remainder terms in slightly more detail, basically in parallel with steps 2--6.

Let $\eta := (\pi, \mu)$
(and $\wh \eta := (\wh \pi, \wh \mu)$, $\bar \eta := (\bar \pi, \bar \mu)$).
Both $D_1(a)$ and $D_3$ involve summations over $\wh{\xi}(\bs Z_i; \wh \eta) - \xi(\bs Z_i; \bar \eta)$, and in analyzing the remainder terms, we
break these summations into
sums over $\wh{\xi}(\bs Z_i; \wh \eta) - \xi(\bs Z_i; \wh \eta)$ and $ \xi(\bs Z_i; \wh \eta) - \xi(\bs Z_i ; \bar \eta)$.
(Recall the definition of $\wh \xi$ given in \eqref{eq:pseudo-outcome}, in which $dP$ is replaced by $d\PP_n$.)
The sums over
$\wh{\xi}(\bs Z_i; \wh \eta) - \xi(\bs Z_i; \wh \eta)$
yield terms that can be written as degenerate V-statistics (or, rather, because of the presence of the random $\wh \eta$, terms whose size is governed by degenerate V-processes).  We are are  able to conclude that these are of very small order of magnitude,
but unfortunately the empirical process tools (that we are aware of) for analyzing
them requires the imposition of stronger entropy conditions than the normal Donsker-type condition.   (Because of the presence of up to order $3$ V-processes, we require $J_3(1, \mc{F}, L_2) < \infty$ rather than the weaker, more standard Donsker condition, $J_1(1, \mc{F}, L_2) < \infty$ (for $\mc{F}$ equal to both $\mc{F}_\mu,\mc{F}_\pi$).)

For instance, in Step 2 we write the term $D_3$ as
\begin{equation*}
  \begin{split}  \MoveEqLeft\PP_n\{\xi(\bs{Z};\bar{\pi},\bar{\mu})-\widehat{\xi}(\bs{Z};\widehat{\pi},\widehat{\mu})\}\\&=\PP_n\{\xi(\bs{Z};\bar{\pi},\bar{\mu})-\xi(\bs{Z};\widehat{\pi},\widehat{\mu})\}+\PP_n\{\xi(\bs{Z};\widehat{\pi},\widehat{\mu})-\widehat{\xi}(\bs{Z};\widehat{\pi},\widehat{\mu})\}.
  \end{split}
\end{equation*}
The first summand  (of the right hand side above) is further decomposed into two terms of types that commonly arise in causal inference or semiparametric problems one an empirical process one and one a `second order remainder'; the former is shown to be negligible by an empirical process asymptotic equicontinuity argument and the latter is small by assumptions on $r_n^\infty s_n^\infty$.  The second summand in the display above can be written as a degenerate order 2 V-process.
We provide an introduction to and discussion of U- and V-processes in Section~\ref{sec:v-process-results}
Under Assumption
\ref{assm:EA-item2}$_2$ (implied by Assumption \ref{assm:EA-item2}$_3$) we show  that the V-process is negligible.  Note that a degenerate order $m$ U- or V-statistic is usually of order $n^{-m/2}$ 
\citep[Chapter 12]{vanderVaart:1998dr}.

In Step 3, we write $D_1(a)$
as $  D_1(a) =d_{1,1}(a)+d_{1,2}(a)$,
where
\begin{equation*}
  \begin{split}
    d_{1,1}(a) & := \PP_n \ls W_{ha}(A)
    \lp  \wh \xi( \bs Z; \wh \pi, \wh \mu ) - \xi(\bs Z; \widehat{\pi},
    \widehat{\mu}) \rp \rs, \\
    d_{1,2}(a)&:=  \PP_n \ls W_{ha}(A)
    \lp   \xi( \bs Z; \wh \pi, \wh \mu ) - \xi(\bs Z; \bar{\pi},
    \bar{\mu}) \rp \rs,
  \end{split}
\end{equation*}
and $W_{ha}(\cdot)$ are ``equivalent'' kernels for the local polynomial estimator (see the proof for definitions and references or see \cite{Fan:1996tk}, page 63).  The added difficulty here from Step 2 is that these terms are local i.e.\ they depend on $h$, but thematically the decomposition works the same as in Step 2.  The term $d_{1,2}(a)$ is handled by an asymptotic equicontinuity argument and assumptions on nuisance estimators, and the term $d_{1,1}(a)$ by a (nonasymptotic) V-process maximal inequality (see Proposition~\ref{prop:1}) for an order 2 V-process.

In Step 4, we consider $\int_{\mc A} D_2(a) D_3 w(a) \, da$; here $D_3$ can be factored out and handled by the result of Step 2, and the remaining integral of $D_2(a)$ can be handled in an elementary way by Taylor expansion and the Central Limit Theorem.
In Step 5, we consider $\int_{\mc A} D_1(a) D_3 w(a) \, da$, whose negligibility follows immediately from the analysis in Steps 2 and 3 and the Cauchy-Schwarz inequality.

Finally, in Step 6 an order 3 V-processes arises (and so Assumption \ref{assm:EA-item2}$_3$ is needed) in the analysis of  $\int_{\mc{A}}D_1(a)D_2(a)w(a)\,da$.  This can be essentially simplified into (a sum of) two main terms,
\begin{equation}
  \label{eq:1}
  \begin{split}
    &  \int    \frac{1}{\varpi_0^2(a)} \PP_n \ls K_{ha}(A)
    \lb   \xi( \bs Z; \wh \pi, \wh \mu ) -  \xi(\bs Z;  \bar \pi,
    \bar \mu) \rb\rs \frac{1}{n} \sum_{i=1}^n
    K_h\lp A_i-a\rp
    \bar\epsilon_i w(a)\,da,  \\
    &   \int    \frac{1}{\varpi_0^2(a)} \PP_n \ls K_{ha}(A)
    \lb  \wh \xi( \bs Z; \wh \pi, \wh \mu ) -  \xi(\bs Z;  \wh \pi,
    \wh \mu) \rb\rs\frac{1}{n}\sum_{i=1}^n
    K_h\lp A_i-a\rp
    \bar\epsilon_i w(a)\,da,    
  \end{split}
\end{equation}
where $\overline{\epsilon}_i:=\xi(\bs{Z}_i;\overline{\pi},\overline{\mu})-\PP\xi(\bs{Z};\overline{\pi},\overline{\mu})$.  The first term in \eqref{eq:1} is further decomposed, with $\PP_n$ written as $(\PP_n - \PP) + \PP$;
usually, the $(\PP_n - \PP)$ term yields an empirical process and the $\PP$ term yields a second order remainder.  Here the empirical process piece can be handled by previous arguments; the `$\PP$' term can not be treated as a second order remainder because taking absolute values (as would be commonly done) breaks the mean zero structure of the term and does not yield the correct size (because of the integral over $\mc{A}$, the $ n^{-1} \sum_{i=1}^n
K_h\lp A_i-a\rp
\bar\epsilon_i w(a)$ term and $\PP$ term  cannot be analyzed separately).   Rather, the term is handled by an empirical process (maximal inequality) argument
(in Lemma~\ref{lem:step5-decompose-term2}).  
Finally, the second term in \eqref{eq:1} is the order three V-process, requiring $J_3(1, \cdot, L_2) < \infty$ in order to apply a maximal inequality.\footnote{Again, in the analysis of this [and other] term[s] we cannot analyze the multiplicands' orders of magnitudes separately, because taking absolute values breaks the mean zero structure; i.e., to explain, let $C_i, D_i$ be generic random variables and then although we can bound $|\sum_i C_i D_i| \le ( \max_i | C_i|) \sum_i |D_i|$, unfortunately then $\sum_i |D_i|$ (with absolute values) is not a sum of mean zero variables, even if the $\{ D_i \}$ are mean zero;
  getting the right order of magnitude requires treating the multiplicands simultaneously as a V-process.}
That completes our proof outline; full details of the proof are given in Section~\ref{sec:proof-main-theor-refm}
\hfill $\square$

\smallskip

Next we state the consistency of our test under alternatives as follows.
(The $\delta_n$ sequence can in particular be $(n \sqrt{h})^{1/2}$.) 

\begin{theorem}
  \label{thm:alternative}
  Let Assumptions \ref{assm:IA}\ref{assm:IA-item1}--\ref{assm:IA}\ref{assm:IA-item3}, \ref{assm:DA}\ref{assm:DA-item1}--\ref{assm:DA}\ref{assm:DA-item4}, \ref{assm:EA}\ref{assm:EA-item3}--\ref{assm:EA}\ref{assm:EA-item5} hold and let $w(\cdot)$ be a continuously
  differentiable weight function on $\mc A$.
  Let $\sigma^2(\cdot)$ and $T_n$ be as given in \eqref{eq:condition-var} and \eqref{eq:test-stat}.
  Then under the
  alternative $\theta_0(a)=c_0+\delta_n(n\sqrt{h})^{-1/2}g(a)$, where
  $c_0=\PP\xi(\bs{Z};\pi_0,\mu_0)$, where 
  $g(\cdot)$ is not the constant $0$, and where $\int
  g(a)\varpi(a)w(a)\,da=0$. Moreover, $\delta_n$ is a sequence converging to
  $\infty$ such that $\lim_{n\to \infty}n^{1/40}/\delta_n=0$. Then  we have
  \begin{align*}
    P(T_n > z_{n,1-\alpha})\to 1
  \end{align*}
  as $n\to \infty$, where we use $z_{n,1-\alpha}$ to denote the upper $\alpha$ quantile of the $N(b_{0h}, V)$ distribution in \eqref{eq:null-limit-dist}.
\end{theorem}

\noindent
The proof is given in the appendices. 
The condition that $\int g(a)\varpi(a)w(a)\,da=0$ is not a
substantive restriction, it is just so that $c_0$ is `identifiable' in a sense.

Finally,
we show that the bootstrap distribution of the statistic can be used to
approximate $T_n$'s unknown distribution (under the null).  For our
proof to hold, we require an extra entropy condition to accommodate a fourth
order V-process that appears in the analysis of the wild bootstrap.  Recall
the definition of $J_4$ from \eqref{eq:uniform-entropy-2-integral} and the definition of the Dudley metric for weak convergence in
\eqref{eq:dudley}.
Let $\mc{L}^*(X) := \mc{L}(X | \bs Z_1, \ldots \bs Z_n)$ denote the  conditional law of a random variable $X$.

A variety of bootstrap definitions are possible (and are included in the simulations) and discussed in Subsection~\ref{sec:proposed-method}. 
For the following theorem, we let $\wh \epsilon_i := \widehat{\xi}(\bs{Z}_i;\widehat{\pi},\widehat{\mu})-\sum_{i=1}^n\widehat{\xi}(\bs{Z}_i;\widehat{\pi},\widehat{\mu})/n$ be centered at the null estimate, and we take $\epsilon_i^*$ to be the Rademacher choice so equal to $\pm \wh \epsilon_i$ with probability $1/2$ each.  Then we proceed as discussed in Subsection~\ref{sec:proposed-method}, with $(\xi_i^{\ast}=\PP_n\wh\xi+\epsilon^{\ast}_i,A_i)$ as bootstrap observations and defining $T_n^*$ by \eqref{eq:test-stat} but using the bootstrap observations.

\begin{theorem}
  \label{thm:bootstrap-consistency-null}
  Let the assumptions of
  Theorem~\ref{thm:null-normality} hold.
  Let $T_n^*$ be the bootstrap test statistic  defined in the paragraph preceding this theorem. 
  Then
  $$d( \mc{L}^*(T_n^*), \mc{L}(N(b_h, V)) ) \to_p 0$$
  as $n \to \infty$.
\end{theorem}

\noindent Next we consider the bootstrap where
$\wh \epsilon_i := \widehat{\xi}(\bs{Z}_i;\widehat{\pi},\widehat{\mu}) - \wh \theta_h(A_i)$ (and the rest of the procedure is as described in the paragraph preceding the previous theorem).  We study this case in the next theorem, which requires an extra entropy condition ($J_4 < \infty$).

\begin{mylongform}
  \begin{longform}
    This centering is probably needed to say anything under the alternative hyptohesis.
  \end{longform}
\end{mylongform}

\begin{theorem}
  \label{thm:bootstrap-consistency-null-2}
  Let the assumptions of
  Theorem~\ref{thm:null-normality} hold.  Further, assume that $J_4(1, \mc{F}, L_2) < \infty$ for $\mc{F} = \mc{F}_\mu$ and for $\mc{F} = \mc{F}_\pi$.  Let $T_n^*$ be the bootstrap test statistic  defined in the paragraph preceding this theorem. 
  Then
  $$d( \mc{L}^*(T_n^*), \mc{L}(N(b_h, V)) ) \to_p 0$$
  as $n \to \infty$.
\end{theorem}

The  structure of the proofs is
analogous to that of
Theorem~\ref{thm:null-normality}, although the
calculations are somewhat more intricate,
and lead to a fourth order V-process which requires the finiteness of $J_4$ as stated in the theorem.
The proof details are given in the appendices. 
Both bootstraps are studied in our simulations and seem to perform similarly.

\section{Simulation studies}
\label{sec:simulation}

\subsection{Simulation for testing constant average treatment effect}
We use simulation to assess the performance of our proposed test in terms of
both type I error probability and power. We consider  two data
generating models, one with a binary reponse and which is defined similarly as in
\citet{Kennedy:2017cq}, and another one with a continuous response. In more detail,
they are specified as follows. 

{\bf Model 1:} we simulate the covariates from independent standard normal distributions,
$\bs{L}=(L_1,\ldots,L_4)^T\sim N(0,\bs{I}_4)$, and   simulate the treatment
level from a Beta distribution,
\begin{align*}
  &(A/20)|\bs{L}\sim \text{Beta}(\lambda(\bs{L}),1-\lambda(\bs{L})),\\
  &\text{logit } \lambda(\bs{L})=-0.8+0.1L_1+0.1L_2-0.1L_3+0.2L_4,
\end{align*}
and finally the binary outcome is simulated as $Y|\bs{L},A\sim
\text{Bernoulli}(\mu(\bs{L},A))$, where
$ \text{logit }\mu(\bs{L},A)=1+(0.2,0.2,0.3,-0.1)\bs{L}
+\delta A(0.1-0.1L_1+0.1L_3-0.13^2A^2)$.

{\bf Model 2:} the covariates are simulated the same as in  Model 1. We
simulate the treatment level from a Beta distribution,
\begin{align*}
  &(A/5)|\bs{L}\sim \text{Beta}(\lambda(\bs{L}),1-\lambda(\bs{L})),\\
  &\text{logit } \lambda(\bs{L})=0.1L_1+0.1L_2-0.1L_3+0.2L_4,
\end{align*}
and we simulate the continuous
response from conditional normal distributions as
$Y|\bs{L},A\sim N(\mu(\bs{L},A),0.5^2)$, where
\begin{equation*}
  \mu(\bs{L},A)= 
  (0.2,0.2,0.3,-0.1)\bs{L}+A(-0.1L_1+0.1L_3)
     +\delta\exp\lb-\frac{(A-2.5)^2}{(1/2)^2}\rb. 
\end{equation*}
In both models, we have a parameter $\delta$ that controls the distance between
the true treatment effect and the null hypothesis, i.e., treatment effect
is constant, with $\delta=0$ indicating no treatment effect in both
models. In Model 1,  $\delta=0$ yields a constant average treatment effect
and is a `strong null' meaning all individuals have the same treated and untreated outcomes;
in Model 2 when
$\delta=0$ the `weak null' holds meaning conditional average treatment effects are
nonconstant but the average treatment effect is constant.
Specifically, for Model 1, we let
$\delta\in\{0,0.002,0.004,0.006,0.008,0.01\}$, and we let 
$\delta\in\{0,0.1,0.2,0.3,0,4,0.5\}$ for Model 2. We plot the treatment effect
curve with $\delta=0.01$ for Model 1  and the treatment effect
curve with $\delta=0.5$ for Model 2 in Figure~\ref{fig:simulation-treatment-curve}.








\begin{figure}
  \centering
  \includegraphics[width=0.8\textwidth]{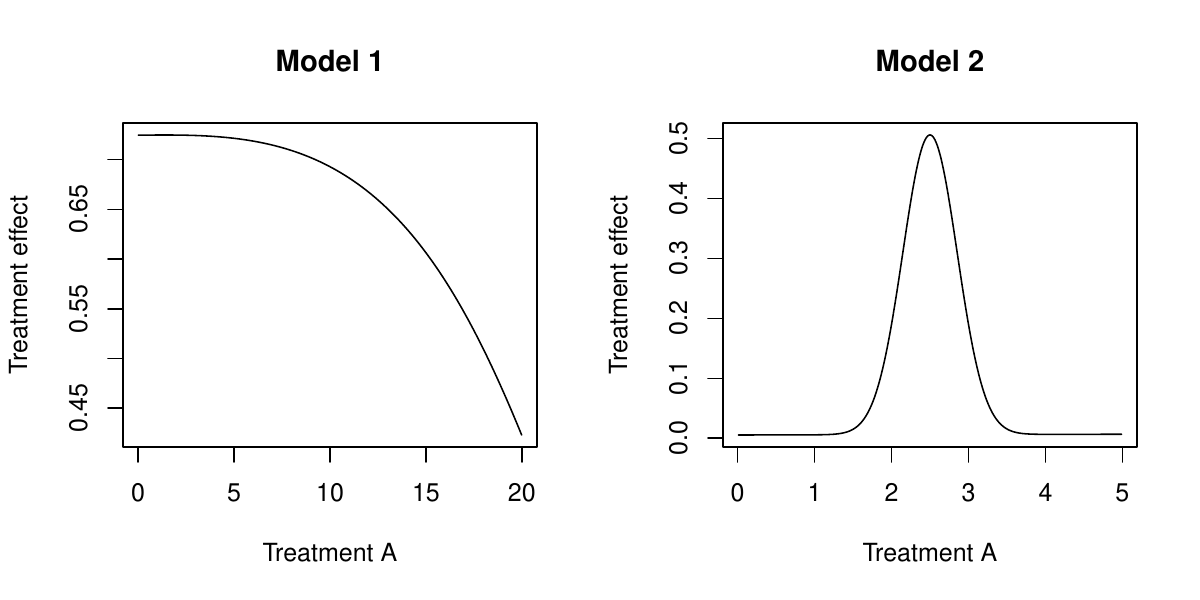}
  \caption{Treatment effect curves in Model 1 and Model 2. Left: Model 1 with
    $\delta=0.01$. Right: Model 2 with $\delta=0.5$.}
  \label{fig:simulation-treatment-curve}
\end{figure}
For each data generating
model, we test the performance of our method under  4 scenarios: (1) $\pi$ is correctly
specified with a parametric model, $\mu$ is incorrectly specified with a parametric model; (2) $\pi$ is incorrectly
specified with a parametric model, $\mu$ is correctly specified with a parametric model; (3)
both $\pi$ and $\mu$ are correctly specified with  a parametric model; (4) both
$\pi$ and $\mu$ are estimated with Super Learners \citep{van2007super}. In the
first two scenarios, the incorrect parametric models are constructed in the same
fashion as in  \citet{kang2007demystifying}. The
first three scenarios are used to test double robustness of our method and last
one to show the empirical performance of  our method when we  use flexible
machine learning models to estimate the nuisance functions. After we calculate the
pseudo-outcomes, we use the rule of thumb for bandwidth selection \citep{Fan:1996tk} for the local
linear estimator.  We compare the performance of
our method with \citet{westling2020nonparametric} in the first three scenarios;
with \citet{westling2020nonparametric} and a discretized version of
TMLE \citep{JSSv051i13} with treatment dichotomized at the middle
point in the last scenario. In our method, we choose the weight function $w(a)\equiv 1$. We
implemented all three versions ($L_1,L_2$ and $L_{\infty}$) of the methods in
\citet{westling2020nonparametric} for comparison.   Rejection probabilities are
estimated with 1000  independent replications of
simulation. Finally,  we consider sample sizes in $\{500,1000,2000\}$.

Figures~\ref{fig:model1-parametric}  and
\ref{fig:model1-nonparametric}
show
the results for Model 1.   We can see when at least one of the nuisance functions is
correctly estimated, our method and Westling's methods performed similarly in
terms of both type I error probability and power. When both nuisance functions
are estimated with Super Learners,  Westling's methods have slightly larger
power than our method 
but also have
slightly larger type I error probabilities. Note that in this case, the
discretized version of TMLE outperformed both our method and Westling's method
in terms of type I error probability and power. The reason may be the shape of the
treatment effect curve in Model 1 is somewhat simple and monotone, and there isn't
much information loss if we dichotomize the continuous treatment to form a
simpler testing problem.
We also note that apparently 
\citet{westling2020nonparametric}'s method is doubly robust
on the specific data generating model we use here.

Figures~\ref{fig:model2-parametric}  and  \ref{fig:model2-nonparametric} show
the results for Model 2 where we have a slightly complicated and non-monotone treatment effect
curve as in  Figure~\ref{fig:simulation-treatment-curve}.   We observe that in the first
three scenarios our methods outperform Westling's methods in terms of power
in all cases. Our method does better even with a small sample size and a weak
deviation from the null model, compared with Westling's method.  Similar
observations hold when we use Super Learner to estimate both nuisance
functions. Another  observation worth noting is that in Model 2 the
discretized TMLE fails to detect any deviation from the null model since the
treatment effect curve is symmetric, which provides an example in which discretizing
a continuous treatment and applying a binary test  may lead to a completely incorrect conclusion.



\begin{figure}
  \centering
  \includegraphics[width=0.75\textwidth]{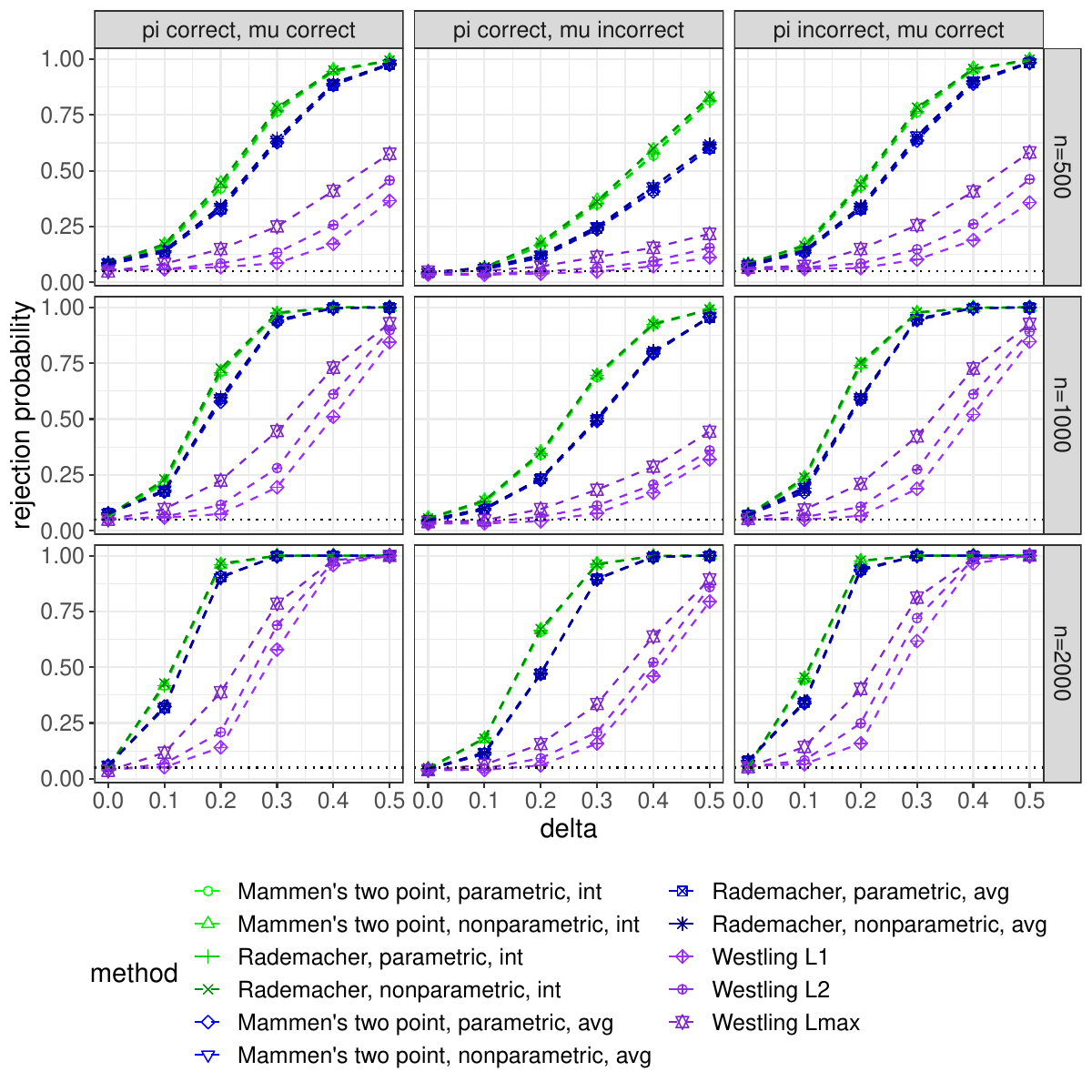}  
  \caption{Simulation result for Model 2 with $\pi$ and $\mu$ estimated from
    parametric models.}
  \label{fig:model2-parametric}
\end{figure}

\begin{figure}
  \centering
  \includegraphics[width=\textwidth]{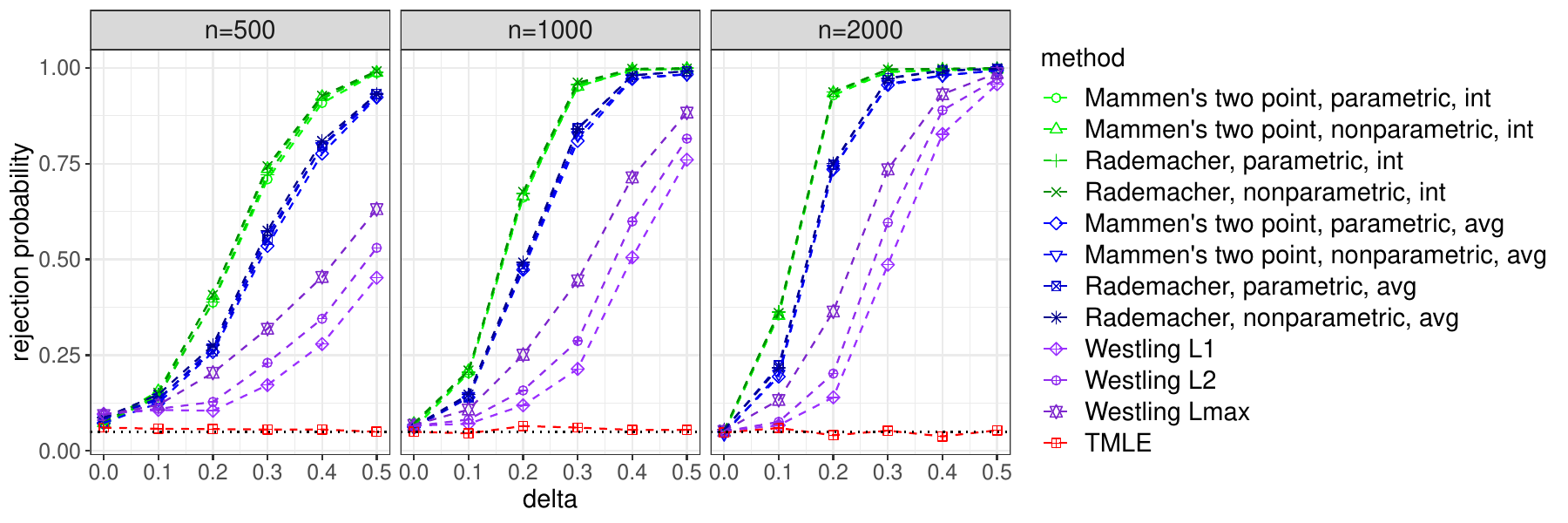}  
  \caption{Simulation result for Model 2 with $\pi$ and $\mu$ estimated from
    nonparametric models.}
  \label{fig:model2-nonparametric}
\end{figure}

\subsection{Cross-fitted test procedure}

We also consider simulations to analyze how dimensionality of the confounders $\bs{L}$ can affect the performance of our test and how cross-fitting could be applied to improve finite sample performance under high dimensional settings. To save space, we defer this part of the content to Section~\ref{sec:cross-fitting}
The detailed description of the test procedure with cross-fitting can be found in Section~\ref{subsec:cross-fitting-test}. We conduct simulation studies to compare our main proposed test (without cross-fitting) and the cross-fitted test under low dimensional data and under high dimensional data, respectively, in Section~\ref{subsec:cross-fitting-simulation}.  The following is a summary of the results from the simulations.  When the dimensionality of the covariates is small, both non-cross-fitted and cross-fitted tests can achieve the desired type I error probability but the cross-fitted version tends to have lower power; when the dimensionality of the covariates is large, only the cross-fitted version maintains the desired type I error probability.

\section{Analysis of data on nursing hours and hospital performance}
\label{sec:real-analysis}

In this section we apply our test to a real data problem. In
\citet{Kennedy:2017cq} and \citet{McHugh:2013gn}, the authors were interested in
whether  nurse staffing (measured in nurse hours per patient day)
affected a hospital's risk of excess readmission penalty after adjusting for
hospital characteristics (for more detail of the data and related background of
the problem, see \citet{McHugh:2013gn}). \citet{Kennedy:2017cq} 
proposed a doubly robust procedure to estimate the probability of readmission
penalty against average adjusted nursing hours per patient day, and provided
pointwise confidence intervals for the estimated treatment curve. However, their
method and analysis
did not answer the question of whether  nurse staffing significantly
affects the probability of excess readmission penalty after adjusting for
hospital characteristics. We apply our method to test the null hypothesis: nurse staffing 
does not affect  hospital's risk of excess readmission penalty after adjusting for
hospital characteristics, with  updated data from the year 2018. As a brief
summary of the data, the outcome $Y$ indicates whether the hospital was
penalized due to excess readmissions and are calculated by the Center for
Medicare \& Medicaid Services (\url{https://www.cms.gov}). The treatment $A$ measures nurse
staffing hours and we calculate it as the ratio of registered nurse hours to
inpatient days, which is slightly different from \citet{Kennedy:2017cq} and
\citet{McHugh:2013gn}, because we don't have access to the hospitals' financial data
and thus are not able to calculate adjusted inpatient days. The covariates
$\bs{L}$ include the following nine variables:
the number of beds,
the teaching intensity,
an indicator for not-for-profit status,
an indicator for whether the location is urban or rural,
the proportion of patients on Medicaid,
the average patient socioeconomic status,
a measure of market competition,
an indicator for whether the hospital has a skilled nursing facility (because our measure of nurse staffing hours $A$ will unfortunately include hours worked in such a skilled nursing facility),
and whether open heart or organ transplant surgery is performed (which serves as a measurement of whether the hospital is high technology).
We omitted patient race proportions and
operating margin from the analysis (present in \citet{Kennedy:2017cq} and
\citet{McHugh:2013gn}) because we don't have  access to those features.
Figure~\ref{fig:observed-data}
shows  an unadjusted loess fit of the readmission
penalty as a function of the average nursing hours and the loess fits of the covariates
against  the average nursing hours. The curves are not identical to those in
\citet{Kennedy:2017cq} since we've used updated data from 2018, but we observe
generally similar patterns and  nurse staffing hours is correlated with many
hospital characteristics. In the analysis, we use Super Learner
\citep{van2007super} with the same implementation as in \citet{Kennedy:2017cq}
to estimate $\pi_0$ and $\mu_0$. We truncate $\wh\pi$ to be 0.01 if it fell
below that value. The rule of thumb \cite{Fan:1996tk} is
applied for bandwidth selection as in Section~\ref{sec:simulation}.
Since
our test statistic is based on the integrated distance between the nonparametric
fit of the treatment effect curve and the parametric fit of the treatment effect curve
under the null hypothesis, a  byproduct of the test is the estimated
treatment effect curve. We plot the estimated treatment effect curve of average
nurse staffing 
in Figure~\ref{fig:aha-treatment-curve}
(the solid red curve).


\begin{figure}
  \centering
  \includegraphics[width=0.75\textwidth]{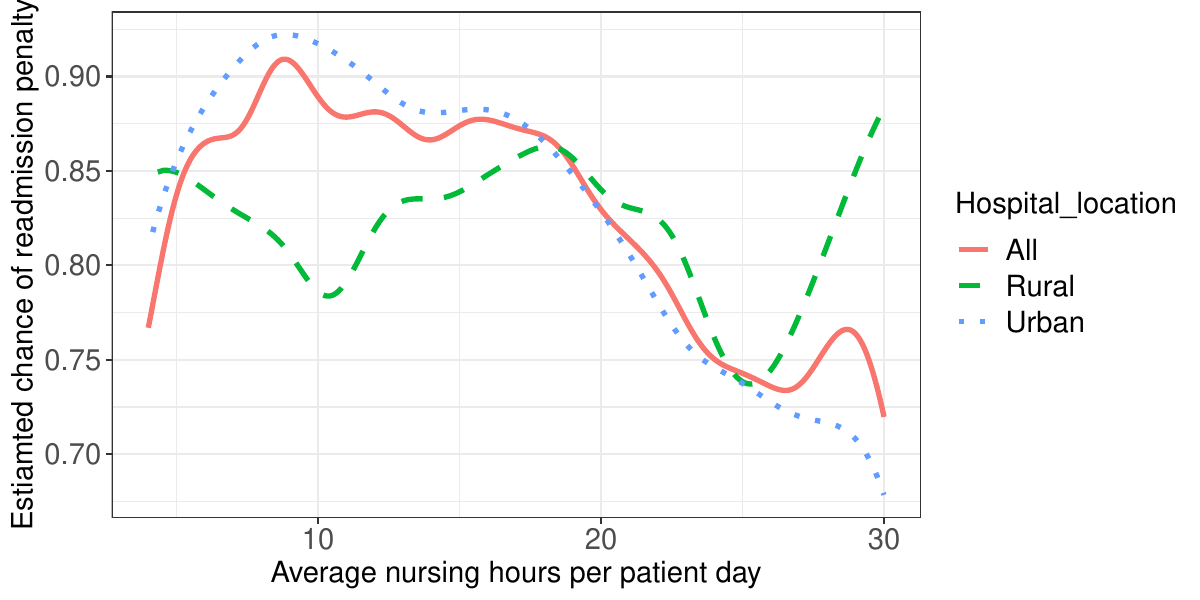}
  \caption{Estimated treatment effect of average nursing hours on probability of
    readmission penalty.   \label{fig:aha-treatment-curve}}
\end{figure}

We apply our test, Westling's test, and the discretized version of TMLE to this
data set. All versions of our methods and Westling's methods have p-values of 0.
(Exact zeros  are due to the fact that we use simulated reference distributions.)
The discretized TMLE reports a p-value of 0.0017. So all the  tests
suggest strong statistical evidence against the null hypothesis of 
constant treatment effect, meaning average nursing hours does have a significant
causal impact on a hospital's chance of being penalized for excess
readmissions. 
This interpretation requires that we have included all important confounders in our analysis.
If we have not, our test result is interpreted as being based on a partially adjusted estimate of association (rather than the treatment effect curve).  

Finally, we test whether 
an indicator for whether a hospital is in a rural or in an urban setting is a treatment effect modifier. 
We present the
estimated 
conditional treatment effect curves for rural hospitals and for the urban hospitals in
Figure~\ref{fig:aha-treatment-curve} (dashed green for the rural hospitals and
dotted blue for the urban hospitals).
We observe that  for the  hospitals in  urban areas, the pattern of the effect
curve has a shape that is close to  concave and is
similar to the pattern 
in the overall average treatment effect
curve: after average nursing hours exceeds
10, increasing the average nursing hour results in a decrease in the probability of
the readmission penalty. The increasing trend of the curve up to 8 average nursing hours seems to be
counter-intuitive, but
it turns out that there are not many hospitals in that  range of the data and thus
the left tail behavior is likely an artifact due to low sample size in that region. On the other hand, 
the effect curve for rural hospitals is wavy and does not suggest a clear pattern.

We first apply our main  test separately to each of the two groups of hospitals to see
whether the two individual treatment effect curves are  constant or not. We obtain
a p-value of 0.28 for the group of rural hospitals and a p-value of approximately 0 for the  urban hospitals.
This analysis suggests that  the conditional treatment curve for the rural hospitals is not
significantly different from constant, so the wavy pattern we see in the
estimated curve is likely due to randomness. Next we apply the extended test
procedure and obtain
a p-value of 0.007, which indicates a significant difference between
the two conditional treatment curves.
Again, these interpretations require that we have included all important confounders in our analysis.
It is somewhat surprising that in this dataset the rural hospitals show no significant effect of nurse staffing on readmission penalty.  It is possible that hospital occupancy rates, case mix, financial stability and differing abilities to recruit and retain nurses are important for understanding the effect of nurse staffing on hospital performance, either as confounders or as treatment effect modifiers.



\section*{Acknowledgements}
Charles R.\ Doss is partially supported by
NSF grant DMS-1712664 and 
NSF grant DMS-1712706.

\appendix

\section{Empirical process lemmas}



In this section, we first discuss the concept of stochastic
equicontinuity and provide two lemmas that will be used in the proof
of our main result. Let $\GG_n=\sqrt{n}(\PP_n-\PP)$. Then a sequence of
empirical processes $\{\GG_nV_n(f): f\in\mc{F}\}$ indexed by elements
$f$ from a metric space $\mc{F}$ equipped with semimetric $\rho$ is
stochastically equicontinuous \citep{vanderVaart:1996tf} if for every
$\epsilon>0$ and $\zeta>0$ there exists a $\delta>0$ such that
\begin{align*}
  \limsup_{n\to \infty}P\lp\sup_{\rho(f_1,f_2)<\delta}|\GG_nV_n(f_1)-\GG_nV_n(f_2)|>\epsilon\rp<\zeta.
\end{align*}
We also rely on a standard measurability condition, that of ``$P$-measurability'' \cite[Definition 2.3.3,][]{vanderVaart:1996tf},
which is generally satisfied and so we do not give the definition here. 

\begin{lemma}[Theorem 2.5.2, \cite{vanderVaart:1996tf}]
  \label{lem:identity-equicont}
  Consider the sequence of processes $\{\mathbb{G}_nV_n(\cdot):n\ge 1\}$ where $V_n$
  is the identity mapping, i.e.,
  \begin{align*}
    V_n(f)=f,
  \end{align*}
  here $f\in \mc{F}$ with envelope $F(\bs{Z})=sup_{f\in\mc{F}}|f(\bs{Z})|$. Assume
  $F$ is uniformly bounded, i.e., $\|F\|_{\mc{Z}}\le F_{\max}$ for some
  $F_{\max}<\infty$, and $\mc{F}$ has a finite uniform entropy
  integral, i.e., $J_1(\delta,\mc{F},L_2)<\infty$, and
  $J_1(\delta,\mc{F},L_2)$ is as defined in
  \eqref{eq:uniform-entropy-2-integral}
  Let the class
  $\mc{F}_{\delta}$ and $\mc{F}_{\infty}^2$ be $P$-measurable
  for every $\delta>0$, where  $\mc{F}_{\delta}=\{f_1-f_2: f_1,f_2\in \mc{F},
  \|f_1-f_2\|_2<\delta\}$ and $\mc{F}^2_{\delta}=\{(f_1-f_2)^2: f_1,f_2\in \mc{F},
  \|f_1-f_2\|_2<\delta\}$.
  Then $\{\mathbb{G}_nV_{n}(\cdot):n\ge 1\}$ is stochastically equicontinuous.
\end{lemma}

\begin{lemma}
  \label{lem:int-equicont}
  Consider the sequence of processes $\{\mathbb{G}_nV_n(\cdot):n\ge 1\}$ with
  \begin{align*}
    V_n(f)=\int_{\mc{A}}f(\bs{L},t)\,dt,
  \end{align*}
where $f\in \mc{F}$ with envelope $F$ as in Lemma~\ref{lem:identity-equicont}. Assume
  $F$ is uniformly bounded  and $\mc{F}$ has a finite uniform entropy  integral.
  Then $\{\mathbb{G}_nV_{n}(\cdot):n\ge 1\}$ is stochastically equicontinuous. 
\end{lemma}
\begin{proof}
  We check conditions (1)-(3) of Theorem 2.11.1 from  \cite{vanderVaart:1996tf}.
  For the Lindeberg condition (1), note
  \begin{align*}
    \EE\lb\|V_n\|_{\mc{F}}^2I(\|V_n\|_{\mc{F}}>\epsilon\sqrt{n})\rb\le \EE&\lb
                                                                            \lp\int_{\mc{A}}
                                                                            F(\bs{Z})\,dt\rp^2I\lp\int_{\mc{A}}
                                                                            F(\bs{Z})\,dt>\epsilon\sqrt{n}\rp\rb\\
                                                                          &\le
                                                                            \ls\lambda(\mc{A})F_{\max}\rs^2I(\lambda(\mc{A})F_{\max}>\epsilon\sqrt{n})\to 0,
  \end{align*}
  as $n\to \infty$ for any $\epsilon>0$, since by
  Assumption~\ref{assm:DA}\ref{assm:DA-item1}
  $\mc{A}$
  is compact. Here
  $\lambda(\cdot)$ is the Lebesgue measure on $\RR$.

  For the Lindeberg condition (2), 
  \begin{align*}
    \sup_{\rho(f_1,f_2)<\delta_n}\EE\lb V_n(f_1)-V_n(f_2)\rb^2&=
                                                                \sup_{\rho(f_1,f_2)<\delta_n}\EE\lb
                                                                \int_{\mc{A}}f_1(\bs{L},t)-f_2(\bs{L},t)\,dt\rb^2\\
                                                              &\le \lb 2\lambda(\mc{A})\delta_n\rb^2\to 0,
  \end{align*}
  for any $\delta\to 0$. So condition (2) is also satisfied.

  For the complexity condition (3), again we check that the process is measure-like. Note
  \begin{align*}
    \frac{1}{n}\lb
    V_n(f_1)-V_n(f_2)\rb^2&=\frac{1}{n}\ls\int_{\mc{A}}f_1(\bs{L},t)-f_2(\bs{L},t)\,dt\rs^2\\    
                          &\le\frac{1}{n}\int_{\mc{A}}\lb f_1(\bs{L},t)-f_2(\bs{L},t)\rb^2\,dt,
  \end{align*}
  by Jensen's inequality. So $V_n(f)$ is measure-like with
  $\nu_{ni}=\frac{1}{n}\delta_{\bs{L}_i}\times \lambda$. That completes the proof.
\end{proof}

Recall that  $N(\epsilon,\mc{F},\|\cdot\|)$ denotes the $\epsilon$-covering number of $\mc{F}$, i.e.,
the minimal number of $\epsilon$-balls (with distance defined by $\|\cdot\|$)
needed to cover $\mc{F}$.
\begin{lemma}
  \label{lem:integral-decreases-entropy}
  Let $\mc F$ be a class of functions $f(x,y)$ on $\mc X \times \mc Y$ with finite uniform entropy.  Let $\PP$ be a measure on $\mc X \times \mc Y$.  Then the class of functions $\mc G:= \{ g(y) = \PP f( X, y) : f \in \mc{F} \}$ has
  $\sup_Q N(\epsilon, \mc{G}, L_2(Q) )
  \le \sup_Q N(\epsilon, \mc{F}, L_2(Q))$, for all $\epsilon  >
  0$. 
\end{lemma}
\begin{proof}
  By Jensen's inequality, for a measure $Q$ on $\mc Y$,
  $$\int g^2(y) dQ(y) = \int \ls \int f(x,y) d\PP(x)\rs^2 dQ(y)
  \le \iint f^2(x,y) d\PP(x) dQ(y),$$ and $d\PP(x) dQ(y)$ is a measure on $\mc X \times \mc Y$.  Thus, with $g_i(y) := \PP f_i(X,y)$, $i=1,2$, we have for a measure $Q$, $\int (g_1-g_2)(y)^2 dQ(y) \le \int (f_1-f_2)(x,y)^2 d \wt Q(x,y)$ (with $\wt Q(x,y) := \PP(x)Q(y)$), and so the conclusion follows.
\end{proof}

The following two results are slight modifications of  Theorem 3 from
\cite{andrews1994empirical}, 
so that we can apply them to $J_2$ as well as to $J \equiv J_1$ (as defined in \eqref{eq:uniform-entropy-2-integral}).
The proof of that theorem gives the following statements.

\begin{lemma}
  \label{lem:andrews-add}
  For two classes of measurable functions $\mc{G},$ $\mc{H}$, with envelopes
  $G$ and $H$, respectively, for any $\epsilon > 0$ and probability measure $Q$, we have
  \begin{align*}
    \MoveEqLeft
    N( L_2(Q), \mc{G} + \mc{H},  \epsilon \| G + H \|_{Q,2})
    \\&    \le  
    N(L_2(Q), \mc{G}, 2^{-1} \epsilon \|G \|_{Q,2})
    N(L_2(Q), \mc{H}, 2^{-1} \epsilon \|H \|_{Q,2}).
  \end{align*}
\end{lemma}

\begin{lemma}
  \label{lem:andrews-mult}
  For two classes of measurable functions $\mc{G},$ $\mc{H}$, with envelopes
  $G$ and $H$, respectively, and for any $\epsilon > 0$, we have
  \begin{align*}
    \MoveEqLeft    \sup_Q N( L_2(Q), \mc{G} \mc{H},  \epsilon \| (G \vee 1) (H \vee 1) \|_{Q,2})
    \\& \le \sup_Q N(L_2(Q), \mc{G}, 2^{-1} \epsilon \|G \|_{Q,2})
    \sup_Q N(L_2(Q), \mc{H}, 2^{-1} \epsilon \|H \|_{Q,2}).
  \end{align*}
\end{lemma}

\section{Proof of Theorem~\ref{thm:null-normality}}
\label{sec:proof-main-theor-refm}

Here we present the proof of Theorem~\ref{thm:null-normality}.  An
outline of the proof is given in the main text.  The proof is broken into 5 steps and a 6th concluding step.  Step 1 is about the main term.  Steps 2--5 are about remainder terms.  The remainder terms are all (essentially, perhaps after some initial analysis) broken up into two main terms, one being a V-process and the other depending essentially on the size of $(\wh \pi, \wh \mu) - (\bar \pi, \bar \mu)$.

\begin{proof}[Proof of Theorem~\ref{thm:null-normality}]
  We write the test statistic as
  \begin{align}
    \label{eq:Tn-decompose}
    T_n=n\sqrt{h}\int_{\mc{A}}\lb D_1(a)+D_2(a)+D_3\rb^2w(a)\,da, 
  \end{align}
  where
  \begin{align*}
    D_1(a)&:=\widehat{\theta}_h(a)-\tilde{\theta}_h(a)\\
    D_2(a)&:=\tilde{\theta}_h(a)-\PP_n\bar{\xi}\\
    D_3&:=\PP_n\bar{\xi}-\PP_n\widehat{\xi},
  \end{align*}
  and $\tilde{\theta}_h(a)$ is the local linear estimator regressing the oracle
  doubly robust mappings $\bar{\xi}_i:=\xi(\bs{Z}_i;\bar{\pi},\bar{\mu})$ on
  $A_i$.

  We will show the dominating term
  $n\sqrt{h}\int_{\mc{A}}\{D_2(a)\}^2w(a)\,da$ converges to the Normal distribution
  given in the conclusion of the theorem. For the rest of the terms, we will
  see that $n\sqrt{h}\int_{\mc{A}}\{D_1(a)\}^2w(a)\,da$,
  $n\sqrt{h}\int_{\mc{A}} D_3^2w(a)\,da$ and all the cross-product terms are
  asymptotically negligible.

  Note that the negligibility of the cross-product terms does {\it not} follow from negligibility of the main squared terms and the Cauchy-Schwarz inequality.  The reason is that the approximate distribution of  $n\sqrt{h}\int_{\mc{A}}\{D_2(a)\}^2w(a)\,da$ has a mean of order $1/\sqrt{h} \to \infty$.
  Thus, to apply a Cauchy-Schwarz argument to, say, $n \sqrt{h} \int_{\mc{A}} D_1(a) D_2(a) w(a) da$, one would need to not just show that $n h^{1/2} \int_{\mc{A}} D_1(a)^2 w(a)da$ is negligible but that it is $o_p(\sqrt{h})$, which will not generally be true.

  
  \bigskip

  
  \noindent \textbf{Step 1.}
  We first show the asymptotic distribution of
  $n\sqrt{h}\int_{\mc{A}}\{D_2(a)\}^2w(a)\,da$ by applying Theorem 2.1 of
  \citet{alcala1999goodness}
  (which extends the results of \citet{Hardle:1993ih} to allow local polynomial estimators)
  to $n\sqrt{h}\int_{\mc{A}}\{D_2(a)\}^2w(a)\,da$ .
  We can see assumptions (A1) and (A2) of \citet{alcala1999goodness} are automatically satisfied by our
  Assumption~\ref{assm:IA}\ref{assm:IA-item2},
  \ref{assm:DA}\ref{assm:DA-item1} and \ref{assm:DA}\ref{assm:DA-item2}. 

  Next, we check Assumption (A4) of \citet{alcala1999goodness}. By Assumption~\ref{assm:DA}\ref{assm:DA-item3}, we know $\mu_0(\bs{l},a)$ is uniformly bounded. And since $m_0(a) = \int\mu_0(\bs{l},a)\,dP(\bs{l})$, for any $a_0\in\mc{A}$ and $a_n\to a_0$, by Assumption~\ref{assm:DA}\ref{assm:DA-item4} and the dominated convergence theorem, we have $|m_0(a_0)-m_0(a_n)|=|\int\{\mu_0(\bs{l},a_0)-\mu_0(\bs{l},a_n)\}\,dP(\bs{l})|\to 0$ as $a_n \to a_0$ and thus $m_0(a)$ is continuous and bounded. Similarly, we can show $\bar{m}, 1/\varpi_0,1/\bar{\varpi}$ are continuous and uniformly bounded. Then applying the dominated convergence theorem to \eqref{eq:condition-var} with $a_n\to a_0$ gives the continuity of $\sigma^2(a)$. Then by the total variance formula, we have
  \begin{align*}
    \MoveEqLeft\Var\lp\xi(\bs{Z};\bar{\pi},\bar{\mu} )|A=a\rp\\& = \EE\lb
    \Var\lp\xi(\bs{Z};\bar{\pi},\bar{\mu} )|\bs{L},A=a\rp|A=a\rb\\&\quad +\Var\lb  \EE\lp\xi(\bs{Z};\bar{\pi},\bar{\mu} )|\bs{L},A=a\rp|A=a\rb .
  \end{align*}
  From \eqref{eq:dr-mapping}, we have  
  \begin{align*}
    \MoveEqLeft   \Var\lp\xi(\bs{Z};\bar{\pi},\bar{\mu} )|\bs{L},A=a\rp\\
    &=\Var\lp Y-\mu(\bs{L},A)|\bs{L},A=a\rp\lb\frac{1}{\bar\pi(a|\bs{L})}\int\bar\pi(a|\bs{l})\,dP(\bs{l})\rb^2\\
    & = \tau(\bs{L},a)\lb\frac{1}{\bar\pi(a|\bs{L})}\int\bar\pi(a|\bs{l})\,dP(\bs{l})\rb^2
  \end{align*}
  and by Assumption~\ref{assm:DA}\ref{assm:DA-item4}, we have $\tau(\bs{l},a)$ is
  positive everywhere, so we know the expectation $\EE\lb
  \Var\lp\xi(\bs{Z};\bar{\pi},\bar{\mu} )|\bs{L},A=a\rp|A=a\rb$ is positive for
  all $a\in \mc{A}$ and thus the
  conditional variance $\sigma^2(a)$ is positive.
  Then since
  $\mc{A}$ is compact, by the continuity of $\sigma^2(a)$, we have $\sigma^2(a)$ is bounded below from 0 and bounded
  above.

  Assumption (A5) of \citet{alcala1999goodness} also holds since our null parametric
  model is a linear model
  with only intercept term which can be estimated with
  $\sqrt{n}$-consistency. Assumptions (K1) and (K2) of \citet{alcala1999goodness} are also met by our
  Assumption~\ref{assm:EA}\ref{assm:EA-item3} and
  \ref{assm:EA}\ref{assm:EA-item4}. So we have
  \begin{align}
    d\lb n\sqrt{h}\int_{\mc{A}} D_2(a)^2 w(a)\,da,
    \, N\lp
    b_{0h},V\rp\rb\to 0,
  \end{align}
  as $n\to \infty$.

  \bigskip



  \bigskip
  \noindent \textbf{Step 2 ($D_3^2$).}  \;
  We now show
  that $D_3 = o_p(1 / \sqrt{n h^{1/2}})$ and so in particular, the integral
  $n\sqrt{h}\int_{\mc{A}}D_3^2w(a)\,da$ is $o_p(1)$ as $n \to \infty$.
  We can write $D_3$ as
  \begin{align}
    \label{eq:D3-decomp}
    \begin{split}  \MoveEqLeft\PP_n\{\xi(\bs{Z};\bar{\pi},\bar{\mu})-\widehat{\xi}(\bs{Z};\widehat{\pi},\widehat{\mu})\}\\&=\PP_n\{\xi(\bs{Z};\bar{\pi},\bar{\mu})-\xi(\bs{Z};\widehat{\pi},\widehat{\mu})\}+\PP_n\{\xi(\bs{Z};\widehat{\pi},\widehat{\mu})-\widehat{\xi}(\bs{Z};\widehat{\pi},\widehat{\mu})\}.
    \end{split}
  \end{align}
  This type of decomposition is used repeatedly below and is discussed/explained in the proof outline in the main document. 
  The first term above (is split up into an empirical process and a `second order remainder' term and) has order of magnitude given by
  Lemma~\ref{lem:step2-1} as
  $O_p(1 /\sqrt{n}+s_n^{\infty}r_n^{\infty})$.  The second term on the right side of
  \eqref{eq:D3-decomp} is a V-process and is shown to be $O_p(n^{-1/2})$ by Lemma~\ref{lem:3}.
So $D_3=O_p(1 /\sqrt{n}+s_n^{\infty}r_n^{\infty})$ and  thus we
conclude that $D_3= o_p((nh^{1/2})^{-1/2})$ by Assumption~\ref{assm:EA}\ref{assm:EA-item5}.

  
  \bigskip
  \noindent \textbf{Step 3 ($D_1^2$).} \;  Now we show the asymptotic negligibility of
  $n\sqrt{h}\int_{\mc{A}}D_1(a)^2w(a)\,da$.  To use the standard
  representation of a local polynomial estimator (see
  Lemma~\ref{lem:LPE-basic-rep}), we let
  $\wh{\bs{D}}_{ha} = \PP_n\{g_{ha}(A)K_{ha}(A)g_{ha}^T(A)\}$ and
  $W_{ha}(t) := g_{h,0}^T(0) \wh{\bs{D}}_{ha}^{-1} g_{ha}(t) K_{ha}(t)$.
  Then (by Lemma~\ref{lem:LPE-basic-rep}) we write $D_1(a)$ as
  \begin{align}
    \label{eq:D1-decomp}
    \begin{split}
      D_1(a)&=g_{ha}^T(a) \widehat {\bs D}_{ha}^{-1} \PP_n \ls g_{ha}(A) K_{ha}(A)
      \lb  \wh \xi( \bs Z; \wh \pi, \wh \mu ) - \xi(\bs Z; \bar \pi,
      \bar \mu) \rb \rs\\
      &=d_{1,1}(a)+d_{1,2}(a)
    \end{split} 
  \end{align}
  where
  \begin{equation}
    \label{eq:13}
    \begin{split}
      d_{1,1}(a) & := \PP_n \ls W_{ha}(A)
      \lp  \wh \xi( \bs Z; \wh \pi, \wh \mu ) - \xi(\bs Z; \widehat{\pi},
      \widehat{\mu}) \rp \rs, \\
      d_{1,2}(a)&:=  \PP_n \ls W_{ha}(A)
      \lp   \xi( \bs Z; \wh \pi, \wh \mu ) - \xi(\bs Z; \bar{\pi},
      \bar{\mu}) \rp \rs.
    \end{split}
  \end{equation}
  We have that
  $\int d_{1,2}(a)^2 w(a) \,da=
  o_p(1/n+(s_n^{\infty}r_n^{\infty})^2)$ by
  Lemma~\ref{lem:step3-1}.
  For the other term, 
  we have $\int \{d_{1,1}(a)\}^2w(a)da = O_p(1/n)$ by Lemma~\ref{lem:5}.
  Applying the Cauchy-Schwarz inequality
  yields $\int
  D_1(a)^2w(a)\,da=o_p(1/n+(s_n^{\infty}r_n^{\infty})^2)$ and thus $n\sqrt{h}\int
  D_1(a)^2w(a)\,da=o_p(1)$ by
  Assumptions~\ref{assm:EA}\ref{assm:EA-item1} and \ref{assm:EA}\ref{assm:EA-item5}.

  \medskip  \noindent\textbf{Step 4 ($D_2 D_3$).} \;  Next we
 will consider  the integrated crossproduct term
  $n\sqrt{h}\int_{\mc{A}}D_2(a)D_3w(a)\,da$.
  This term can be
  analyzed using the above results, since $D_3$ factors out of the integral. 
  The analysis is given in Lemma~\ref{lem:d2d3}, which shows that
  $n\sqrt{h}\int_{\mc{A}}D_2(a)D_3w(a)\,da$ is $o_p(1)$
  by showing that
  $\int_{\mc A} D_2(a) w(a) \, da$ is  $O_p(n^{-1/2}) = o_p(1 / \sqrt{n \sqrt{h}})$
  (and combining that with the result of Step~2).


  \bigskip \noindent
  \textbf{Step 5 ($D_1 D_3$).} \;
  Here we show $\int_{\mc{A}} D_1(a) D_3 w(a) \, da$ is $o_p ( (n \sqrt{h})^{-1})$. 
  This follows directly by the results of Step 2 and of Step 3, and the Cauchy-Schwarz inequality.

  \bigskip
  \noindent\textbf{Step 6 ($D_1D_2$).} \; Now we show that $\int_{\mc{A}}D_1(a)D_2(a)w(a)\,da$  is $  o_p( (n\sqrt{h})^{-1})$.
  Recall in Step 4, we write
  \begin{align}
    D_2(a)=\tilde{\theta}_h(a)-\PP\bar{\xi}+\PP\bar{\xi}-\PP_n\bar{\xi},
  \end{align}
  and thus there we can focus on $\int D_1(a)\{\tilde{\theta}_h(a)-\PP\bar{\xi}\}
  w(a)\,da$, because the other term
  \begin{align*}
    \int D_1(a)(\PP\bar{\xi}-\PP_n\bar{\xi})
    w(a)\,da
    &=(\PP\bar{\xi}-\PP_n\bar{\xi})\int D_1(a)w(a)\,da\\
    &\le |\PP\bar{\xi}-\PP_n\bar{\xi}|\int |D_1(a)|w(a)\,da\\
    &\le |\PP\bar{\xi}-\PP_n\bar{\xi}|\sqrt{\int \{D_1(a)\}^2w(a)\,da}\\
    &=O_p(n^{-1/2}) o_p( (n\sqrt{h})^{-1/2}).
  \end{align*}
  where the last line above comes from the results in Step 3.
  So we can write
  \begin{align*}
    & \int_{\mc{A}}D_1(a)D_2(a)w(a)\,da\\
    &\quad=\int_{\mc{A}}D_1(a)\lb\tilde{\theta}_h(a)-\PP\bar{\xi}\rb w(a)\,da +o_p((n\sqrt{h})^{-1}),
  \end{align*}
  and further  we write $\int_{\mc{A}}D_1(a)\lb\tilde{\theta}_h(a)-\PP\bar{\xi}\rb
  w(a)\,da$ as the sum 
  \begin{align}
    \label{eq:step5-decompose}
    \int_{\mc{A}}d_{1,1}(a)\lb\tilde{\theta}_h(a)-\PP\bar{\xi}\rb
    w(a)\,da+\int_{\mc{A}}d_{1,2}(a)\lb\tilde{\theta}_h(a)-\PP\bar{\xi}\rb
    w(a)\,da,
  \end{align}
  where $d_{1,1}$ and $d_{1,2}$ are defined in equation \eqref{eq:13}.
  From Lemma~\ref{lem:equiv-kernel-lemma}
  $ g_{ha}(a) \wh{\bs D}_{ha}^{-1} g_{ha}(A_1)$ converges to
  $\varpi_0(a)^{-1}$ almost surely (so also in-probability) regardless of the
  value of $A_1$.  Further, as in the proof of Lemma~\ref{lem:d2d3}, we write
  out $\tilde{\theta}_h(a)-\PP\bar{\xi}$ 
  as
  \begin{align}
    \wt \theta(a) - \PP \overline{\xi}
    =
    \frac{1}{n}\sum_{i=1}^n\frac{1}{h\varpi_0(a)}
    K\lp\frac{A_i-a}{h}\rp
    \lp \bar{\xi}_i-\PP\bar{\xi}\rp
    \lb                                      1+o_p(1)\rb.  
  \end{align}
Let $\overline{\epsilon}:=\xi(\bs{Z};\overline{\pi},\overline{\mu})-\PP\xi(\bs{Z};\overline{\pi},\overline{\mu})$, and specifically $\overline{\epsilon}_i:=\xi(\bs{Z}_i;\overline{\pi},\overline{\mu})-\PP\xi(\bs{Z};\overline{\pi},\overline{\mu})$ for $i=1,\ldots, n$.  So the order of  the first term on right side of   \eqref{eq:step5-decompose} is dominated by
  \begin{align}
    \label{eq:3-order-v-simple}
    \int    \frac{1}{\varpi_0^2(a)} \PP_n \ls K_{ha}(A)
    \lb  \wh \xi( \bs Z; \wh \pi, \wh \mu ) -  \xi(\bs Z;  \wh \pi,
    \wh \mu) \rb\rs\frac{1}{n}\sum_{i=1}^n
    K_h\lp A_i-a\rp
    \bar\epsilon_i w(a)\,da.  
  \end{align}
 Since we can write
 \begin{align*}
 \wh \xi( \bs Z; \wh \pi, \wh \mu ) -  \xi(\bs Z;  \wh \pi,
    \wh \mu ) &=
                \frac{Y-\wh\mu(\bs{L},A)}{\wh\pi(A|\bs{L})}\lp\frac{1}{n}\sum_{j=1}^n\wh\pi(A|\bs{L}_j)-\int\wh\pi(A|\bs{l})\,d\PP(\bs{l})\rp\\
   &\qquad +
     \frac{1}{n}\sum_{j=1}^n\wh\mu(\bs{L}_j,A)-\int\wh\mu(\bs{l},A)\,d\PP(\bs{l})\\
   &=\frac{Y-\wh\mu(\bs{L},A)}{\wh\pi(A|\bs{L})}\frac{1}{n}\sum_{j=1}^n\wt{\wh\pi}(A|\bs{L}_j)
     + \frac{1}{n}\sum_{j=1}^n\wt{\wh\mu}(\bs{L}_j, A),
 \end{align*}
 where   we let tilde $ \wt{ \cdot } $
  operate on any $\mu, \pi$ to yield
  $\wt{\pi}(a_1|l_2) := \pi(a_1|l_2) - \PP \pi(a_1| \bs L)$,   and similarly
  $\wt\mu(l_2,a_1) := \mu(l_2,a_1) - \PP \mu(\bs L,a_1)$,
 we can further decompose \eqref{eq:3-order-v-simple} as
 \begin{align*}
 \int    \frac{1}{\varpi_0^2(a)} \PP_n \ls K_{ha}(A)
    \lb \frac{Y-\wh\mu(\bs{L},A)}{\wh\pi(A|\bs{L})}\frac{1}{n}\sum_{j=1}^n\wt{\wh\pi}(A|\bs{L}_j) \rb\rs\frac{1}{n}\sum_{i=1}^n
    K_h\lp A_i-a\rp
    \bar\epsilon_i w(a)\,da
 \end{align*}
 plus
 \begin{align*}
    \int    \frac{1}{\varpi_0^2(a)} \PP_n \ls K_{ha}(A)
  \frac{1}{n}\sum_{j=1}^n\wt{\wh\mu}(\bs{L}_j, A) \rs\frac{1}{n}\sum_{i=1}^n
    K_h\lp A_i-a\rp
    \bar\epsilon_i w(a)\,da,
 \end{align*}
 which are in the forms of the two terms studied in Lemma~\ref{lem:4}
 ($V_{12,\pi}$ and $V_{12,\mu}$, respectively.).
 So applying Lemma~\ref{lem:4} yields that \eqref{eq:3-order-v-simple} is
 $o_p((n\sqrt{h})^{-1})$.
  For the second term on the right side of  \eqref{eq:step5-decompose},  we have
  \begin{align*}
    \MoveEqLeft    \int d_{1,2}(a)\lb\tilde{\theta}_h(a)-\PP\bar{\xi}\rb w(a)\,da\\
    &= \int R_{n,1,a}\lb\tilde{\theta}_h(a)-\PP\bar{\xi}\rb w(a)\,da+\int R_{n,2,a}\lb\tilde{\theta}_h(a)-\PP\bar{\xi}\rb w(a)\,da,
  \end{align*}
  where
  \begin{align*}
    R_{n,1,a} & := g_{ha}^T \widehat {\bs D}_{ha}^{-1} (\PP_n - \PP) \ls g_{ha}(A) K_{ha}(A)
                \lb   \xi( \bs Z; \wh \pi, \wh \mu ) - \xi(\bs Z; \bar \pi, \bar \mu) \rb \rs  \\
    R_{n,2,a} & := g_{ha}^T \widehat {\bs  D}_{ha}^{-1}  \PP \ls  g_{ha}(A) K_{ha}(A)
                \lb   \xi( \bs Z; \wh \pi, \wh \mu ) - \xi(\bs Z; \bar \pi, \bar \mu) \rb\rs.
  \end{align*}
  The first term $\int R_{n,1,a}\{\tilde{\theta}_h(a)-\PP\bar{\xi}\} w(a)\,da$  could be bounded by Cauchy-Schwarz as
  \begin{align*}
    \MoveEqLeft\int R_{n,1,a}\lb\tilde{\theta}_h(a)-\PP\bar{\xi}\rb w(a)\,da\\
    &\le \sqrt{\int (R_{n,1,a})^2w(a)\,da\int
      \lb\tilde{\theta}_h(a)-\PP\bar{\xi}\rb^2 w(a)\,da}.
  \end{align*}
From Lemma~\ref{lem:step3-1}, we know $\int
(R_{n,1,a})^2w(a)\,da=o_p(1/n)$ and from Step 1, we have $\int
      \{\tilde{\theta}_h(a)-\PP\bar{\xi}\}^2
      w(a)\,da=\{O(1/\sqrt{h})+O_p(1)\}/(n\sqrt{h})$. Thus
        \begin{align*}
    \MoveEqLeft\int R_{n,1,a}\lb\tilde{\theta}_h(a)-\PP\bar{\xi}\rb w(a)\,da\\
     &=\sqrt{o_p(1/n)\lb O\lp\frac{1}{\sqrt{h}}\rp+O_p(1)\rb/(n\sqrt{h})}\\
    &=o\lp(n\sqrt{h})^{-1}\rp.
  \end{align*}
  Next, with a similar argument as in  \eqref{eq:3-order-v-simple}, we see
the   order of the 
  second term $\int R_{n,2,a}\{\tilde{\theta}_h(a)-\PP\bar{\xi}\} w(a)\,da$ is
  \begin{align}
    \int    \frac{1}{\varpi_0^2(a)} \PP \ls K_{ha}(A)
    \lb   \xi( \bs Z; \wh \pi, \wh \mu ) -  \xi(\bs Z;  \bar \pi,
    \bar \mu) \rb\rs\frac{1}{n}\sum_{i=1}^n
    K_h\lp A_i-a\rp
    \bar\epsilon_i w(a)\,da\lb 1+ o_p(1)\rb,      
  \end{align}
  and thus dominated by
    \begin{align}
    \int    \frac{1}{\varpi_0^2(a)} \PP \ls K_{ha}(A)
    \lb   \xi( \bs Z; \wh \pi, \wh \mu ) -  \xi(\bs Z;  \bar \pi,
    \bar \mu) \rb\rs\frac{1}{n}\sum_{i=1}^n
    K_h\lp A_i-a\rp
    \bar\epsilon_i w(a)\,da.      
  \end{align}
  An empirical process argument,
  given in
  Lemma~\ref{lem:step5-decompose-term2},
  shows that the above term is $o_p( (n\sqrt{h})^{-1})$.
  Finally, combining the above, it is straight forward to see
  $\int_{\mc{A}}D_1(a)D_2(a)w(a)\,da = o_p((n\sqrt{h})^{-1})$.

  \bigskip
  \noindent\textbf{Step 7 (conclusions).}
  We have shown that $T_n = n\sqrt{h}\int_{\mc{A}}\{D_2(a)\}^2w(a)\,da+o_p(1)$. Then
  by the triangle inequality, we have
  \begin{align*}
    d\lb T_n,N(b_{0h},V))\rb&\le  d\lb T_n,n\sqrt{h}\int_{\mc{A}}\{D_2(a)\}^2w(a)\,da\rb\\&\quad+d\lb n\sqrt{h}\int_{\mc{A}}\{D_2(a)\}^2w(a)\,da, N\lp
    b_{0h},V\rp\rb.
  \end{align*}
  The first term on the right side of the inequality has been shown to go to 0 as
  $n\to \infty$. For the second term,
  by the definition of Dudley metric, we have
  \begin{align*}
    \MoveEqLeft  d\lb T_n,n\sqrt{h}\int_{\mc{A}}\{D_2(a)\}^2w(a)\,da\rb\\ &= \sup\lb \EE g(T_n)-\EE g\lp
                                                                            n\sqrt{h}\int_{\mc{A}}\{D_2(a)\}^2w(a)\,da\rp:
                                                                            \|g\|_{BL}\le 1\rb\\
                                                                          &\le \EE \lb\left\lvert T_n-n\sqrt{h}\int_{\mc{A}}\{D_2(a)\}^2w(a)\,da)\right\rvert\wedge 2\rb.
  \end{align*}
  Since $|T_n-n\sqrt{h}\int_{\mc{A}}\{D_2(a)\}^2w(a)\,da|=o_p(1)$, by the dominated
  convergence theorem, the expectation in the last line above converge to 0 as
  $n\to \infty$ and thus we have
  \begin{align*}
    d\lb T_n,N(b_{0h},V))\rb\to 0,
  \end{align*}
  as $n\to \infty$ and that completes the proof.

\end{proof}

\section{Applying U- or V-process results to remainder terms}
\label{sec:proofs-u-or-v}

The following three lemmas provide the negligibility of remainder terms in the analysis of our test statistic's limit distribution.  They are all in V-statistic form (if one regards $\wh \pi,\wh\mu$ as fixed) and thus their analysis (allowing $\wh \pi, \wh \mu $ to vary) requires the theory of V-processes.


\begin{lemma}
  \label{lem:3}
  Let the assumptions
  of Theorem~\ref{thm:null-normality} hold.
  Then
  \begin{equation}
    \label{eq:102}
    \PP_n\lb\widehat{\xi}(\bs{Z};\widehat{\pi},\widehat{\mu}) -
    \xi(\bs{Z};\widehat{\pi},\widehat{\mu})\rb
    = O_p(n^{-1/2}) 
    \text{ as } n \to \infty.
  \end{equation}
\end{lemma}

\begin{lemma}
  \label{lem:5}
  Let the assumptions of Theorem~\ref{thm:null-normality} hold.
  Let 
  \begin{equation*}
    d_{1,1}(a) := g_{ha}^T \widehat {\bs D}_{ha}^{-1} \PP_n \ls g_{ha}(A) K_{ha}(A)
    \lb  \wh \xi( \bs Z; \wh \pi, \wh \mu ) - \xi(\bs Z; \widehat{\pi},
    \widehat{\mu}) \rb \rs .
  \end{equation*}
  Then
  $n \sqrt{h} (\sup_{a \in \mc{A}} d_{1,1}(a)^2)  = O_p( \sqrt{h})$ and so
  $n \sqrt{h} \int (d_{1,1}(a))^2 w(a) da = O_p(\sqrt{h}) = o_p(1)$ as $n \to \infty$.
\end{lemma}

\medskip

\noindent Recall in the above statement, $\wh{\bs{D}}_{ha} = \PP_n\{g_{ha}(A)K_{ha}(A)g_{ha}^T(A)\}$. The V-statistic terms analyzed in the next lemma arise
from certain cross product error terms.  We let tilde $ \wt{ \cdot } $
operate on any $\mu, \pi$ to yield
$\wt{\pi}(a_1|l_2) := \pi(a_1|l_2) - \PP \pi(a_1| \bs L)$
and similarly
$\wt\mu(l_2,a_1) := \mu(l_2,a_1) - \PP \mu(\bs L,a_1)$. 
Let
$\overline{\epsilon}:=\xi(\bs{Z};\overline{\pi},\overline{\mu})-\PP\xi(\bs{Z};\overline{\pi},\overline{\mu})$,
and specifically
$\overline{\epsilon}_i:=\xi(\bs{Z}_i;\overline{\pi},\overline{\mu})-\PP\xi(\bs{Z};\overline{\pi},\overline{\mu})$
for $i=1,\ldots, n$.
\begin{lemma}
  \label{lem:4}
  Let the assumptions of Theorem~\ref{thm:null-normality} hold.
  Let
  \begin{equation*}
    V_{12,\mu} :=
    \inv{n^3} \sum_{i,j,k=1}^n
    \int_{\mc A}
    \varpi_0(a)^{-2} 
    \ls  K_{ha}(A_i) \wt \mu (\bs L_j , A_i)  K_{ha}(A_k) \overline \epsilon_k \rs
    w(a)
    da, \label{eq:3V-term-defn-mu} \\
  \end{equation*}
  \begin{equation*}
    \label{eq:3V-term-defn-pi}
    \begin{split}
      \MoveEqLeft
      V_{12, \pi} := 
      \inv{n^3} \sum_{i,j,k=1}^n
      \int_{\mc A}
      \varpi_0(a)^{-2} \times \\
      &      \ls  K_{ha}(A_i)
      \lp  \frac{Y_i - \mu(\bs L_i, A_i)}{\pi(A_i | \bs L_i) } \wt \pi(A_i | \bs L_j) \rp
      K_{ha}(A_k) \overline \epsilon_k \rs
      w(a) da.
    \end{split}
  \end{equation*}
  Then $V_{12,\wh \mu}$ and $V_{12, \wh \pi}$ are $O_p(n^{-2}h^{-3/2})$
  as $n \to \infty$.
\end{lemma}

\begin{proof}[Proof of Lemma~\ref{lem:3}]
  We first write
  \begin{equation}
    \label{eq:102-extra}
    \begin{split}
      \MoveEqLeft \PP_n
      \lb\widehat{\xi}(\bs{Z};\widehat{\pi},\widehat{\mu}) -
      \xi(\bs{Z};\widehat{\pi},\widehat{\mu})\rb\\
      &=\frac{1}{n}\sum_{i=1}^n\ls\frac{Y_i-\widehat{\mu}(\bs{L}_i,A_i)}{\widehat{\pi}(A_i|\bs{L}_i)}\lb\frac{1}{n}\sum_{j=1}^n\widehat{\pi}(A_i|\bs{L}_j)-\int\widehat{\pi}(A_i|\bs{l})\,dP(\bs{l})\rb\right.\\
      &\qquad\left.+\lb\frac{1}{n}\sum_{j=1}^n\widehat{\mu}(\bs{L}_j,A_i)-\int\widehat{\mu}(\bs{l},A_i)\,dP(\bs{l})\rb\rs.
    \end{split}
  \end{equation}
  Thus, the  mean of the given term can be bounded by the mean of a process where
  $\pi$, $\mu$, range over $\mc F_\pi$, $\mc F_\mu$.
  We define  the V-processes of interest as follows. For $w_i \in \mc{Z} := \mc L \times \mc A \times \mc Y$, $i=1,2$, let
  \begin{equation}
    \label{eq:1002}
    h_1(w_1, w_2)
    \equiv   h_{1, \pi, \mu}(w_1, w_2) 
    := \frac{y_1 - \mu(l_1,a_1)}{\pi(a_1|l_1)}
    \wt \pi(a_1 | l_2) + \wt \mu(l_2,a_1)
  \end{equation}
  be a (non-symmetric) U- or V-statistic ``kernel'', indexed by $\mu,\pi$,
  and where, recall
  we let  tilde $ \wt{ \cdot  } $  operate on  any $\mu, \pi$ to yield
  $\wt{\pi}(a_1|l_2) := \pi(a_1|l_2) - \PP \pi(a_1|L),$
  and similarly $\wt\mu(l_2,a_1) := \mu(l_2,a_1) -  \PP \mu(L,a_1)$. Then, recalling $\{ W_i\} =  \{ (L_i,A_i,Y_i)\}$ is our i.i.d.\ sample,  \eqref{eq:102}  is of the form
  \begin{equation}
    \label{eq:1113}
    n^{-2} \lp \sum_{i=1}^n h_1(W_i, W_i) +
    \sum_{1 \le i < j \le n} h_1(W_i, W_j) + h_1(W_j,W_i)\rp.
  \end{equation}
  Note that
  $\PP h_1(w_1, W_2) = 0$, for almost any $w_1$, so
  $\PP h_1(W_1, W_2) = 0$. Now by Assumption~\ref{assm:DA}\ref{assm:DA-item3} and
  \ref{assm:DA}\ref{assm:DA-item4}, using the total variance formula, we can see
  $Y$ has a finite variance. In addition, by
  Assumption~\ref{assm:EA-item2}
  $\pi$ and $\mu$ are bounded above, and $\pi$ is
  bounded below, it is immediate that $\PP |h_1(W_1, W_1)| < \infty$ and
  $\PP h_1(W_1,W_2)^2<\infty$. 
  The larger term will be the double summation(s)
  in \eqref{eq:1113}, which  is  in U-statistic form so we now
  introduce some notation so we can then apply U-process results.

  Let
  \begin{equation}
    \label{eq:10003}
    g_1(w_1,w_2) \equiv g_{1,\mu,\pi}(w_1,w_2) := h_1(w_1,w_2) + h_1(w_2,w_1).
  \end{equation}
  To apply the  maximal inequality in Proposition~\ref{lem:100} to our particular U-processes we thus need to bound the appropriate uniform entropy-type integrals and compute the envelope moments, for the classes of functions
  \begin{equation*}
    \mc G_1:=
    \{ g_{1,\mu,\pi} :
    \mu \in \mc F_\mu, \pi \in \mc F_\pi
    \}.
  \end{equation*}
  We start by considering the covering numbers.
  Take a generic class of functions $\mc F$ on a space $\mc X$ with finite covering number $N( \mc F, \| \cdot \|_{2,Q}, \tau)$ and envelope $F$.  Then it is easy to verify
  that the class $\mc{F}^\circ$ of functions $f^\circ(x,z) := f(x)$ defined on the
  extended space $\mc X \times \tilde{\mc X} $, some measurable space $\tilde{\mc W}$, has
  $N( \mc F, \| \cdot \|_{2,Q}, \tau)
  = N( \mc{F}^\circ , \| \cdot \|_{2,Q^\circ}, \tau)$  for any $Q^\circ$ on $\mc X \times \tilde{\mc X}$ that extends $Q$ in the sense that $Q$ is the marginal of $Q^\circ$ on $\mc X$.  In particular,
  $\sup_Q N( \mc F, \| \cdot \|_{2,Q}, \tau)
  = \sup_{Q^\circ} N( \mc{F}^\circ , \| \cdot \|_{2,Q^\circ}, \tau)$.
  Thus, if we consider
  the class of functions $\mc{F}_\mu^\circ$ given by $(l_1, a_1, l_2, a_2) \mapsto
  \mu(l_2, a_1)$ for $\mu \in \mc{F}_\mu$, then this class has the same uniform
  covering number as the original class $\mc{F}_\mu$.  Similarly for $\mc{F}_\pi$.

  \noindent {\bf Bounding entropy.}  First we bound $J$ and $J_2$ for the appropriate classes of functions.
  Let $\mc{H}_1$ be the class of $h_1$ functions (defined in
  \eqref{eq:1002}).
  We can write the class $\mc{H}_1$ as
  \begin{equation*}
    \mc{H}_1 =
    (Y_1 - \mc{F}_\mu)
    \mc{F}_{\pi}^{-1} \wt{\mc{F}}^\circ_\pi + \wt{\mc{F}}^\circ_{\mu},
  \end{equation*}
  where: we abuse notation to let, e.g., $\mc{F}_\mu$ refer to the class $(l_1, a_1, y_1, l_2, a_2,y_2) \mapsto \mu(l_1, a_1)$, and similarly for $\mc{F}_\pi$;
  we define operations on classes of functions so $\mc{F} \mc{G} := \{ fg : f \in \mc{F}, g \in \mc{G} \}$, $\mc{F} + \mc{G} := \{ f + g : f \in \mc{F}, g \in \mc{G}\}$, $\mc{F}^{-1} := \{ f^{-1} : f \in \mc{F} \}$;
  we let $Y_1$ denote the function on $\mc{W}^2$ that returns just $y_1$;
  and let $\wt{\mc{F}}^\circ_\mu$ be the class of
  functions
  $(l_1, a_1, l_2, a_2) \mapsto \wt{\mu}(l_2, a_1)$
  and similarly 
  $\wt{\mc{F}}^\circ_\pi$ is the class of
  functions
  $(l_1, a_1, l_2, a_2) \mapsto \wt{\pi}(a_1|l_2)$.  The class $\wt{\mc{F}}^\circ_\mu$ can be written as $\mc{F}^\circ_\mu - \{ \PP \mu(L, \cdot) ; \mu \in \mc{F}^\circ_\mu \}$.
  By the proof of
  Lemma~20
  of \cite{Nolan:jw},  
  $\{ \PP \mu(L, \cdot) ; \mu \in \mc{F}_\mu \}$ has uniform entropy smaller than
  that of $\mc{F}_\mu$.
  A similar statement holds for $\wt{\mc{F}}^\circ_\pi$.
  Thus, we can apply
  Lemmas \ref{lem:andrews-add} and \ref{lem:andrews-mult}
  (see also Lemma 16 of \cite{Nolan:jw}) 
  to $\mc{H}_1$.
  This shows $J_2(1, \mc{H}_1) < \infty$.
  Let $\mc{G}_1 := \{ g_{1,\mu,\pi} :
  \mu \in \mc F_\mu, \pi \in \mc F_\pi
  \}$.
  Note that $N(\mc{G}_1,  L_2(Q), \epsilon \sqrt{2})
  = N(\mc{H}_1, L_2(Q), \epsilon)$ so we can conclude
  $J_2(1, \mc{G}_1) < \infty$. By  Lemma~20
  of \cite{Nolan:jw},
  we also conclude that $J(1, \PP \mc{G}_1) < \infty$.  By
  Proposition~\ref{lem:100}, we conclude that $\PP \| U_n \|_{\mc{G}_1}$ is thus
  bounded above by a constant, as desired.

  \noindent {\bf Envelopes.} Now we consider envelopes and their squared expectation for the appropriate classes.
  Let $G_1$  
  be the envelope for $\mc{G}_1$.
  By our assumptions~\ref{assm:DA}\ref{assm:DA-item3}, \ref{assm:DA}\ref{assm:DA-item4}
  and \ref{assm:EA-item2} 
  we see that $EG_1^2 < \infty$ (and is independent of $n$).
  It's also easily seen that   $\mc{H}_1$ has squared-integrable envelope $H_1$.

  Next, for the minimal envelope $K_1$ for $\PP\mc{G}_1$, it is clear again
  that we have $\PP K_1(W)^2<\infty$.  Finally, we apply
    Proposition~\ref{lem:100} to \eqref{eq:102} and conclude \eqref{eq:102} is
  of order $O_p(n^{-1/2})$.
\end{proof}

\begin{proof}[Proof of Lemma~\ref{lem:5}]
  For $n\sqrt{h}\int\{d_{1,1}(a)\}^2w(a)\,da$, we apply a 
  V-process approach, like that in Lemma~\ref{lem:3}, but now we must accommodate the kernel. 
  For $w_i \in \mc{Z} = \mc L
  \times \mc A \times \mc Y$, $i=1,2$, let $h_2 \equiv h_{2,\mu, \pi}$ be
  \begin{equation}
    \label{eq:1003}
    \begin{split}
      h_2(w_1, w_2)
      & :=  g_{ha}(a_1)
      K_{ha}(a_1)h_1(w_1, w_2),
    \end{split}
  \end{equation}
  where $h_1 \equiv h_{1, \mu, \pi}$
  is as defined in $\eqref{eq:1002}$,
  and then we can write 
  \begin{align}
    \label{eq:1114}
    \begin{split}
      d_{1,1}'(a)
      &  := \PP_n  \ls g_{ha}(A) K_{ha}(A)
      \lb  \wh \xi( \bs Z; \hat \pi, \hat \mu ) - \xi(\bs Z; \widehat{\pi},
      \widehat{\mu}) \rb \rs \\
      & =  \inv{n^{2}} \lp\sum_{i=1}^n h_{2, \wh \mu, \wh \pi}(W_i, W_i) +
      \sum_{1 \le i < j \le n} h_{2, \wh \mu, \wh \pi}(W_i, W_j) + h_{2, \wh \mu, \wh \pi}(W_j,W_i)\rp.  
    \end{split}
  \end{align}
  Note that
  $\PP h_2(w_1, W_2) = 0$, for almost any $w_1$, so
  $\PP h_2(W_1, W_2) = 0$. And we have that  $\PP h_2(W_1,W_1) $
  equals
  \begin{align}
    \label{eq:1005}
    \begin{split}
      \MoveEqLeft    \int g_{1,0}(u) K(u) \lp \frac{y - \mu(l,a + h u)}{\pi(a + hu|l)}
      \wt \pi(a + hu | l) + \wt \mu(l, a + hu) \rp\times\\
      &p(l, a+hu, y)   
      d\nu(l, y) du
    \end{split}
  \end{align}
  which is bounded (uniformly in $a$ and $h$) since $Y$ has a finite mean (and by
  Assumption~\ref{assm:EA-item2}
  $\pi$ is uniformly bounded below away from 0, and $\pi,\mu$ bounded above). Thus
  \begin{equation}
    \label{eq:U-diagonals-Op-h2}
    n^{-2} \sum_{i=1}^n h_2(W_i,W_i) = O_p(n^{-1}).
  \end{equation}
  The larger term will be the double summation(s)
  in \eqref{eq:1114}.  Now let
  \begin{equation}
    \label{eq:10003}
    g_2(w_1,w_2) \equiv g_{2,\mu,\pi}(w_1,w_2) := h_2(w_1,w_2) + h_2(w_2,w_1).
  \end{equation}
  Since $h_2$ and $g_2$ are vector functions
  (recall $g_{ha}(a_1) = (1, (a_1-a)/h)^T$),
  we refer
  to their components as
  $h_{2;i} \equiv h_{2 , \mu, \pi; i} $ and
  $g_{2; i} \equiv g_{2, \mu, \pi; i}$, $i=1,2$.
  Similarly as in Step 2, to apply the above maximal inequality in
  Proposition~\ref{lem:100} to the  particular U-processes we thus need to bound the appropriate uniform entropy-type integrals and compute the envelope moments, for the class of functions
  \begin{equation*}
    \mc{G}_{2;i} := \{ g_{2,\mu,\pi; i} :
    \mu \in \mc F_\mu, \pi \in \mc F_\pi
    \},
    \quad i=1,2.
  \end{equation*}

  \noindent {\bf Bounding entropy.} We consider $\mc{H}_{2;i}$,
  $i=1,2$, which are the classes of the two coordinate functions of
  the $h_{2}$ functions (defined in \eqref{eq:1003} and \eqref{eq:1002}). We can write
  $\mc{H}_{2;1}, \mc{H}_{2;2}$ as
  \begin{align}
    \begin{split}
      h  \mc{H}_{2;1} &:=
      \lb
      (w_1, w_2)\mapsto
      h K_{ha}(a_1) h_{1,\mu,\pi}(w_1,w_2) ;
      \mu \in \mc{F}_\mu, \pi \in \mc{F}_\pi , a \in \mc{A}
      \rb \\
      h^2 \mc{H}_{2;2} & :=
      \lb
      (w_1, w_2)\mapsto
      (a_1-a) h K_{ha}(a_1) h_{1,\mu,\pi}(w_1,w_2) ;
      \mu \in \mc{F}_\mu, \pi \in \mc{F}_\pi , a \in \mc{A}
      \rb
    \end{split}
  \end{align}
  Note that the classes of functions $ K( (\cdot - a)/h)$
  and
  $(\cdot - a)  K( (\cdot - a)/h)$,
  indexed over $a \in \mc{A}$, are both bounded VC classes of measurable functions  
  by Assumption~\ref{assm:EA}\ref{assm:EA-item4}.
  Again applying
  Lemmas \ref{lem:andrews-add} and \ref{lem:andrews-mult} we can conclude that $h^i \mc{H}_{2;i}$ has $J_2(1, h^i \mc{H}_{2;i}) < \infty$, for $i=1,2$.  (Note that the $L_2(Q_1)$ covering number of a class $\{ f(x) \}$ is equal to the $L_2(Q)$ covering number, for any $Q(x,y)$ with $Q_1$ as its $x$-marginal, of $\{ (x,y) \mapsto f(x) \}$.  That is, we can extend each function (class) of functions $f: \mc Z \to \RR$ to functions defined on $\mathcal{Z}^2$, letting $(z_1,z_2) \mapsto f(z_1)$ without changing the entropy.)  By Assumption~\ref{assm:EA-item2}$_3$ 
  this then allows us to conclude that $J_2(1, h^i \mc{G}_{2;i}) < \infty$ (since $J_3(1, h^i \mc{G}_{2;i}) < \infty$) and further, by Lemma~20 of \cite{Nolan:jw}, that $J(1, h^i \PP \mc{G}_{2;i}) < \infty$, $i=1,2$.

  \noindent {\bf Envelopes.} From the proof of Lemma~\ref{lem:3}, we know $\mc{H}_1$ has squared-integrable
  envelope $H_1$ (by the boundedness assumptions on the function classes, and the assumption $Y$ has a finite variance), so an envelope for $h \mc{G}_{2;1}$ is $ G_{2;1}:= K_{\max} H_1$
  and for $h^2 \mc{G}_{2;2}$ is $ G_{2;2}:= a_{\text{width}} K_{\max} H_1$ where $a_{\text{width}} := \max_{a_1,a_2 \in \mc{A}} |a_2-a_1|$, and
  where $K_{\max} := \sup K < \infty$ by assumption.
  These two envelopes have second moment finite (and independent of $n$). Next
  consider minimal envelopes  $K_{2;i}(w)$ for $\PP \mc{G}_{2;i}$, $i=1,2$.  Note that
  \begin{equation*}
    \PP g_2(w, W) =
    \PP h_2(W, w)
  \end{equation*}
  since $\PP h_2(w,W) =0$ for almost every $w$.
  (Recall $h_2$ is defined in \eqref{eq:1003}
  and $g_i$ is defined in \eqref{eq:10003}.)
  And we can see, by the change of variables $u = (a_1-a)/h$, that
  \begin{equation}
    \label{eq:1006}
    \begin{split}
      \MoveEqLeft \PP h_2(W,w_2)\\
      &= \int \big( 1, \frac{a_1-a}{h} \big)^T h^{-1}K \big( \frac{a_1-a}{h} \big)
      \lp \frac{y - \mu(l,a_1)}{\pi(a_1|l)} \wt \pi(a_1 | l_2) + \wt \mu(l_2,
      a_1) \rp\times\\
      &\qquad p(l, a_1, y)
      d\nu(l, y) da_1  \\
      & =\int   (1,u)^T 
      K(u) \lp \frac{y - \mu(l,a + h u)}{\pi(a
        + hu|l)} \wt \pi(a + hu | l_2) + \wt \mu(l_2, a + hu) \rp\times\\
      &\qquad p(l, a+hu, y)
      d\nu(l, y) du 
    \end{split} 
  \end{equation}
  which, by our Assumptions~\ref{assm:EA-item2}
  and \ref{assm:EA}\ref{assm:EA-item4} (and $Y$ having a finite mean),
  is bounded in absolute value above by a constant (independent of $n$).
  We can check 
  that the class of functions $w \mapsto \PP h_2(W,w)$ (and so $w \mapsto \PP g_2(w, W)$)
  has a (vector) envelope $(K_{2;1}, K_{2;2})$ satisfying $\PP K_{2;1}(W)^2 <
  \infty$  and $\PP K_{2;2}(W)^2  < \infty $ (independent of $n$). Next we will
  analyze the two components of \eqref{eq:1114} separately.  The term with the $(a_1-a)/h$ factor will be larger in our analysis so we focus on it.  This term can be written, as described above, as
  \begin{equation}
    \label{eq:10002}
    \int_{\mc A} \lp h^{-2}  n^{-2}
    \sum_{i=1}^n \sum_{j=1}^n
    \wt{h}_{2, \wh \mu, \wh \pi; 2}(W_i,W_j) \rp^2 w(a) da,
  \end{equation}
  where  $ \wt{ h}_{2, \wh \mu, \wh \pi; 2} \in h^2 \mc{H}_{2;2}$.
  (Recall from \eqref{eq:1003} that $h_{2;i}$ depends on $a$ but we notationally suppressed this dependence for simplicity.)
  Considering
  $$   n^{-3/2}
  \sum_{i=1}^n \sum_{j=1}^n
  \wt{h}_{2, \mu, \pi; 2}(W_i,W_j),$$
  and  then taking a sup over $\mu \in \mc{F}_\mu,$ $\pi \in \mc{F}_\pi$, $a \in \mc{A}$,
  and applying \eqref{eq:U-diagonals-Op-h2}
  and Proposition~\ref{lem:100}
  with the class $h^2 \mc{G}_{2;2}$
  and envelopes $G_{2;2}$ and $h^2 K_{2;2}$,
  (recalling \eqref{eq:1114} to go from $\mc{H}_{2;2}$ to $\mc{G}_{2;2}$), 
  we see that
  \eqref{eq:10002}
  is bounded above by
  \begin{equation}
    \label{eq:6}
    O_p( n^{-2} h^{-4} ) + O_p( n^{-1}).
  \end{equation}
  Similarly, we write the other term in \eqref{eq:1114} as
  $$\int_{\mc A} \lp h^{-1} n^{-2} \sum_{i=1}^n \sum_{j=1}^n
  \wt{h}_{2, \wh \mu, \wh \pi; 1}(W_i,W_j) \rp^2 w(a) da,$$ where
  $ \wt{h}_{2, \wh \mu, \wh \pi; 1} \in h \mc{H}_{2;1}$, and (by the same argument) be
  seen to be of smaller order, $O_p(n^{-2} h^{-2} ) + O_p(n^{-1})$.

  Thus,
  multiplying 
  $O_p( n^{-2} h^{-4} ) + O_p( n^{-1})$
  by $n \sqrt{h}$ we see that

  \begin{equation*}
    \sup_{a \in \mc{A}} (n\sqrt{h})  d_{1,1}'(a)^2
    =     O_p( n^{-1} h^{-7/2}) + O_p(h^{1/2}),
  \end{equation*}
  so that   
  \begin{align*}
    \label{eq:4}
    \MoveEqLeft   n\sqrt{h}\int
    d_{1,1}'(a)^2
    w(a)\,da
    =  O_p( n^{-1} h^{-7/2}) + O_p(h^{1/2}),
  \end{align*}
  and
  is thus $O_p(\sqrt{h})$ by
  Assumption~\ref{assm:EA}\ref{assm:EA-item3}. Lastly, in order to derive  the order of
  $n \sqrt{h} \int\{ d_{1,1}(a)\}^2 w(a)\, da$, similar to the proof of
  Lemma~\ref{lem:step3-1}, 
  \begin{align*}
    \MoveEqLeft    n\sqrt{h}\int_{\mc A} d_{1,1}(a)^2w(a)\,da\\ &= n\sqrt{h}\int \lb
                                                          g_{ha}^T\widehat{D}_{ha}^{-1}d'_{1,1}(a)\rb^2w(a)\,da\\
                                                        &= n\sqrt{h}\int \lb
                                                          \ls
                                                          g_{ha}^T\widehat{D}_{ha}^{-1}-(\varpi(a)^{-1},0)+(\varpi(a)^{-1},0)\rs
                                                          d'_{1,1}(a)\rb^2w(a)\,da.
  \end{align*}
  Then
  \begin{align*}
    n\sqrt{h}\int d_{1,1}(a)^2w(a)\,da
    &\le n\sqrt{h}\int 2\lb
      |g_{ha}^T\widehat{D}_{ha}^{-1}-(\varpi(a)^{-1},0)|s'_{1,1}(a)\rb^2w(a)\,da\\
    &\quad+
      n\sqrt{h}\int 2\lb
      |(\varpi(a)^{-1},0)|d'_{1,1}(a)\rb^2w(a)\,da.    
  \end{align*}
  The first term on the last two lines above can be bounded as
  \begin{align*}
    \MoveEqLeft   n\sqrt{h}\int 2\lb
    |g_{ha}^T\widehat{D}_{ha}^{-1}-(\varpi(a)^{-1},0)|d'_{1,1}(a)\rb^2w(a)\,da\\
    &\le 2n\sqrt{h}\int 
      \|g_{ha}^T\widehat{D}_{ha}^{-1}-(\varpi(a)^{-1},0)\|^2_{l_2}\|d'_{1,1}\|_{l_2}^2w(a)\,da;
  \end{align*}
  From the proof of Lemma~\ref{lem:step3-1}, we know $ |g_{ha}^T\widehat{D}_{ha}^{-1}-(\varpi(a)^{-1},0)|$
  is uniformly $o(1)$ a.s., then we have the right hand side of the above inequality
  is $o_p(\sqrt{h})$.
  By
  Assumption~\ref{assm:IA}\ref{assm:IA-item2}, $\varpi(a)$ is bounded below from
  $0$; similarly it is easy to see
  \begin{align*}
    \MoveEqLeft   n\sqrt{h}\int 2\lb
    |(\varpi(a)^{-1},0)|d'_{1,1}(a)\rb^2w(a)\,da\\
    &\le 2n\sqrt{h}\int \|(\varpi(a)^{-1},0)\|_{l_2}^2\|d'_{1,1}(a)\|_{l_2}^2w(a)\,da\\
    &=O_p(\sqrt{h}).
  \end{align*}
  That completes the proof of $n\sqrt{h}\int \{d_{1,1}(a)\}^2w(a)\,da=O_p(\sqrt{h})$.
\end{proof}

\begin{proof}[Proof of Lemma~\ref{lem:4}]
  First we focus on the
  $V_{12, \wh \mu}$ term.
  This is a third degree V-statistic.  
  Start by defining the asymmetric kernel $H$ for a generic $\mu$
  to be
  \begin{equation}
    \label{eq:1}
    H(\bs Z_1, \bs Z_2, \bs Z_3) :=
    \epsilon_3 \int_{\mc A}  \varpi_0^{-2}(a)
    \ls K_{ha}(A_1) \wt \mu (\bs L_2, A_1)  K_h(A_3-a) \rs
    w(a) da.
  \end{equation}
  Recall that $\wt\mu(l_2,a_1) := \mu(l_2,a_1) -  \PP \mu(L,a_1)$.
  Then the $\mu$ term V-statistic is
  $  V_{12, \mu} :=
  n^{-3} \sum_{i,j,k=1}^n H(\bs Z_i, \bs Z_j, \bs Z_k).$
  We begin by analyzing the corresponding U-statistic
  \begin{equation*}
    U_{12, \mu} :=
    n^{-3} \sum_{i,j,k}^n H(\bs Z_i, \bs Z_j, \bs Z_k)
  \end{equation*}
  where the sum is over $(i,j,k)$ with unique coordinates (i.e., over
  the $n(n-1)(n-2)$ ordered choices of $3$ distinct elements of $\{1,\ldots, n\}$).
  Let the symmetrized kernel be
  $G(\bs Z_1, \bs Z_2, \bs Z_3) :=
  6^{-1} \sum_{\sigma} H( \bs Z_{\sigma_1}, \bs Z_{\sigma_2},\bs Z_{\sigma_3})$, 
  where $\sigma$ ranges over the $6$ permutations of $3$ elements.

  One can decompose any $U$-statistic (kernel) into sums of $U$-statistics (kernels) which have certain
  ``degenerate'' structure.  (See
  \citet[pages 177--178]{serfling1980approximation}).)
  We define as follows.
  Let
  \begin{align*}
    &G_{1, \{1\}}( \bs z_1) := \PP^2 H(\bs z_1, \bs Z_2, \bs Z_3),\\
    & G_{1, \{2\}}(\bs z_2) := \PP^2 H(\bs Z_1, \bs z_2, \bs Z_3),\\
    &G_{1, \{3\}}(\bs z_3) := \PP^2 H(\bs Z_1, \bs Z_2, \bs z_3),
  \end{align*}
  where $\PP^i$ denotes the product measure of $\PP$ ($i$ copies).
  Note that
  these three functions are all identically zero.
  It is immediately seen  that $\PP H(\bs Z_1, \bs Z_2, \bs z) = \PP H(\bs z , \bs Z_2, \bs Z_3) = 0$, since $\PP \wt \mu(a, \bs L) = 0$ for any $a$.
  It is also immediately clear that  $\PP H(\bs Z_1, \bs z_2, \bs Z_3) = 0$,
  since $\EE( \overline \epsilon_3 | A_3) = 0$ almost surely.
  Therefore, $G$ is degenerate, in that
  \begin{equation*}
    \PP G( \bs z, \bs Z_2, \bs Z_3) = 0
  \end{equation*}
  for any $z$.
  Next, let
  \begin{align*}
    G_{2, \{1,2\}}(\bs z_1, \bs z_2) & := \PP H (\bs z_1, \bs z_2, \bs Z_3), \\
    G_{2, \{1,3\}}(\bs z_1, \bs z_3) & := \PP H (\bs z_1, \bs Z_2, \bs z_3), \\
    G_{2, \{2,3\}}(\bs z_2, \bs z_3) & := \PP H (\bs Z_1, \bs z_2, \bs z_3). 
  \end{align*}
  Of the above three functions, only
  $G_{2, \{2,3\}}(\bs z_2, \bs z_3)$
  is not identically zero (again, since
  $\EE( \overline \epsilon_3 | A_3) = 0$ and $\PP \wt \mu(a, \bs L) = 0$).
  Finally, let
  \begin{equation}
    \label{eq:G3defn}
    G_3(\bs z_1, \bs z_2, \bs z_3) :=
    H( \bs z_1, \bs z_2, \bs z_3)
    - G_{2, \{2,3\}}(\bs z_2, \bs z_3).
  \end{equation}
  Note that both $G_3$ and $G_{2 , \{2,3\}}$ are maximally degenerate (i.e., averaging over any variable yields the zero function). 
  \begin{mylongform}
    \begin{longform}
      Note that both $G_3$ and $G_{2 , \{2,3\}}$ are degenerate, but to differing extents.  We make a definition here, to distinguish these amounts of degeneracy: we say that a kernel $F^{\star}$ of order $i$ is {\it degenerate of degree j} ($j \le i$) 
      if averaging over any $i-j$ variables (leaving $j$ functional variables) yields an identically zero function; e.g., in the case of a symmetric kernel, we have $\PP F^{\star}(\bs{z}_1, \ldots, \bs{z}_j, \bs{Z}_{j+1}, \ldots, \bs{Z}_{i}) = 0$.

  Note that
  $G_3$ is a
  degree 2 degenerate kernel
  since 
  we have  almost surely that $\PP G_3(\bs Z_1, \bs z_2, \bs z_3) = 0$ (and similarly if we average over $\bs Z_2$ or $\bs Z_3$).
  And note also that $G_{2, \{2,3\}}$ is degree 1 degenerate in that averaging over either of its arguments (yielding one remaining functional variable) yields an identically zero function.

    \end{longform}
  \end{mylongform}

  Now trivially by definition \eqref{eq:G3defn},
  \begin{align*}
    H( \bs z_1, \bs z_2, \bs z_3)
    = G_3( \bs z_1, \bs z_2, \bs z_3)
    + G_{2, \{2,3\}}(\bs z_2, \bs z_3),
  \end{align*}
  meaning that $H$ is decomposed into
  a degree 2 degenerate kernel of degree 3,
  and a degree 1 degenerate kernel of degree 2.

  First we will consider the latter term, $G_{2, \{2,3\}}(\bs z_2, \bs z_3)$.
  We wish to find an envelope for this class of functions, and then can apply a maximal inequality to the sum. 
  Since $\varpi$, $\varpi^{-1}$, and $\wt \mu$ are all uniformly bounded, 
  \begin{align*}
    \varpi_0^{-2}(a)
    \PP \ls K_{ha}(A_1) \wt \mu (\bs L_2, A_1) \rs
    =
    \varpi_0^{-2}(a)
    \int K(u) \wt \mu(a + uh, l_1) \varpi(a + uh) du,
  \end{align*}
  is $O(1)$ (independently of $a$).
  Then,
  plugging the above in to  $G_{2, \{2,3\}}(\bs z_2, \bs z_3)$, we have
  \begin{align*}
    \overline{\epsilon}_3   \int_{\mc A} O(1)  K_h(a_3-a) w(a) da
    = \overline{\epsilon}_3 \int_{(a_3 - \mc A)/h} O(1)  K(u) w(a_3-hu) du.
  \end{align*}
  Thus
  $|\overline{\epsilon}_3| \int_{\RR} O(1)  K(u) (\max w)
  du$
  is an envelope for the class (and is independent of $h$ and $n$).

  Now  we consider entropies.  By
  Assumption~\ref{assm:EA-item2}$_3$, we have
  $J_3(1, \mc{F}_\mu, L_2) < \infty$; the class of functions $H$ under consideration is
  \begin{equation*}
    \mc{H} := \{ (z_1,z_2,z_3) \mapsto \epsilon_3 \varpi_0^{-2} hK_{ha}(a_1) \wt
    \mu(l_2,a_1)K_{ha}(a_3) : \mu \in \mc{F}_\mu, a \in \mc{A} \},
  \end{equation*}
  which by
  an argument almost identical to that in the proof of
  Lemma~\ref{lem:5}, has $J_3(1, \mc{H}, L_2) < \infty$.
  Furthermore, by Lemmas
\ref{lem:integral-decreases-entropy},
\ref{lem:andrews-add},
and \ref{lem:andrews-mult}, we can see that the class of functions $\{ G_{2, \{2,3\}} \}$ and then $\{ G_3 \}$ (using shorthand notation for these classes) also have their corresponding $J_3$ integral finite.

Now, 
since $J_2(1, \{ G_{2, \{2,3\}} \} , L_2) < J_3(1, \{ G_{2, \{2,3\}} \} , L_2) < \infty$,
by Proposition~\ref{prop:1} (alternatively, see Nolan and Pollard (1987)), 
  $n^{-2} \sum_{i\ne j}^n G_{2, \{2,3\}}(\bs Z_i, \bs Z_j) = O_p(n^{-1})$,
  and  so $n^{-3} \sum_{i,j,k}  G_{2, \{2,3\}}(\bs Z_i, \bs Z_j) = O_p(n^{-1})$ (summing over $(i,j,k)$ not equal to each other).

  \medskip

  Now consider  $G_3( \bs z_1, \bs z_2, \bs z_3)$.  We have just seen $J_3(1, \{ G_3 \}, L_2) < \infty$ so we only need to focus on finding the envelope.  Since we  have an envelope for 
  $G_{2, \{2,3\}}(\bs z_2, \bs z_3)$
  we just need an envelope for $H( \bs z_1, \bs z_2, \bs z_3)$.
  By
  the change of variables $u = (a_1-a)/h$
  \begin{align}
    \label{eq:2000}
    \begin{split}
    & H( \bs z_1, \bs z_2, \bs z_3)\\
    &= \overline{\epsilon}_3 \int_{(a_1 - \mc A)/h}
      \varpi_0(a_1 - uh)^{-2} K(u) \wt \mu(a_1, \bs l_2)
      \inv{h} K \lp \frac{a_3-a_1}{h} + u \rp
      w( a_1 - uh) du.
      \end{split}
  \end{align}
  Now
  $K$ has support $[-1,1]$ and so
  \begin{align*}
    K(u)
    \inv{h} K \lp \frac{a_3-a_1}{h} + u \rp
    \le K_{\text{max}}^2 \one_{[-1,1]}(u) \one_{[-2h,2h]}(a_3-a_1)  
  \end{align*}
  since $-h \le a_3-a_1 +uh \le h$ implies $-2h \le a_3-a_1 \le 2h$ for $u \in [-1,1]$.
  Thus, \eqref{eq:2000} is bounded above in absolute value by
  \begin{equation*}
    \begin{split}
      \MoveEqLeft
      \overline{\epsilon}_3 h^{-1} \one_{[-2h,2h]}(a_3-a_1)
      \int_{\RR} \varpi_0^{-2}(a_1-uh) K(u) \wt \mu(a_1-\bs l_2) w(a_1-uh) du \\
      & \le C_1 \overline{\epsilon}_3 h^{-1} \one_{[-2h,2h]}(a_3-a_1)
      =: F_3(\bs z_1 , \bs z_2, \bs z_3),
    \end{split}
  \end{equation*}
  for a constant $C_1$, where we take $F_3$ as our envelope. We see that
  \begin{equation*}
    \PP^3 F_3^2(\bs Z_1, \bs Z_2, \bs Z_3)
    \le C_2 h^{-3/2}
  \end{equation*}
  for a constant $C_2$.
  Therefore,
  by Proposition~\ref{prop:1},
  we have
  $n^{-3} \sum_{i,j,k}^n G_3(\bs Z_i, \bs Z_j, \bs Z_k) = O_p(n^{-3/2} h^{-3/4}) $ (sum over $i,j,k$ not ever equal to each other).


  Finally, the above analysis of $U_{12,\mu}$ ignored the summands in $V_{n,\mu}$
  where an argument is repeated (i.e., $H(\bs Z_1, \bs Z_1, \bs Z_2),$ etc., or $ H(\bs Z_1,\bs Z_1,\bs Z_1)$).  The sums of such terms are very small since there are many fewer terms than $n^3$.  A very simple analysis can show each such sum is not larger than $O_p( n^{-2} h^{-3/2}) = O_p( (n \sqrt{h})^{-1} (nh)^{-1}) $, which is much smaller than we need.

  Analyzing the ``$\pi$ term'',
  $V_{12, \wh \pi}$, 
  is similar to analyzing the $\mu $ term and we do
  not present the details. 
\end{proof}

\section{Lemma~\ref{lem:1}}

The quantity $\PP_n \ls \bs{g}_{ha}(A) K_{ha}(A) \lb \wh{\xi}(\bs{Z}; \wh \pi, \wh \mu) - \xi(\bs Z; \overline \pi, \overline \mu) \rb \rs$, broken into three pieces, is analyzed in the following lemma;
the result
allows to conclude that
$  | \wh \theta_h(a) - \theta(a) |
= O_p(  (nh)^{-1/2} + h^2 + r_n s_n).$

\begin{lemma}
  \label{lem:1}
  Let $\overline{\pi}$ and $\overline{\mu}$ denote fixed functions to which $\wh \pi$ and $\wh \mu$ converge in the sense that $\sup_{\bs z} | \wh \pi - \overline \pi| = o_p(1)$ and $\sup_{\bs z} | \wh \mu - \overline{\mu}| = o_p(1)$.  Let $a \in \mc{A}$ denote a point in the interior of the compact support $\mc A$ of $A$.
  Let Assumption~\ref{assm:IA} hold,
  Assumption \ref{assm:DA} parts \ref{assm:DA-item1}, \ref{assm:DA-item2}, \ref{assm:DA-item3} hold, and
  Assumption  \ref{assm:EA}\ref{assm:EA-item4} holds.
  Assume \ref{assm:EA-item2}$_2$ holds.
  Assume $Y$ has finite variance.
  Assume the conditional density of $\xi(\bs Z; \overline \pi, \overline \mu)$ given $A=a$ is continuous in $a$.
  Assume $h \equiv h_n \to 0$
  and $nh^3 \to \infty$ 
  as $n \to \infty$.  Assume either $\overline \pi = \pi_0$ or $\overline \mu = \mu_0$.
  Let $r_n \equiv r_n(a)$ and $s_n\equiv s_n(a)$ be given by  $\sup_{t: |t-a| \le h} \| \wh \pi(t | \bs L) - \pi(t | \bs L) \|_2 = O_p(r_n)$
  and
  $\sup_{t: |t-a| \le h} \| \wh \mu( \bs L, t) - \mu(\bs L, t) \|_2 = O_p(s_n)$.
  Then
  \begin{equation}
    \label{eq:1}
       \PP_n\ls \bs{g}_{ha}(A) K_{ha}(A)
    \lb \wh{\xi}(\bs{Z}; \wh \pi, \wh \mu) -
    \xi(\bs{Z}; \wh{\pi} , \wh{\mu} )\rb\rs
   =  o_p(1/\sqrt{nh}),
 \end{equation}
 \begin{equation}
   \label{eq:2}
   (\PP_n-\PP)\ls \bs{g}_{ha}(A) K_{ha}(A)
   \lb \xi(\bs{Z}; \wh \pi, \wh \mu) -
   \xi(\bs{Z}; \overline{\pi} , \overline{\mu} )\rb\rs
   =  o_p(1/\sqrt{nh}),
   \text{ and} 
 \end{equation}
 \begin{equation}
   \label{eq:3}
   \PP\ls \bs{g}_{ha}(A) K_{ha}(A)
   \lb \xi(\bs{Z}; \wh \pi, \wh \mu) -
   \xi(\bs{Z}; \overline{\pi} , \overline{\mu} )\rb\rs
   =  O_p(r_ns_n)
 \end{equation}
  as $n \to \infty$.
\end{lemma}


\begin{proof}[Proof of Lemma~\ref{lem:1}]
  {\bf Second term, \eqref{eq:2}}:
  \citet{Kennedy:2017cq} show that
  $$    (\PP_n - \PP) \ls \bs{g}_{ha}(A) K_{ha}(A)
  \lb {\xi}(\bs{Z}; \wh \pi, \wh \mu) -
  \xi(\bs{Z}; \overline{\pi} , \overline{\mu} )\rb\rs
  =  o_p(1/ \sqrt{nh})$$
  (with no hat over the first $\xi$).
  (This follows from their  proof/analysis of their  $R_{n,1}$ in the proof of their Theorem 2
  (page 13 of the Web Appendix
  of \citet{Kennedy:2017cq}).)

  {\bf First term,  \eqref{eq:1}}: The order of the first term,  \eqref{eq:1}, follows from the proof of Lemma~\ref{lem:5}; note that the proof of Lemma~\ref{lem:5} proceeds (initially) for a fixed $a$.  We make a few comments here about the differing assumptions for Lemma~\ref{lem:1}
  and for Lemma~\ref{lem:5} (i.e., for Theorem~\ref{thm:null-normality})
  in
  the context of the proof of Lemma~\ref{lem:5}.

  Lemma~\ref{lem:5}
  relies on the
  assumption that $\pi,\mu, 1/\pi$ are uniformly bounded above,
  which
  both 
  Theorem~\ref{thm:null-normality} and
   Lemma~\ref{lem:1} make
  (Assumption \ref{assm:EA-item2}$_2$).  
  Similarly, both theorems make the assumption that $K$ is bounded above (Assumption~\ref{assm:EA}\ref{assm:EA-item4}). 
  Both theorems assume that $J_m(1, \mc{F}, L_2) < \infty$ for some $m > 1$,
  so $J(1, \mc{F},L_2)<\infty$, for $\mc{F}$ equal to
  $\mc{F}_\mu$,$\mc{F}_\pi$; this latter assumption is all that is needed for
  Lemma~\ref{lem:5} (which calls upon Proposition~\ref{lem:100}).  Both
  theorems assume that $Y$ has a finite variance, as needed by
  Lemma~\ref{lem:5}.

  Finally, by the assumption that $nh^{3} \to \infty$, we see that multiplying \eqref{eq:6} by $nh$ yields an order of $o_p(1)$, so the term on the left of \eqref{eq:1} is $o_p(1/\sqrt{nh}$) (rather than $O_p(1/\sqrt{n})$
  under the stronger assumption of Theorem~\ref{thm:null-normality} that $h$ is of order $n^{-1/5}$) as desired.

  {\bf Third term,  \eqref{eq:3}}: 
  Note 
  $ \PP[ \bs{g}_{ha}(A) K_{ha}(A) \{ \xi(\bs{Z}; \wh \pi, \wh \mu) -
  \xi(\bs{Z}; \overline{\pi} , \overline{\mu} )\}]$ is a vector with $j$th
  element ($j=1,2$) equal to
  \begin{equation*}
    \int \bs{g}_{ha,j}(t) K_{ha}(t)
    \PP    \lb \xi(\bs{Z}; \wh \pi, \wh \mu) -
    \xi(\bs{Z}; \overline{\pi} , \overline{\mu} )|A=t\rb\varpi_0(t)\,dt,
  \end{equation*}
  where $g_{ha,j}(t)=\{(t-a)/h\}^{j-1}$. Note  that
  \begin{align}
    \label{eq:conditional-db}
    \begin{split}
      \MoveEqLeft \PP\lb\xi(\bs{Z};\widehat{\pi},\widehat{\mu})-\xi(\bs{Z};\bar{\pi},\bar{\mu})|A=t\rb\\
      &= \PP\ls
      \lb\mu_0(\bs{L},t)-\widehat{\mu}(\bs{L},t)\rb\lb\frac{\pi_0(t|\bs{L})/\varpi_0(t)}{\widehat{\pi}(t|\bs{L})/\int\widehat{\pi}(t|\bs{l})\,dP(\bs{l})}\rb\rs\\
      &\quad+\int\widehat{\mu}(\bs{l},t)\,dP(\bs{l})-\int\mu_0(\bs{l},t)\,dP(\bs{l})\\
      &=\PP\ls            \lb\mu_0(\bs{L},t)-\widehat{\mu}(\bs{L},t)\rb\lb\frac{\pi_0(t|\bs{L})/\varpi_0(t)}{\widehat{\pi}(t|\bs{L})/\int\widehat{\pi}(t|\bs{l})\,dP(\bs{l})}-1\rb\rs\\
      &=\PP\ls          \lb\mu_0(\bs{L},t)-\widehat{\mu}(\bs{L},t)\rb\lb\frac{\pi_0(t|\bs{L})\int\widehat{\pi}(t|\bs{l})\,dP(\bs{l})-\widehat{\pi}(t|\bs{L})\varpi_0(t)}{\widehat{\pi}(t|\bs{L})\varpi_0(t)}\rb\rs.
    \end{split}
  \end{align}
  Then by further calculation, we see
  \begin{align}
    \begin{split}
      \label{eq:conditional-db-part2}
      &\PP\ls          \lb\mu_0(\bs{L},t)-\widehat{\mu}(\bs{L},t)\rb\lb\frac{\pi_0(t|\bs{L})\int\widehat{\pi}(t|\bs{l})\,dP(\bs{l})-\widehat{\pi}(t|\bs{L})\varpi_0(t)}{\widehat{\pi}(t|\bs{L})\varpi_0(t)}\rb\rs\\
      &=\PP\ls                                                             \lb\mu_0(\bs{L},t)-\widehat{\mu}(\bs{L},t)\rb\lb\frac{\pi_0(t|\bs{L})\int\widehat{\pi}(t|\bs{l})\,dP(\bs{l})-\widehat{\pi}(t|\bs{L})\int\widehat{\pi}(t|\bs{l})\,dP(\bs{l})}{\widehat{\pi}(t|\bs{L})\varpi_0(t)}\rb\rs\\
      &\quad+\PP\ls          \lb\mu_0(\bs{L},t)-\widehat{\mu}(\bs{L},t)\rb\lb\frac{\widehat{\pi}(t|\bs{L})\int\widehat{\pi}(t|\bs{l})\,dP(\bs{l})-\widehat{\pi}(t|\bs{L})\varpi_0(t)}{\widehat{\pi}(t|\bs{L})\varpi_0(t)}\rb\rs\\        &=\frac{\int\widehat{\pi}(t|\bs{l})\,dP(\bs{l})}{\varpi_0(t)}\PP\ls                                                             \lb\mu_0(\bs{L},t)-\widehat{\mu}(\bs{L},t)\rb\lb\frac{\pi_0(t|\bs{L})-\widehat{\pi}(t|\bs{L})}{\widehat{\pi}(t|\bs{L})}\rb\rs\\
      &\quad+\frac{1}{\varpi_0(t)}\PP\lb\widehat{\pi}(t|\bs{L})-\pi_0(t|\bs{L})\rb\PP\lb\mu_0(\bs{L},t)-\widehat{\mu}(\bs{L},t)\rb.
        \end{split}
  \end{align}
  By Assumptions~\ref{assm:IA}\ref{assm:IA-item2},
  \ref{assm:EA-item2},
  and the Cauchy-Schwarz inequality, we see the
  conditional integral
  $\PP\lb\xi(\bs{Z};\widehat{\pi},\widehat{\mu})-\xi(\bs{Z};\bar{\pi},\bar{\mu})|A=t\rb$
  is  bounded by
  \begin{align*}
    O\lp\|\wh\pi(t|\bs{L})-\pi_0(a|\bs{L})\|_2\|\wh\mu(\bs{L},t)-\mu_0(\bs{L},t)\|_2\rp,   
  \end{align*}
  and thus
  \begin{align*}
    \MoveEqLeft       \left\lvert \PP\ls \bs{g}_{ha,j}(A) K_{ha}(A)
    \lb \xi(\bs{Z}; \wh \pi, \wh \mu) -
    \xi(\bs{Z}; \overline{\pi} , \overline{\mu} )\rb\rs  \right\lvert\\
    & =O\lp\int_{\mc{A}}
      \bs{g}_{ha,j}(A)\|\wh\pi(t|\bs{L})-\pi_0(a|\bs{L})\|_2\|\wh\mu(\bs{L},t)-\mu_0(\bs{L},t)\|_2\varpi_0(t)\,dt\rp\\
    &=O_p(r_ns_n).
  \end{align*}
\end{proof}

\section{Proof of main lemmas for Theorem~\ref{thm:null-normality}}
The lemmas proved in this section, combined with several of those from Section~\ref{sec:proofs-u-or-v}, form the backbone of
the proof of Theorem~\ref{thm:null-normality}.

\begin{lemma}
  \label{lem:step2-1}
  Let the assumptions of Theorem~\ref{thm:null-normality} hold. Then
  \begin{equation*}
    \PP_n\{\xi(\bs{Z};\bar{\pi},\bar{\mu})-\xi(\bs{Z};\widehat{\pi},\widehat{\mu})\}
    = O_p\lp 1/\sqrt{n} + s_n^{\infty}r_n^{\infty}\rp,
  \end{equation*}
  as $n\to \infty$.
\end{lemma}
\begin{proof}[Proof of Lemma~\ref{lem:step2-1}]
  We write the quantity of interest as 
  \begin{align*}
    \begin{split}  \MoveEqLeft\PP_n\{\xi(\bs{Z};\bar{\pi},\bar{\mu})-\xi(\bs{Z};\widehat{\pi},\widehat{\mu})\}\\&=(\PP_n-\PP)\{\xi(\bs{Z};\bar{\pi},\bar{\mu})-\xi(\bs{Z};\widehat{\pi},\widehat{\mu})\}+\PP\{\xi(\bs{Z};\bar{\pi},\bar{\mu})-\xi(\bs{Z};\widehat{\pi},\widehat{\mu})\}.
    \end{split}
  \end{align*}
  Now, we apply the concept of stochastic
  equicontinuity to deal with first term on the right side of the previous display, using an argument
  similar to the one used by \cite{Kennedy:2017cq}. 
  Let
  \begin{align}
    \label{eq:def-xi}
    \Xi =(Y\oplus\mc{F}_{\mu})\mc{F}_{\pi}^{-1}\mc{F}_{\varpi}\oplus\mc{F}_m,
  \end{align}
  where $(\mc{F}_{\mu},\mc{F}_{\pi},\mc{F}_{\varpi},\mc{F}_m)$ are the classes of
  functions containing $(\pi,\mu,\varpi,m)$ (Recall $\mc{F}_{\mu}$ and $\mc{F}_{\pi}$ are
  defined in Assumption~\ref{assm:EA-item2}) and
  specifically, the latter two are defined as 
  $\mc{F}_{\varpi} := \{\int\pi(\cdot|\bs{l})\,dP(\bs{l}),\pi\in\mc{F}_{\pi}\}$
  and $\mc{F}_{m} := \{\int\mu(\bs{l},\cdot)\,dP(\bs{l}),\mu\in\mc{F}_{\mu}\}$.   With slight abuse of notation, $Y$
  is the single identity function that takes $\bs{z}=(\bs{l},a,y)$ as input and outputs
  $y$.
  Moreover, we define the operators
  $\mc{F}_1\oplus\mc{F}_2=\{f_1+f_2:f_j\in\mc{F}_j\}$, $\mc{F}^{-1}=\{1/f:f\in \mc{F}\}$
  and  $\mc{F}_1\mc{F}_2=\{f_1f_2:f_j\in\mc{F}_j\}$, for arbitrary function classes
  $\mc{F}$.

  By construction of $\Xi$ and
  Assumption~\ref{assm:EA-item2}, $\xi(\bs{z};\bar\pi,\bar\mu)$ and
  $\xi(\bs{z};\wh\pi,\wh\mu)$ fall in the class $\Xi$.  Moreover, after some rearranging, we can write
  \begin{align}
    \begin{split}
      \MoveEqLeft
      \xi(\bs{Z};\bar{\pi},\bar{\mu})-\xi(\bs{Z};\widehat{\pi},\widehat{\mu})\\
      &=\frac{Y-\bar{\mu}(\bs{L},A)}{\bar{\pi}(A|\bs{L})}\int\bar{\pi}(A|\bs{l})\,dP(\bs{l})+
      \int\bar{\mu}(\bs{l},A)\,dP(\bs{l})\\
      &\quad-\frac{Y-\widehat{\mu}(\bs{L},A)}{\widehat{\pi}(A|\bs{L})}\int\wh{\pi}(A|\bs{l})\,dP(\bs{l})-
      \int\wh{\mu}(\bs{l},A)\,dP(\bs{l})\\
      &=\frac{Y-\bar{\mu}(\bs{L},A)}{\bar{\pi}(A|\bs{L})}\frac{\int\widehat{\pi}(A|\bs{l})\,dP(\bs{l})}{\widehat{\pi}(A|\bs{L})}
      \lb
      \widehat{\pi}(A|\bs{L})-\bar{\pi}(A|\bs{L})\rb\\
      &\quad+\frac{\int\widehat{\pi}(A|\bs{l})\,dP(\bs{l})}{\widehat{\pi}(A|\bs{L})}
      \lb\widehat{\mu}(\bs{L},A)-\bar{\mu}(\bs{L},A)\rb\\
      &\quad
      +\frac{Y-\bar{\mu}(\bs{L},A)}{\bar{\pi}(A|\bs{L})}\lb\int\bar{\pi}(A|\bs{l})\,P(\bs{l})-\int\widehat{\pi}(A|\bs{l})\,dP(\bs{l})\rb\\
      &\quad   +\lb\int\bar{\mu}(\bs{l},A)\,dP(\bs{l})-\int\widehat{\mu}(\bs{l},A)\,P(\bs{l})\rb\\
      &=O_p(\|\widehat{\pi}-\bar{\pi}\|_{\mathcal{Z}}+\|\widehat{\mu}-\bar{\mu}\|_{\mathcal{Z}}).
    \end{split}
  \end{align}
  Thus, by
  Assumption~\ref{assm:EA}\ref{assm:EA-item1}, we have $\|
  \xi(\bs{z};\bar{\pi},\bar{\mu})-\xi(\bs{z};\widehat{\pi},\widehat{\mu})\|_{\mc{Z}}=o_p(\sqrt{h})$. 
  Then by Lemma~\ref{lem:identity-equicont} we have
  \begin{align}
    (\PP_n-\PP)\{\xi(\bs{Z};\bar{\pi},\bar{\mu})-\xi(\bs{Z};\widehat{\pi},\widehat{\mu})\}=o_p(1/\sqrt{n}).
  \end{align}
  From the proof of Lemma~\ref{lem:1}, 
  we have $\PP\lb\xi(\bs{Z};\widehat{\pi},\widehat{\mu})-\xi(\bs{Z};\bar{\pi},\bar{\mu})|A=t\rb$
  is 
  $O(\|\wh\pi(t|\bs{L}-\pi_0(a|\bs{L})\|_2\|\wh\mu(\bs{L},t)-\mu_0(\bs{L},t)\|_2)$.
  Then by Assumption~\ref{assm:DA}\ref{assm:DA-item3}, we have the uniform
  boundedness of $\varpi_0$. Thus we have
  \begin{align}
    \begin{split}
      \MoveEqLeft    \lvert\PP\{\xi(\bs{Z};\bar{\pi},\bar{\mu})-\xi(\bs{Z};\widehat{\pi},\widehat{\mu})\}\rvert    \\&=O\lp\left\lvert\int_{\mc{A}}\|\widehat{\pi}(t|\bs{L})-\pi_0(t|\bs{L})\|_2\|\widehat{\mu}(\bs{L},t)-\mu_0(\bs{L},t)\|_2\,dt\right\rvert\rp\\
      &=O\lp
        \sup_{t\in\mc{A}}\|\widehat{\pi}(t|\bs{L})-\pi_0(t|\bs{L})\|_2\sup_{t\in\mc{A}}\|\widehat{\mu}(\bs{L},t)-\mu_0(\bs{L},t)\|_2\rp\\
      &=O_p \lp s_n^{\infty}r_n^{\infty}\rp.
    \end{split}
  \end{align}
\end{proof}

\begin{lemma}
  \label{lem:step3-1}
  Let the assumptions of Theorem~\ref{thm:null-normality} hold. Then
  \begin{equation*}
    \int\lp g_{ha}^T \widehat {\bs D}_{ha}^{-1} \PP_n \ls g_{ha}(A) K_{ha}(A)
    \lb   \xi( \bs Z; \wh \pi, \wh \mu ) - \xi(\bs Z; \bar{\pi},
    \bar{\mu}) \rb \rs\rp^2 w(a)\,da
  \end{equation*}
  is  $o_p(1/n+(s_n^{\infty}r_n^{\infty})^2)$   as $n\to \infty$.          
\end{lemma}
\begin{proof}[Proof of Lemma~\ref{lem:step3-1}]
  Here  we  
  use a further decomposition $d_{1,2}(a)=R_{n,1,a}+R_{n,2,a}$  (recall the definition of
  $d_{1,2}$ in \eqref{eq:13}) with
  \begin{equation}
    \label{eq:17}
    \begin{split}
    R_{n,1,a} & := g_{ha}^T(a) \widehat {\bs D}_{ha}^{-1} (\PP_n - \PP) \ls g_{ha}(A) K_{ha}(A)
                \lb   \xi( \bs Z; \wh \pi, \wh \mu ) - \xi(\bs Z; \bar \pi, \bar \mu) \rb \rs  \\
    R_{n,2,a} & := g_{ha}^T(a) \widehat {\bs  D}_{ha}^{-1}  \PP \ls  g_{ha}(A) K_{ha}(A)
                \lb   \xi( \bs Z; \wh \pi, \wh \mu ) - \xi(\bs Z; \bar \pi, \bar \mu) \rb\rs.      
    \end{split}
  \end{equation}
  We will bound $\sup_{a \in \mc{A}} R_{n,2,a}$, whereas for $R_{n,1,a}^2 $ we will show a moment bound for each $a$ and then to control the order of magnitude of the integral we control its expectation, by interchanging the integral and the expectation.

  \begin{mylongform}
    \begin{longform}

      For bootstrap if we define errors as centered at $\wh \theta$ rather than at the null estimate $\PP_n \wh \xi$, we need some control of quantities like $\PP_n \wh \theta(A) - \wt \theta(A)$ and $\PP_n W_{ha}{A} (\wh \theta(A) - \wt \theta(A))$.  The v process is ok and the second order remainder above is ok.  It is $R_{n,1,a}$ that is the issue.  It can be done with just another V-process argument but I basically ran out of steam.
      
    \end{longform}
  \end{mylongform}
  
  Lemma~\ref{lem:equiv-kernel-lemma} yields that
  $\sup_{a\in
    \mc{A}}|g_{ha}(a)^T\widehat{D}_{ha}^{-1}-(\varpi_0(a)^{-1},0)|=o(1)$ a.s.\
  for each element.
  Now we let
  \begin{align*}
    R'_{n,1,a} & := (\PP_n - \PP) \ls g_{ha}(A) K_{ha}(A)
                 \lb   \xi( \bs Z; \wh \pi, \wh \mu ) - \xi(\bs Z; \bar \pi, \bar \mu) \rb \rs  \\
    R'_{n,2,a} & :=  \PP \ls  g_{ha}(A) K_{ha}(A)
                 \lb   \xi( \bs Z; \wh \pi, \wh \mu ) - \xi(\bs Z; \bar \pi, \bar \mu) \rb\rs.
  \end{align*}
  We first show $n\sqrt{h}\int (R'_{n,1,a})^2w(a)\,da$ is $o_p(\sqrt{h}) = o_p(1)$. 
  For a measurable class of functions $\mc F$, recall that 
  \begin{equation*}
    J(\delta, \mc F, L_2) :=
    \int_0^\delta     \sup_{Q} \sqrt{1 +
      \log N(\epsilon \| F \|_{Q,2}, \mc F, L_2(Q))} d \epsilon
  \end{equation*}
  where the supremum is taken over all discrete probability measures $Q$ with $\| F\|_{Q,2} > 0$.    \begin{mylongform}
    \begin{longform}
      \KR{Definition of $J$ different from \cite{Kennedy:2017cq}.}

      Definition of $J$ now modified to match with earlier definition in our paper.  for moment bound can have sup on outside but moving sup on inside only increases the value.  
    \end{longform}
  \end{mylongform}
  Here $N(\epsilon, \mc F, L_2(Q))$ is the $L_2(Q)$-$\epsilon$ covering number of $\mc F$ \citep{vanderVaart:1996tf}.
  Also, let $\| \GG_n \|_{\mc F} := \sup_{f \in \mc F} | \GG_n(f) |$.

  \begin{mylongform}
    \begin{longform}
      Let $\Xi$ be as defined on page 10 of the supplement of \cite{Kennedy:2017cq}.
      The class $\Xi$ is shown to be uniformly bounded (under the assumptions of Theorem 2 of  \cite{Kennedy:2017cq}) and to have a finite uniform entropy integral, meaning $J(\infty, \Xi) < \infty$.
    \end{longform}
  \end{mylongform}

  For $\delta > 0$, $j=1,2$, let $\mc{G}_{\delta,j,a,n} \equiv \mc G_{\delta,j,
    a}$ be the class of functions
  \begin{align*}
    \MoveEqLeft h^{-(j-1 + 1/2)} \times\\& \lb \bs{Z} \mapsto \lp A-a\rp^{j-1} K\lp\frac{A-a}{h}\rp (\xi_1(\bs Z)-\xi_2(\bs Z)) :
    \| \xi_1 - \xi_2 \|_{\mc Z} \le \sqrt{h}\delta \rb
  \end{align*}
  for $\bs{Z} = (\bs L,A,Y)$.
  Here $K$ is a kernel satisfying  Assumption~\ref{assm:EA}\ref{assm:EA-item4}.
  For any $a$, $\mc{G}_{\delta,j,a}$ has envelope given by
  \begin{align*}
    G_n(\bs Z)  :=  \lp \frac{A-a}{h}\rp^{j-1} K\lp\frac{A-a}{h}\rp\delta.
  \end{align*}
  By Theorem 2.14.1 of \cite{vanderVaart:1996tf}, we have (suppressing outer expectations related to measurability concerns)
  \begin{equation}
    \label{eq:1}
    \EE \| \GG_n \|^2_{\mc{G}_{\delta,j,a}}
    \lesssim J(1, \mc{G}_{\delta,j,a})^2 \EE | G_n |^2 
  \end{equation}
  where  $\lesssim$ means $\le$ up to a universal constant factor.  
  (To apply the theorem, the class $\mc{G}_{\delta,j,a}$ must be $\PP$-measurable; see  Definition 2.3.3 page 110 of  \cite{vanderVaart:1996tf}.)

  A standard kernel density calculation shows that 
  $\EE G_n^2 \le O(h)\delta$
  independent of $a$ (by  Assumption~\ref{assm:EA}\ref{assm:EA-item4} $\int |u|^{j-1} K(u) du < \infty$ and  $K$ is bounded).
  Thus, 
  we can
  conclude that \eqref{eq:1} is bounded above by $O(h)$, independently of $a$.








  Then for any $\delta,\epsilon>0$, we have
  \begin{align*}
    \MoveEqLeft
    P\lp n\int (R'_{n,1,a})^2w(a)\,da>\epsilon\rp\\&\le   P\lp
    n\int (R'_{n,1,a})^2w(a)\,da>\epsilon,\|\xi(\bs{z};\wh\pi,\wh\mu)-\xi(\bs{z};\bar\pi,\bar\mu)\|_{\mc{Z}}\le
    \sqrt{h}\delta\rp\\
    &\qquad+P\lp\|\xi(\bs{z};\wh\pi,\wh\mu)-\xi(\bs{z};\bar\pi,\bar\mu)\|_{\mc{Z}}> \sqrt{h}\delta\rp,
  \end{align*}
  where the latter probability converges to 0 since
  $\|\xi(\bs{z};\wh\pi,\wh\mu)-\xi(\bs{z};\bar\pi,\bar\mu)\|_{\mc{Z}}=o_p(\sqrt{h})$
  from 
  Step 2. 
  And for the other term, by
  Markov's theorem,
  \begin{align*}
    \MoveEqLeft   P\lp
    n\int (R'_{n,1,a})^2w(a)\,da>\epsilon,\|\xi(\bs{z};\wh\pi,\wh\mu)-\xi(\bs{z};\bar\pi,\bar\mu)\|_{\mc{Z}}
    \le
    \sqrt{h}\delta\rp\\
    &\le \epsilon^{-1}\EE\lb \int n(R'_{n,1,a})^2w(a)\,da\rb      \\
    & = \epsilon^{-1} \int \EE n (R'_{n,1,a})^2 w(a) da
    \le
      \frac{w_{\max}}{\epsilon}\EE\|\bb{G}_n\|^2_{\mc{G}_{\delta,j,a}}
    \le \frac{w_{\max}O(\delta)}{\epsilon},
  \end{align*}
  which goes to 0 by letting $\delta\searrow 0$, and here $w_{\max}:=\sup_{\mc{A}}|w(a)|<\infty$. Thus $n\sqrt{h}\int
  (R'_{n,1,a})^2w(a)\,da$ is $o_p(\sqrt{h})$.

  \begin{mylongform}
    \begin{longform}
       HERE we do NOT have sup control over $a$.  We do use the fact
        of the integral ! 
    \end{longform}
  \end{mylongform}

  To go from $R_{n,1,a}'$ to $R_{n,1,a}$, 
  we use the argument (used above) on $g_{ha}^T \widehat{D}_{ha}^{-1}$ being $O_p(1)$; we have
  \begin{align*}
    \MoveEqLeft
    n\sqrt{h}\int R_{n,1,a}^2w(a)\,da \\&= n\sqrt{h}\int \lb
    g_{ha}^T\widehat{D}_{ha}^{-1}R'_{n,1,a}\rb^2w(a)\,da\\
    &= n\sqrt{h}\int \lb
      \ls
      g_{ha}^T \widehat{D}_{ha}^{-1}-(\varpi_0(a)^{-1},0)+(\varpi_0(a)^{-1},0)\rs
      R'_{n,1,a}\rb^2w(a)\,da\\
    &\le n\sqrt{h}\int 2\lb
      |g_{ha}^T\widehat{D}_{ha}^{-1}-(\varpi_0(a)^{-1},0)|R'_{n,1,a}\rb^2w(a)\,da\\
    &\quad+
      n\sqrt{h}\int 2\lb
      |(\varpi_0(a)^{-1},0)|R'_{n,1,a}\rb^2w(a)\,da.
  \end{align*}
  The first term on the last two lines above can be bounded as
  \begin{align*}
    \MoveEqLeft   n\sqrt{h}\int 2\lb
    |g_{ha}^T\widehat{D}_{ha}^{-1}-(\varpi_0(a)^{-1},0)|R'_{n,1,a}\rb^2w(a)\,da\\
    &\le 2n\sqrt{h}\int 
      \|g_{ha}^T\widehat{D}_{ha}^{-1}-(\varpi_0(a)^{-1},0)\|^2_{l_2}\|R'_{n,1,a}\|_{l_2}^2w(a)\,da,
  \end{align*}
  where $\|x\|_{l_2}:=\sqrt{x^Tx}$  is the $l_2$ norm for vectors on
  $\RR^{d}$. Since we have $ |g_{ha}^T\widehat{D}_{ha}^{-1}-(\varpi_0(a)^{-1},0)|$
  is uniformly $o(1)$ a.s. and each element of $n\sqrt{h}\int(R'_{n,1,a})^2w(a)\,da$ is
  $o_p(\sqrt{h})$, we have the right hand side of the above inequality is $o_p(\sqrt{h})$. By
  Assumption~\ref{assm:IA}\ref{assm:IA-item2}, $\varpi_0(a)$ is bounded below from
  $0$; similarly it is easy to see
  \begin{align*}
    \MoveEqLeft   n\sqrt{h}\int 2\lb
    |(\varpi_0(a)^{-1},0)|R'_{n,1,a}\rb^2w(a)\,da\\
    &\le 2n\sqrt{h}\int \|(\varpi_0(a)^{-1},0)\|_{l_2}^2\|R'_{n,1,a}\|_{l_2}^2w(a)\,da\\
    &=o_p(\sqrt{h}).
  \end{align*}
  So $n\sqrt{h}\int R_{n,1,a}^2w(a)\,da=o_p(\sqrt{h})$.

  Next we show $\int (R'_{n,2,a})^2w(a)\,da$  is $o_p((s_n^{\infty}r_n^{\infty})^2)$. From \eqref{eq:conditional-db}, we have
  \begin{align*}
    \begin{split}
      |R'_{n,2,a}|
      &=O\lp
      \left\lvert\int_{\mc{A}}g_{ha}(t)K_{ha}(t)\|\widehat{\pi}(t|\bs{L})-\pi(t|\bs{L})\|_2\|\widehat{\mu}(\bs{L},t)-\mu(\bs{L},t)\}\|_2\,dt\right\lvert\rb
    \end{split}
  \end{align*}
  We let
  \begin{align*}  I_a&:=\left\lvert\int_{\mc{A}}g_{ha}(t)K_{ha}(t)\|\widehat{\pi}(t|\bs{L})-\pi(t|\bs{L})\|_2\|\widehat{\mu}(\bs{L},t)-\mu(\bs{L},t)\|_2\,dt\right\lvert
  \end{align*}
  We  show $\int_{\mc{A}}I_a^2w(a)\,da$ is $o_p((n\sqrt{h})^{-1})$.  By
  Assumption~\ref{assm:EA}\ref{assm:DA-item4}, we know $K$ has bounded support,
  and thus
  \begin{align*}
    I_a&=O\lp\int_{\mc{A}}K_{ha}(t)\|\widehat{\pi}(t|\bs{L})-\pi(t|\bs{L})\|_2\|\widehat{\mu}(\bs{L},t)
         -\mu(\bs{L},t)\|_2\,dt\rp\\
       &=O\lp\sup_{t\in\mc{A}}\|\widehat{\pi}(t|\bs{L})-\pi(t|\bs{L})\|_2\sup_{t\in\mc{A}}\|\widehat{\mu}(\bs{L},t)
         -\mu(\bs{L},t)\|_2\rp\\
    &=O_p\lp s_n^{\infty}r_n^{\infty}\rp.
  \end{align*}
  With a similar argument as for $n\sqrt{h}\int R_{n,1,a}^2w(a)\,da$, we   see
  the integral
  $\int R_{n,2,a}^2w(a)\,da=o_p((s_n^{\infty}r_n^{\infty})^2)$. Then by the Cauchy-Schwarz inequality, it's
  easily seen that $\int \{d_{1,2}(a)\}^2w(a)\,da=o_p(1/n+(s_n^{\infty}r_n^{\infty})^2)$.
\end{proof}

\bigskip

\begin{lemma}
  \label{lem:d2d3}
  Let the assumptions of Theorem~\ref{thm:null-normality} hold. Let $D_2(a), D_3$ be
  as defined in \eqref{eq:Tn-decompose}. Then
  $\int_{\mc A} D_2(a) w(a) \, da = O_p(n^{-1/2})$, and so
  \begin{equation}
    \label{eq:8}
    n\sqrt{h}\int_{\mc{A}}D_2(a)D_3w(a)\,da=o_p(1)\text{ as }n\to \infty.
  \end{equation}
\end{lemma}
\begin{proof}[Proof of Lemma~\ref{lem:d2d3}]
  In step 2 of the proof of Theorem~\ref{thm:null-normality}, we have
  seen that $D_3=o_p(1/\sqrt{nh^{1/2}})$. So \eqref{eq:8} follows from
  $\int_{\mc A} D_2(a) w(a) \, da = O_p(n^{-1/2})$.  We first write $D_2(a)$
  as
  \begin{align}
    D_2(a)=\tilde{\theta}_h(a)-\PP_n\bar{\xi}=\tilde{\theta}_h(a)-\PP\bar{\xi}+\PP\bar{\xi}-\PP_n\bar{\xi}.
  \end{align}
  By the Central Limit Theorem, it is easily seen that
  $\PP\bar{\xi}-\PP_n\bar{\xi}=O_p(1/\sqrt{n})$. 
  By
  Lemma~\ref{lem:equiv-kernel-lemma},
  we can
  write
  \begin{align*}    
    \tilde{\theta}_h(a)-\PP\bar{\xi}&= g^T_{ha}\widehat{\bs{D}}^{-1}_{ha}\PP_n\ls
                                      g_{ha}(A)K_{ha}(A)\lb\bar{\xi}(\bs{Z};\bar{\pi},\bar{\mu})-\PP\bar{\xi}\rb\rs\\
                                    &= \frac{1}{n}\sum_{i=1}^n\ls(\varpi_0(a)^{-1},0)+o_p(1)\rs\ls
                                      g_{ha}(A_i)K_{ha}(A_i)\lb\bar{\xi}(\bs{Z}_i;\bar{\pi},\bar{\mu})-\PP\bar{\xi}\rb\rs\\
                                    &=\frac{1}{n}\sum_{i=1}^n\frac{1}{h\varpi_0(a)}K\lp\frac{A_i-a}{h}\rp\lp \bar{\xi}_i-\PP\bar{\xi}\rp\lb
                                      1+o_p(1)\rb
  \end{align*}
  uniformly for all $a\in\mc{A}$. So we can focus on integrating the dominating term
  \begin{align*}
    \frac{1}{n}\sum_{i=1}^n\frac{1}{h\varpi_0(a)}K\lp\frac{A_i-a}{h}\rp\lp
    \bar{\xi}_i-\PP\bar{\xi}\rp,
  \end{align*}
  and we have
  \begin{align*}
    \MoveEqLeft  \int_{\mc{A}} \frac{1}{n}\sum_{i=1}^n\frac{1}{h\varpi_0(a)}K\lp\frac{A_i-a}{h}\rp\lp
    \bar{\xi}_i-\PP\bar{\xi}\rp w(a)\,da\\&=\frac{1}{n}\sum_{i=1}^n\int_{\mc{A}}\frac{w(a)}{h\varpi_0(a)}K\lp\frac{A_i-a}{h}\rp\,da\lp
    \bar{\xi}_i-\PP\bar{\xi}\rp.
  \end{align*}
  By a change of variables ($a=A_i+hu$), we can evaluate the integral as
  \begin{align*}
    \int_{\mc{A}}\frac{w(a)}{h\varpi_0(a)}K\lp\frac{A_i-a}{h}\rp\,da=\int \frac{w(A_i+hu)}{\varpi_0(A_i+hu)}K(u)\,du.
  \end{align*}
  By Taylor expansion, we can write $w(A_i+hu)/\varpi_0(A_i+hu)$ as 
  \begin{align*}
    \frac{w(A_i)}{\varpi_0(A_i)}+\frac{w'(A_i+\eta
    uh)\varpi_0(A_i+\eta uh)-w(A_i+\eta
    uh)\varpi_0'(A_i+\eta uh)}{\{\varpi_0(A_i+\eta uh)\}^2}uh,
  \end{align*}
  where $0<\eta<1$ is some constant depending on $u$. By
  assumption~\ref{assm:DA}\ref{assm:DA-item2} and our assumption on $w(a)$, we
  know $\varpi_0(a)$ is bounded below from 0 and bounded above, and $\varpi'$ and $w'$
  are both bounded.
  \begin{align*}
    \int_{\mc{A}} \frac{w(a)}{h\varpi_0(a)}K\lp \frac{A_i-a}{h}\rp\,da
    =\frac{w(A_i)}{\varpi_0(A_i)}+O(h),
  \end{align*}
  uniformly for all $A_i$. Then by central limit theorem,
  \begin{align*}
    \frac{1}{n}\sum_{i=1}^n\int_{\mc{A}}\frac{w(a)}{h\varpi_0(a)}K\lp\frac{A_i-a}{h}\rp\,da\lp
    \bar{\xi}_i-\PP\bar{\xi}\rp=O_p(1/\sqrt{n}), 
  \end{align*}
  and thus $\int D_2(a)w(a)\,da=O_p(1/\sqrt{n}))$. So we have $n\sqrt{h}\int
  D_2(a)D_3 w(a)\,da=o_p(h^{1/4})=o_p(1)$.
\end{proof}


\bigskip

\begin{lemma}
  \label{lem:step5-decompose-term2}
  Let the assumptions of Theorem~\ref{thm:null-normality} hold. Then
  \begin{equation*}
    \int    \frac{1}{\varpi_0^2(a)} \PP \ls K_{ha}(A)
    \lb   \xi( \bs Z; \wh \pi, \wh \mu ) -  \xi(\bs Z;  \bar \pi,
    \bar \mu) \rb\rs\frac{1}{n}\sum_{i=1}^n
    K_h\lp A_i-a\rp
    \bar\epsilon_i w(a)\,da
  \end{equation*}
  is $o_p((n\sqrt{h})^{-1})$.
\end{lemma}
\begin{proof}[Proof of Lemma~\ref{lem:step5-decompose-term2}]
  Recall
  the definition of the function class $\Xi$ of functions $\xi(\cdot; \pi, \mu)$:
  \begin{equation}
    \label{eq:26}
    \Xi :=  (Y + \mc{F}_\mu) \mc{F}^{-1}_\pi \mc{F}_\varpi + \mc{F}_m,
  \end{equation}
  where $Y$ is shorthand for the single function outputting $Y$ from $\bs Z$.
  Thus $\Xi$ is indexed by the above function classes. 
  Define also the shifted class $\Xi_{1} := \lb \xi - \overline{\xi} : \xi \in \Xi\rb $.  Now for functions $\pi \in \mc{F}_\pi, \mu \in \mc{F}_\mu$,
  let
  \begin{equation}
    \label{eq:27}
    \begin{split}
      \psi_{\mu,\pi}(t) & :=
      \frac{\PP\pi(t|\bs{L}))}{\PP\pi_0(t|\bs{L})} \PP \ls \lb \mu_0(\bs L, t) - \mu(\bs L, t) \rb
      \lb \frac{\pi_0(t|\bs L) - \pi(t | \bs L)}{ \pi(t | \bs L)} \rb \rs \\
      & \quad + \inv{\PP\pi_0(t|\bs{L})} \PP  \lb \pi(t | \bs L) - \pi_0( t | \bs L) \rb
      \PP \lb \mu_0(\bs L, t) - \mu(\bs L, t) \rb.
    \end{split}
  \end{equation}
  (regarding $\varpi$ as dependent on $\pi$).
  Note that $\mc{F}_\psi:=\{ \psi_{\pi, \mu} : \pi \in \mc{F}_\pi, \mu \in \mc{F}_\mu \}$ is uniformly bounded, and also has finite (independent of $n$) uniform entropy integral.  This latter statement follows because, first, from
  the properties discussed in the previous paragraph (i.e., by
  \cite[Theorem 2.10.20]{vanderVaart:1996tf}) the classes of functions
  $ (\bs l, t) \mapsto \frac{\varpi(t)}{\varpi_0(t)}  \lb \mu_0(\bs l, t) - \mu(\bs l, t) \rb
  \lb \frac{\pi_0(t|\bs l) - \pi(t | \bs l)}{ \pi(t | \bs l)} \rb$,
  as well as 
  $ (\bs l, t) \mapsto \inv{\varpi_0(t)}   \lb \pi(t | \bs l) - \pi_0( t | \bs l) \rb$,
  and $ (\bs l, t) \mapsto \lb \mu_0(\bs l, t) - \mu(\bs l, t) \rb$, all have finite uniform entropy integral.
  And then, second, applying $\PP$ (with $t$ fixed and $\bs L$ random) cannot
  increase the uniform entropy integral by
  Lemma~\ref{lem:integral-decreases-entropy}. Thus,
  $ t \mapsto \PP \ls \frac{\varpi(t)}{\varpi_0(t)}  \lb \mu_0(\bs L, t) - \mu(\bs L, t) \rb
  \lb \frac{\pi_0(t|\bs L) - \pi(t | \bs L)}{ \pi(t | \bs L)} \rb \rs$,
  $ t \mapsto   \inv{\varpi_0(t)}  \PP \lb \pi(t | \bs L) - \pi_0( t | \bs L) \rb $
  and $ t \mapsto \PP \lb \mu_0(\bs L, t) - \mu(\bs L, t) \rb$, all have finite uniform entropy integral;
  and then by the preservation properties  of the uniform entropy integral (i.e., by
  \cite[Theorem 2.10.20]{vanderVaart:1996tf}), 
  the class of functions $\{ \psi_{\mu,\pi}$, $\mu\in \mc{F}_\mu$, $\pi \in \mc{F}_\pi \}$, given in \eqref{eq:27}, has a finite uniform entropy integral.  It can also be checked to be uniformly bounded.

  For a sequence $\delta_n > 0$, define another class  
  \begin{equation}
    \label{eq:24}
    \KXI :=
    \lb
    \begin{array}{l}
      f(a_1) := \varpi_0(a_1)^{-1} \int K_h(a_2-a_1) \wt f(l,a_2,y)  d\PP(l,a_2,y) ;
      \wt f \in \Xi_1 , \\
      \qquad f(a_1) = \varpi_0(a_1)^{-1} \int_{\mc A}  K_h(a_2-a_1) \psi_{\mu,\pi}(a_2)
      \varpi_0(a_2) da_2,
      \| f  \|_{\mc A} \le \delta_n    
    \end{array}
    \rb
  \end{equation}
  where recall that functions in $\Xi$ (and so $\Xi_1$) are indexed by the classes
  given in \eqref{eq:26}. Define a third class
  \begin{equation*}
    \KKXI :=
    \lb
    f(A) =
    \int
    K_h(A - a_1)  \wt f(a_1) da_1
    ; \wt f \in \KXI
    \rb.
  \end{equation*}
  And, for a probability measure $Q$ on $\mc A$, let $\wt Q$ be a measure on $\mc A$ given by
  \begin{equation}
    \label{eq:28}
    \wt Q \wt \psi :=
    Q \lp 
    \iint  K_h(A - a_1) K_h(a_2-a_1) \wt \psi(a_2)  da_2  da_1 
    \rp.
  \end{equation}
  (The properties of a measure, e.g.\ additivity, can be trivially verified given the definition as an integral.)

  For a function $f \in \KKXI$, we would like to use Jensen's inequality to show $Q f^2$ is bounded above
  by (a constant times) $\wt Q \wt{\psi}^2$ for a $\wt{\psi} \in \mc{F}_\psi$.
  Any $Qf^2$ is of the form
  \begin{align*}
    \MoveEqLeft
    E_Q \lp 
    \iint  K_h(A - a_1) K_h(a_2-a_1) \wt f(l,a_2, y)
    \varpi_0^{-1}(a_1) d\PP(l,a_2,y)  da_1 
    \rp^2
  \end{align*}
  which equals
  \begin{align*}
    \MoveEqLeft  E_Q \lp 
    \iint  \varpi^{-1}(a_1) K_h(A - a_1) K_h(a_2-a_1) \psi_{\mu,\pi}(a_2)
    \varpi_0^{-1}(a_2)   da_2  da_1
    \rp^2 \\
    & \le \varpi^{-2}_{\text{min}}
      E_Q \lp 
      \iint   K_h(A - a_1) K_h(a_2-a_1) \psi^2_{\mu,\pi}(a_2)
      da_2  da_1
      \rp
      =  \varpi^{-2}_{\text{min}}  \wt{Q} \psi^2_{\mu,\pi}
  \end{align*}
  where the inequality is by applying Jensen's inequality twice.  (Since both kernels integrate to $1$, Jensen's inequality applies.)


  It is immediately clear that $F \equiv \delta_n $ is a (constant) envelope for the class $\KKXI$. (Recall $\varpi_{\max} := \max \varpi$ and $\varpi_{\min} := \min \varpi$.)
  Then, with $\wt F$ being the envelope for $\mc{F}_\psi$, we have
  \begin{equation*}
    N \lp \epsilon  \delta_n /\varpi_{\min}^{2} ,
    \KKXI, L_2(Q) \rp
    \le N \lp \epsilon \| \wt F \|_{\wt Q,2} , \mc{F}_{\psi}, L_2(\wt Q) \rp.
  \end{equation*}
  Thus we can take a sup on the left over all measures $Q$ and this is upper bounded by a sup on the right over all measures $\wt Q$. 
  (We take sups over all measures, not just all discrete measures).  Thus by change of variables 
  \begin{align*}
    J(\infty 
    ) &:=
        \int_0^\infty  \sup_Q  \sqrt{ \log N(\epsilon \| F \|_{Q,2} , \KKXI, L_2(Q))} d \epsilon\\
      &\le  
        \int_0^\infty  \sup_Q \sqrt{ \log N(\epsilon \| \wt F \|_{Q,2} , \mc{F}_\psi, L_2(Q))} d \epsilon
        < \infty
  \end{align*}
  where the second integral does not depend on $n$ and is finite as discussed above.

  Finally, note that the entropy for $\KKXI$ equals the entropy for $e \times \KKXI$ where $e$ stands for the singleton function giving $\epsilon$;
  and this class has envelope $|e| \delta_n$ with $L_2(\PP)$-integral of order $\delta_n$.
  Thus, applying Theorem 2.14.1 of \cite{vanderVaart:1996tf}
  (see also \cite{vanderVaart:2011ju}) to the previous display (and the order $\delta_n$ envelope) we have
  \begin{equation}
    \label{eq:112}
    \PP \| \GG_n \|_{e \times \KKXI}
    \le C_1  J(\infty) \delta_n  \le C_2 \delta_n
  \end{equation}
  where $\GG_n = \sqrt{n} ( \PP_n - \PP) = \sqrt{n} \PP_n$ since $E (\epsilon | A) = 0$.
  and $C_i$ are constants (depending on $\PP$ and the classes/methods for $\pi, \mu$.).

  Finally, we will apply \eqref{eq:112}  to  see that
  \begin{equation}
    \label{eq:3}
    n\sqrt{h}\int
    \frac{\frac{1}{n}\sum_{i=1}^nK_h(A_i-a)\bar\epsilon_i}{\varpi_0(a)}\frac{\PP\lb
      K_h(A-a)(\xi(\cdot; \wh \pi, \wh \mu) -\bar{\xi})\rb}{\varpi_0(a)}\,da
  \end{equation}
  is $O_p(\delta_n \sqrt{h n } )$.


  This is because by  \eqref{eq:conditional-db} and \eqref{eq:conditional-db-part2} we have
  \begin{equation}
    \label{eq:6}
    \PP \lb K_h(A-a) (\xi( \cdot; \wh \pi , \wh \mu) - \bar{\xi}) \rb
    =
    \int K_h(a_2 - a) \psi_{\wh{\pi}, \wh{\mu}}(a_2) \varpi_0(a_2) da_2,
  \end{equation}

  \begin{align}
    \begin{split}
      \MoveEqLeft\PP \lb K_h(A-a) (\xi( \cdot; \wh \pi , \wh \mu) -
      \bar{\xi}) \rb\\
      &= \int K_h(t - a) \PP \lb \xi( \cdot; \wh \pi , \wh \mu) -
        \bar{\xi}|A=t \rb \varpi_0(t) dt\\
      &=     \int K_h(t - a) \psi_{\wh{\pi}, \wh{\mu}}(t) \varpi_0(t)\,  dt,
    \end{split}
  \end{align}
  and so by applying the Cauchy Schwartz inequality, $\| \psi_{\wh \pi, \wh \mu}
  \|_{\mc A} \le \delta_n$ and so the sup over $a \in \mc{A}$ of \eqref{eq:6} is
  also bounded above by $\delta_n$. Now by
  Assumptions~\ref{assm:IA}\ref{assm:IA-item2} and 
  \ref{assm:EA-item2}, we have
  \begin{align}
    \psi_{\wh\mu,\wh\pi}(t) & =
                              \frac{\PP\pi(t|\bs{L}))}{\PP\pi_0(t|\bs{L})} \PP \ls \lb\wh \mu_0(\bs L, t) - \wh\mu(\bs L, t) \rb
                              \lb \frac{\pi_0(t|\bs L) - \wh\pi(t | \bs L)}{ \pi(t | \bs L)} \rb \rs \\
                            & \quad + \inv{\PP\pi_0(t|\bs{L})} \PP  \lb \wh\pi(t | \bs L) - \pi_0( t | \bs L) \rb
                              \PP \lb \mu_0(\bs L, t) -\wh \mu(\bs L, t) \rb\\
                            &= O\lp \| \mu_0(\bs L, t) - \wh\mu(\bs L, t)\|_2
                              \| \pi_0(t|\bs L) - \wh\pi(t | \bs L)\|_2\rp\\
                            &=O(r_n^{\infty}s_n^{\infty}).
  \end{align}
Substituting $\delta_n$ with $r_n^{\infty}s_n^{\infty}$ yields \eqref{eq:3} is 
$o_p((n\sqrt{h})^{-1})$ because  $s^{\infty}_nr^{\infty}_n
    = o\{(n\sqrt{h})^{-1/2}\}$ from Assumption~\ref{assm:EA}~\ref{assm:EA-item4}.

\end{proof}

\section{Proof of Theorem~\ref{thm:alternative} }
\begin{proof}[Proof of Theorem~\ref{thm:alternative}]
  Recall that  we can write our test statistic as
  \begin{align*}
    T_n = n\sqrt{h}\int _{\mc{A}}\lb D_1(a)+D_2(a)+D_3\rb^2 w(a)\,da.
  \end{align*}
  Under the stated  alternative, by slightly modifying the proof of Theorem~2.1 of
  \citet{alcala1999goodness}, we have
  \begin{align*}
    n\sqrt{h}\int\lb D_2(a)\rb^2w(a)\,da = b_{0h}+\delta_n^2b_{1h}+o_p(\delta_n^2),
  \end{align*}
  where
  \begin{align*}
    b_{1h}&=\int_{\mc{A}}\ls K_{h}\ast g(a)\rs ^2w(a)\,da.
  \end{align*}

  From step 3 of the proof of Theorem~\ref{thm:null-normality} we know
  that
  
  \begin{align}
    \int\lb D_1(a)\rb^2w(a)\,da = O\lp\frac{1}{n}+\lb r_n^{\infty} s_n^{\infty}   \rb^2\rp,
  \end{align}
  regardless of the validity of null hypothesis.
  then by Cauchy-Schwart, we have
  \begin{align}
    \MoveEqLeft    n\sqrt{h}\int D_1(a)D_2(a)w(a)\,da\\&\le \sqrt{n^2h\int\lb D_1(a)\rb^2
    w(a)\,da \int\lb D_2(a)\rb^2 w(a)\,da}  \\
    &=\sqrt{n^2hO\lp\frac{1}{n}+\lb r_n^{\infty} s_n^{\infty}   \rb^2\rp\lb
      b_{0h}+\delta_n^2b_{1h}+o_p(\delta_n^2)\rb/(n\sqrt{h})}\\
    &=\sqrt{\lp O(1)+n\{r_n^{\infty} s_n^{\infty}\}^2\rp\lb
      O(1)+\delta_n^2\sqrt{h}b_{1h}+o_p(\sqrt{h}\delta_n^2)\rb}.
  \end{align}
  It's easily seen that
  \begin{align}
    O(1)\lb O(1)+\delta_n^2\sqrt{h}b_{1h}+o_p(\sqrt{h}\delta_n^2)\rb =o_p(\delta_n^2)
  \end{align}
  For
  the other term,
  \begin{align}
    \MoveEqLeft
    n\{r_n^{\infty} s_n^{\infty}\}^2\lb
    O(1)+\delta_n^2\sqrt{h}b_{1h}+o_p(\sqrt{h}\delta_n^2)\rb \\
    &= O\lp n\{r_n^{\infty} s_n^{\infty}\}^2+n\{r_n^{\infty} s_n^{\infty}\}^2\delta_n^2\sqrt{h}b_{1h}\rp.
  \end{align}
  From Assumption~\ref{assm:EA}\ref{assm:EA-item5},  We know $\{r_n^{\infty}
  s_n^{\infty}\}^2 = o(1/(n\sqrt{h}))$, so the second above $O(n\{r_n^{\infty}
  s_n^{\infty}\}^2\delta_n^2\sqrt{h}b_{1h})$ is  $o(\delta_n^2)$. Moreover,
  $O( n\{r_n^{\infty} s_n^{\infty}\}^2) = o(\sqrt{h}) = o(n^{1/10})$ by
  Assumption~\ref{assm:EA}\ref{assm:EA-item1}, thus  $O(
  n\{r_n^{\infty} s_n^{\infty}\}^2) = o(\delta_n^4)$. Then  we have the
  cross-product term
  $n\sqrt{h}\int D_1(a)D_2(a)w(a)\,da = o(\delta^2_n) $. With a
  similar argument, 
  from step 1 of the proof of Theorem~\ref{thm:null-normality}, we
  have
  \begin{align}
    \int\lb D_3(a)\rb^2w(a)\,da = O\lp\frac{1}{n}+\lb r_n^{\infty} s_n^{\infty}   \rb^2\rp
  \end{align}
  under the alternative model stated and then
  $n\sqrt{h}\int D_2(a)D_3(a)w(a)\,da = o(\delta^2_n) $. Finally we have 
  $T_n = b_{0h}+\delta_n^2b_{1h}+o_p(\delta_n^2)$ and it's straight forward
  to see
  \begin{align*}
    P(T_n>t^{\ast}_{n,1-\alpha})\to 1,
  \end{align*}
  as n$\to \infty$.  

\end{proof}


\section{Basic local polynomial estimator lemmas}
\label{sec:further-lemmas}


Here we provide  lemmas that are standard in the analysis of local polynomial estimators. 

\begin{lemma}
  \label{lem:equiv-kernel-lemma}
  Let the assumptions of Theorem~\ref{thm:null-normality} hold.  Recall
  that   $\wh{\bs{D}}_{ha} = \PP_n\{g_{ha}(A)K_{ha}(A)g_{ha}^T(A)\}$ 
  and $g^T_{ha}(A) = (1, (A-a)/h)$.  Then
  $\sup_{a\in \mc{A}}|g_{ha}(a)^T\widehat{\bs
    D}_{ha}^{-1}-(\varpi_0(a)^{-1},0)|=o(1)$ almost surely.
  And thus $g_{ha}(a)^T \wh{\bs{D}}^{-1}_{ha} g_{ha}(t)$ converges to $\varpi_0^{-1}(a)$ almost surely for any $t$.
\end{lemma}
\begin{proof}
  Using Assumptions~\ref{assm:EA}\ref{assm:EA-item3} and
  \ref{assm:EA}\ref{assm:EA-item4}, by Section 4.1 of
  \cite{dony2006uniform}, we have
  \begin{align*}
    \sup_{a\in\mc{A}}\left\lvert\PP_n\lb
    K_{ha}(A)\lp\frac{A-a}{h}\rp^j\rb-\varpi_0(a)\int (-u)^jK(u)\,du\right\rvert\to
    0,\qquad \text{a.s.},
  \end{align*}
  for $j = 0,1,2$. Thus $\widehat{D}_{ha}$ converges to
  $\text{diag}\{ \varpi_0(a), \varpi_0(a)\int u^2K(u)\,du\}$ uniformly almost
  surely over all elements.
  Here $\text{diag}(c_1,c_2)$ is
  a $2\times 2$ diagonal matrix with elements $c_1$ and $c_2$ on the
  diagonal. 
  Then by
  Assumption~\ref{assm:IA}\ref{assm:IA-item2} and
  \ref{assm:DA}\ref{assm:DA-item2}, we have that
  $\sup_{\mc{A}}\varpi_0(a)<\infty$ and
  $\inf_{\mc{A}}\varpi_0(a)\ge\pi_{\min}>0$, so
  $\sup_{a\in
    \mc{A}}|\widehat{D}_{ha}^{-1}-\varpi_0(a)^{-1}\text{diag}\{1,1/\int
  u^2K(u)\,du\}|=o(1)$ a.s.  for each element.
  Thus
  $$\sup_{a\in
    \mc{A}}|g_{ha}(a)^T\widehat{D}_{ha}^{-1}-(\varpi_0(a)^{-1},0)|=o(1)$$
  a.s.\ as $n \to \infty$ for each element.
  \begin{mylongform}
    \begin{longform}
      Because for a invertible $2\times 2$ symmetric matrix, we have
      \begin{align*}
        \begin{bmatrix}
          a&b\\
          b&c
        \end{bmatrix}^{-1}
             = \frac{1}{ac-b^2}      \begin{bmatrix}
               c&-b\\
               -b&a
             \end{bmatrix}.
      \end{align*}
      And for $f_n(x)\to f(x)$, $g_n(x)\to g(x)$, we have
      \begin{itemize}
      \item
        \begin{align*}
          \frac{1}{f_n(x)} = \frac{1}{f(x)}-\frac{1}{f(x)+\delta_x(f_n(x)-f(x))}(f_n(x)-f(x)),
        \end{align*}
        where $\delta_x\in (0,1)$.
      \item
        \begin{align*}
          f^2_n(x) = f^2(x)+2(f(x)+\delta_x(f_n(x)-f(x)))(f_n(x)-f(x)).
        \end{align*}
      \item
        \begin{align*}
          f_n(x)g_n(x) =f(x)g(x)+ f_n(x)(g_n(x)-g(x))+g(x)(f_n(x)-f(x))
        \end{align*}
      \end{itemize}
      So matrix inverse maintains uniform convergence if each element in the limit is
      bounded above from infinity and below from 0.
    \end{longform}
  \end{mylongform}
\end{proof}

\noindent
The following is a standard representation of a local polynomial estimator.

\begin{lemma}[\cite{Fan:1996tk}, page 63]
  \label{lem:LPE-basic-rep}
  Consider the local linear estimator at a point $a$ based on data
  $(A_1, Y_1), \ldots, (A_n,Y_n)$ and kernel $K$, given as
  \begin{equation*}
    \widehat{\beta}_h(a)
    =\argmin_{\beta\in \RR^2}\QQ_n\ls K_{ha}(A)
    \lb Y - g_{ha}(A)^T\beta\rb^2 \rs,
  \end{equation*}
  with $K_{ha}(t)=h^{-1}K\{(t-a)/h\}$ and $g_{ha}(t)=(1,\frac{t-a}{h})^T$
  (and $\QQ_n$ being the empirical distribution of the $(A_i,Y_i)$
  observations). The
  first coordinate
  of $\wh{\beta}_h$ (the function estimator) can be written as
  \begin{equation*}
    \wh{\beta}_{h}
    = n^{-1} 
    \sum_{i=1}^n W_{h} \lp A_i-a \rp Y_i
    = \QQ_n W_h(A-a) Y,
  \end{equation*}
  where 
  \begin{equation}
    \label{eq:7}
    W_{h}(t) := g_{h,0}^T(0) \wh{\bs{D}}_{ha}^{-1} g_{ha}(t) K_h(t),
  \end{equation}
  and where  $\wh{\bs{D}}_{ha} = \QQ_n\{g_{ha}(A)K_{ha}(A)g_{ha}^T(A)\}$.
\end{lemma}

\section{Extension for testing for a treatment effect modifier}
\label{sec:extens-test-finite}

In this section, we briefly introduce an extension of our proposed test to a scenario
where we are interested in testing whether a  covariate
is
a treatment effect modifier or not, meaning the causal treatment effect is different conditional on the covariate value than it is unconditionally.
The following provides a mapping or pseudo-outcome that has the double robustness property for assessing treatment effect modifiers.

\begin{theorem}
  \label{thm:conditional-dr-mapping}
  Let Assumption~\ref{assm:IA} %
  hold.  Then the following
  mapping 
  \begin{align}
    \label{eq:conditional-dr-mapping} 
    \phi(\bs{Z};\pi,\mu) = \frac{Y-\mu(\bs{L},A)}{\pi(A|\bs{L})}
    \int_{\mc{A}}\pi(A|\bs{l})\,dP(\bs{l}|L_p)
    +\int_{\mc{A}}\mu(\bs{l},A)\,dP(\bs{l}|L_p)
  \end{align}
  satisfies  $\EE\{ \phi(\bs{Z};\pi,\mu)|A=a,L_p=l_p\}=\EE\{Y^a|L_p=l_p\}$  if
  either $\pi=\pi_0$ or $\mu=\mu_0$.
\end{theorem}

In real data analysis, we need to estimate the unknown aspects of 
\eqref{eq:conditional-dr-mapping},
so we arrive at
\begin{equation}
  \label{eq:estimated-conditional-dr-mapping} 
  \wh \phi(\bs{Z};\wh\pi,\wh\mu) = \frac{Y-\wh\mu(\bs{L},A)}{\wh\pi(A|\bs{L})}\int_{\mc{A}}\wh\pi(A|\bs{l})\,d\PP_n(\bs{l}_{}|L_p)+\int_{\mc{A}}\wh\mu(\bs{l},A)\,d\PP_n(\bs{l}_{}|L_p)  
\end{equation}
Once we have esimated the pseudo-outcome, we can use develop methodology analogous to the methodology we developed for testing $\theta_0(\cdot)$ above, based on $\{ (A_i, L_{i,p}, \wh{\phi}(\bs{Z}_i, \wh \pi, \wh \mu)\}_{i=1}^n$.

In this manuscript we focus on the case of discrete variables that are treatment effect modifiers, in which case it is straightforward
to extend the methodology from Subsection~\ref{sec:proposed-method}.
\citet{dette2001nonparametric} proposed three test statistics based on
Nadaraya-Watson estimators of the regression functions to test equality of
regression functions. Specifically, one of the three test statistics they
proposed has a similar form as our statistic $T_n$ above, and can
be written as
\begin{align}
  \label{eq:dette-test-stat3}
  \sum^{k}_{m=2}\sum_{j=1}^{m-1}\int_{\mc{A}}\lb\hat g_m(a)-\hat g_j(a)\rb^2
  w_{m,j}(a)\,da,
\end{align}
where $w_{m,j}$ are (positive, user-chosen) weight functions satisfying $w_{m,j}=w_{j,m}$,
there are $k$ groups (i.e., values of the discrete variable),
and $\hat g_j$, $j=1,\ldots,k$, is the  Nadaraya-Watson estimator for the $j$-th group
regression function.
Combining
the 
conditional doubly robust mapping with \eqref{eq:dette-test-stat3}, we 
consider the following test statistic for  testing whether a discrete
covariate with a finite number, $k$, of possible values  is an effect modifier or not:
\begin{align}
  \label{eq:condition-test-stat}
  T^p_n = \sum_{m=2}^k\sum_{j=1}^{m-1}\int_{\mc{A}}\lb\wh\theta_{h,m}(a)-\wh\theta_{h,j}(a)\rb^2w_{ij}(a)\,da,
\end{align}
where $\wh\theta_{h,j}(a)$ is the local linear estimator applied to the subset
of the tuples $\{(\wh\phi(\bs{Z}_i;\wh\pi,\wh\mu), A_i)\}$ with the $p$-th covariate
$L_{i,p}$ taking the $j$-th category.
We summarize this extended testing procedure as follows.
\begin{enumerate}
\item Estimate the nuisance functions $(\pi_0, \mu_0)$ with $(\wh\pi,\wh\mu)$.
\item Calculate the pseudo-outcomes $\wh\phi(\bs{Z};\wh\pi,\wh\mu)$  according
  to \eqref{eq:estimated-conditional-dr-mapping} and construct
  the local linear estimator $\wh\theta_{h,j}(a)$ for each conditional treatment
  effect curve with the corresponding subset of
  $\{(\wh\phi(\bs{Z}_i;\wh\pi,\wh\mu), A_i)\}$.
\item Calculate the test statistic $T_n^p$ according
  to   \eqref{eq:condition-test-stat}.
\item To estimate the critical value of the asymptotic distribution under the null,
  \begin{enumerate}[label=(\arabic*)]
  \item Use the pooled data $\{(\wh\phi(\bs{Z}_i;\wh\pi,\wh\mu),
    A_i)\}_{i=1}^n$ to construct a single local linear estimator $\wh\theta_h(a)$
  \item  Calculate $ \wh\theta_h(A_i)$ and $\wh\epsilon_i =
    \wh\phi(\bs{Z}_i;\wh\pi,\wh\mu)- \wh\theta_h(A_i) $. Generate the
    bootstrap sample of residuals $\wh\epsilon^{\ast}_i\sim\hat{F}_i$ and response $
    \wh\phi^{\ast}_i= \wh\theta_h(A_i)+\wh\epsilon^{\ast}_i$, where choices of
    $\hat{F}_i$ are the same as discussed in Section~\ref{sec:proposed-method}.
  \item Calculate $T_n^{p\ast}$ using the bootstrap sample $\{(A_i,L_{i,p},\wh\phi_i^{\ast})\}_{i=1}^n$.
  \item Repeat (2) and (3) for $B$ times, where $B$ is the desired number of
    bootstrap samples; We estimate
    critical value of rejection  by $\hat{t}_{n,1-\alpha}^{p\ast}$, the $1-\alpha$ sample
    quantile of $T_n^{p\ast}$ from the bootstrap samples.             
  \end{enumerate}
\item Reject the null hypothesis if $T_n^p>\hat{t}_{n,1-\alpha}^{p\ast}$.
\end{enumerate}

We use simulation to show the validity of this extended procedure in
Section~\ref{sec:simulation}.

\begin{proof}[Proof of Theorem~\ref{thm:conditional-dr-mapping}]
  We use $\bs{l}_{-p}$ to denote of vector of covariates after deleting the
  $p$-th component.

  Let $\bar{m}(a,l_p):=\EE\{\bar{\mu}(\bs{L},a)|L_p=l_p\}$,
  $\bar\varpi(a,l_p):=\EE\{\bar{\pi}(a| \bs L)|L_p=l_p\}$,
  and $\theta_0(a, l_p) := \EE(Y^a | L_p=l_p)$. Then
  \begin{align*}
    \MoveEqLeft \EE\lb\phi(\bs{Z};\bar{\pi},\bar{\mu})|A=a,L_p=l_p\rb\\&=\EE\lb\frac{Y-\bar{\mu}(\bs{L},A)}{\bar{\pi}(A|\bs{L})/\bar\varpi(A,L_p)}+\bar{m}(A,L_p)\Bigr|A=a,L_p=l_p\rb\\
    &=\int\frac{\mu_0(\bs{l},a)-\bar{\mu}(\bs{l},a)}{\bar{\pi}(a|\bs{l})/\bar\varpi(a,l_p)}\,dP(\bs{l}_{}|A=a,L_p=l_p)+\bar{m}(a,l_p)\\
    &= \int \lb\mu_0(\bs{l},a)-\bar{\mu}(\bs{l},a)\rb
      \frac{\pi_0)(a|\bs{l})/\varpi_0(a,l_p)}{\bar{\pi}(a|\bs{l})/\bar\varpi(a,l_p)}\,dP(\bs{l}_{}|L_p=l_p)+\bar{m}(a,l_p)\\
    &=\theta_0(a,l_p)+ \int \lb\mu_0(\bs{l},a)-\bar{\mu}(\bs{l},a)\rb
      \lb\frac{\pi_0(a|\bs{l})/\varpi_0(a,l_p)}{\bar{\pi}(a|\bs{l})/\bar\varpi(a,l_p)}-1\rb\,dP(\bs{l}_{}|L_p=l_p),
  \end{align*}
  where the last line shows the double robustness of the proposed mapping. Some
  details of the calculation are given as follows.
  \begin{itemize}
  \item The second equality above is calculated by iterated expectations:
    \begin{align*}
      \MoveEqLeft    
      \EE\lb\frac{Y-\bar{\mu}(\bs{L},A)}{\bar{\pi}(A|\bs{L})/\bar\varpi(A,L_p)}\Bigr|A=a,L_p=l_p\rb\\
      &=\EE\ls
        \EE\lb\frac{Y-\bar{\mu}(\bs{L},A)}{\bar{\pi}(A|\bs{L})/\bar\varpi(A,L_p)}\Bigr|A,\bs{L}\rb
        \Bigr|A=a,L_p=l_p\rs\\
      &=\EE\lb\frac{\mu_0(\bs{L},A)-\bar{\mu}(\bs{L},A)}{\bar{\pi}(A|\bs{L})/\bar\varpi(A,L_p)}\Bigr|A=a,L_p=l_p\rb\\
      &=\int\frac{\mu(\bs{l},a)-\bar{\mu}(\bs{l},a)}{\bar{\pi}(a|\bs{l})/\bar\varpi(a,l_p)}\,dP(\bs{l}|a,l_p).
    \end{align*}
  \item The third equality comes from the calculation that
    \begin{align*}
      dP(\bs{l}_{}|a,l_p)& =\frac{p(a,\bs{l})}{p(a,l_p)}\, d\nu(\bs{l})\\
                         & = \frac{p(a|\bs{l})p(\bs{l})}{p(a|l_p)p(l_p)} \,d\nu(\bs{l}) \\
                         &=\frac{\pi_0(a|\bs{l})}{\varpi_0(a,l_p)}p(\bs{l}_{-p}|l_p)\,d\nu(\bs{l})\\
                         &=\frac{\pi_0(a|\bs{l})}{\varpi_0(a,l_p)}\,dP(\bs{l}|l_p).
    \end{align*}
  \item For the last equality, note that
    \begin{align*}
      \bar{m}(a,l_p)&=      \int \bar{\mu}(\bs{l},a)\,dP(\bs{l}_{}|l_p) \\
                    &=  \int \lb\bar{\mu}(\bs{l},a)- \mu_0(\bs{l},a)\rb\,dP(\bs{l}_{}|l_p) +
                      \int \mu_0(\bs{l},a)\,dP(\bs{l}_{}|l_p) \\
                    &=\int \lb\bar{\mu}(\bs{l},a)- \mu_0(\bs{l},a)\rb\,dP(\bs{l}_{}|l_p) +\theta_0(a,l_p).
    \end{align*}
  \end{itemize}
\end{proof}

\section{Plots from simulations and from data analysis}

\subsection{Description of simulation for testing for a treatment effect modifier}

We consider the following data generating process with a continuous outcome, which is similar to Model 2 above,
to
perform the simulations for testing for a treatment effect modifier. First, we let
\begin{align*}
  \wt{\bs{L}}=(L_1, L_2, L_3, \wt L_4)^T\sim N(0,\bs{I}_4).
\end{align*}
And 
let $L_4 = \mathbbm{1}\{ \wt L_4>1\}$
and suppose we observe
$\bs{L}=(L_1,L_2,L_3, L_4)^T$. Then we simulate the
treatment level from  Beta distributions,
\begin{align*}
  (A/5)|\bs{L}&\sim \text{Beta}(\lambda(\bs{L}),1-\lambda(\bs{L})),\\
  \text{logit } \lambda(\bs{L})&=0.1L_1+0.1L_2-0.1L_3+0.2L_4,
\end{align*}
and we simulate the continuous response from a normal distribution, 
\begin{align*}
  Y|\bs{L},A&\sim N(\mu(\bs{L},A),0.5^2),\\
  \mu(\bs{L},A)&=0.2L_1+0.2L_2+0.3L_3-0.1\delta L_4 -0.1AL_1+0.1\delta AL_4\\
  &\quad+\exp\lb\frac{(A-2.5)^2}{(1/2)^2}\rb.
\end{align*}
So we see that the parameter $\delta$ controls the distance between the
conditional treatments when $L_4=1$ and when $L_4=0$, and when $\delta=0$, there
is no difference between the two conditional treatment effect curves, i.e.,
$L_4$ is not an effect modifier. In the simulation, we let $\delta$ take values
in $\{0,0.1,0.2,0.3,0.4,0.5\}$. And similar to Section~\ref{sec:simulation},
we
test the performance
of this treatment modifier test
under 4 scenarios: (1) $\pi$ is correctly
specified with a parametric model, $\mu$ is incorrectly specified with a parametric model; (2) $\pi$ is incorrectly
specified with a parametric model, $\mu$ is correctly  specified with a parametric model; (3)
both $\pi$ and $\mu$ are correctly specified with a parametric model; (4) both
$\pi$ and $\mu$ are estimated with Super Learners \citep{van2007super}.

We display the simulation results in
Figures~\ref{fig:effect-modifier-parametric} and
\ref{fig:effect-modifier-nonparametric}. In
Figures~\ref{fig:effect-modifier-parametric}, we see that when at least one of
the nuisance functions are correctly estimated, we have type I error probability
converges to the nominal significance level $\alpha=0.05$ and rejection
probability converging to 1 as we increase the sample size. So this suggests the
extended test maintains double robustness as we expected. In
Figures~\ref{fig:effect-modifier-nonparametric}, with the nonparamtric Super
Learner, we also see type I error probability
converges to the nominal significance level $\alpha=0.05$ and the rejection
probability converges to 1 as we increase the sample size.

\subsection{Plots}

\label{sec:plots-from-simul}
\begin{figure}[H]
  \centering
  \includegraphics[width=\textwidth]{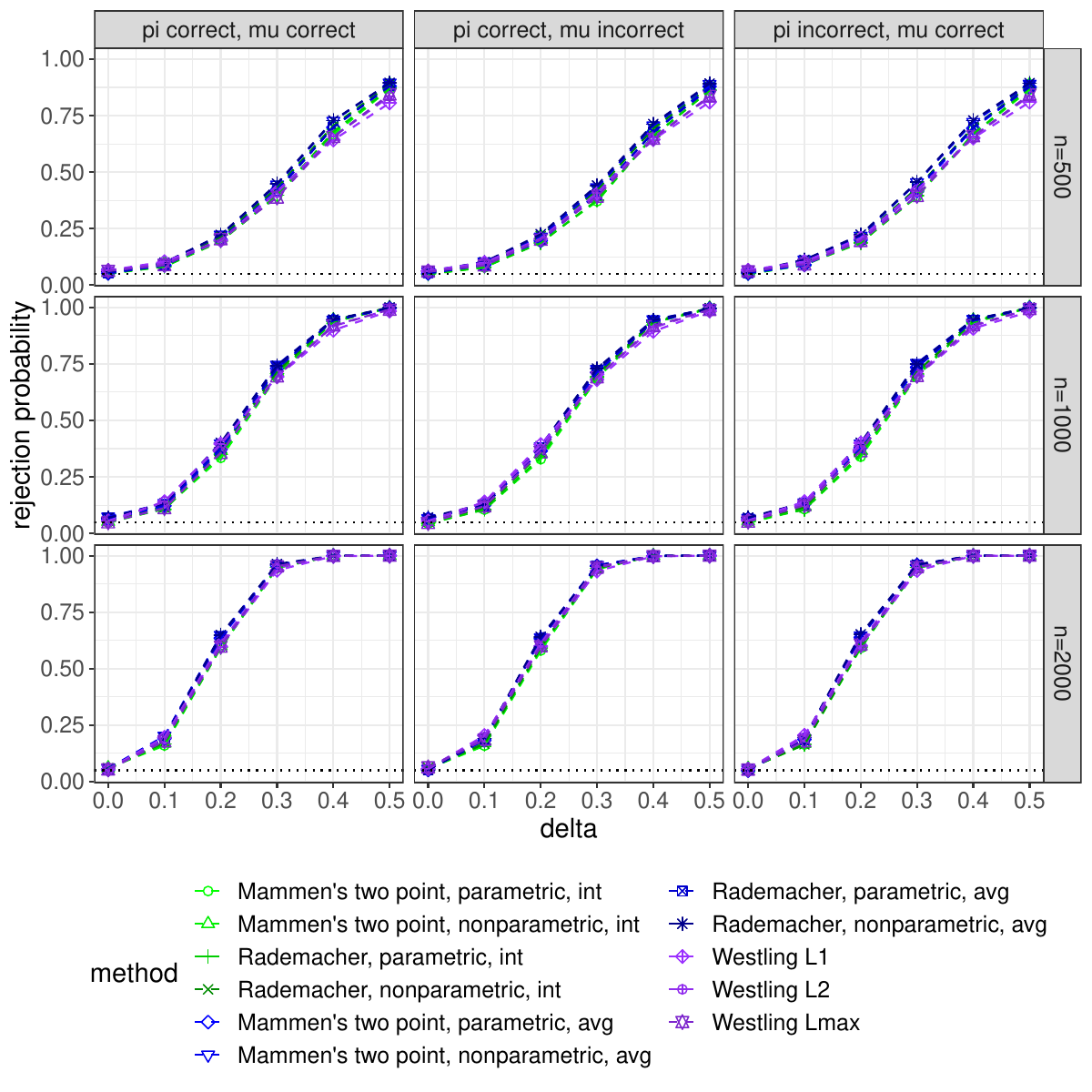}  
  \caption{Simulation result for Model 1 with $\pi$ and $\mu$ estimated from
    parametric models.}
  \label{fig:model1-parametric}
\end{figure}

\begin{figure}[H]
  \centering
  \includegraphics[width=\textwidth]{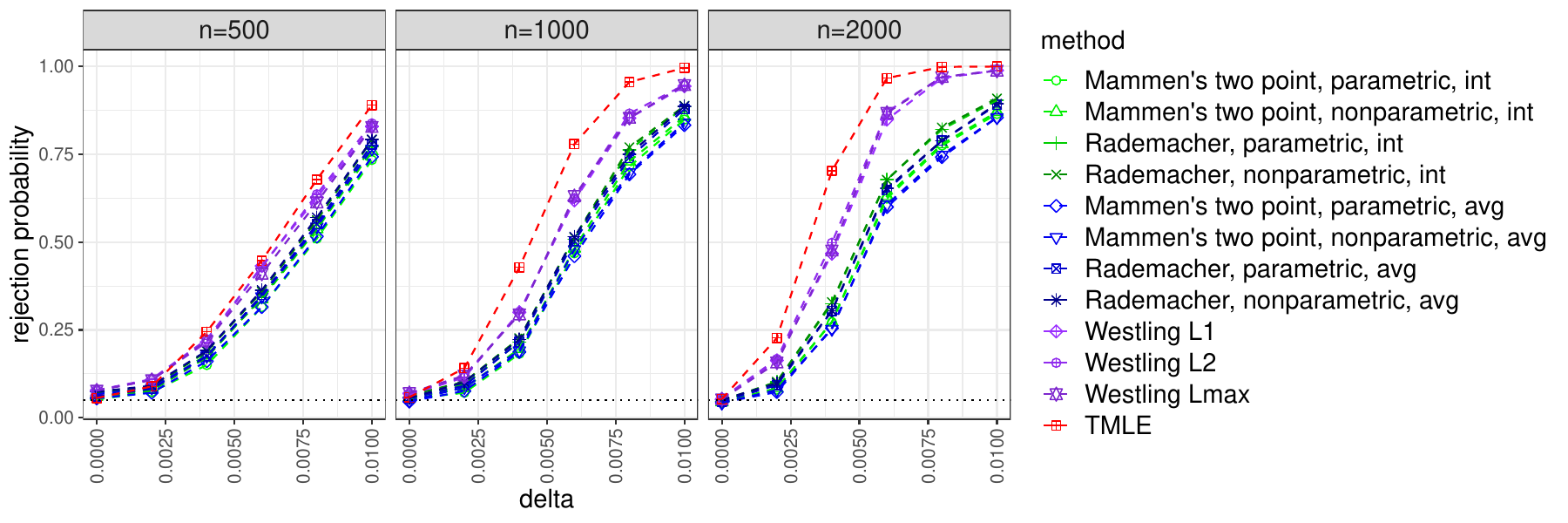}  
  \caption{Simulation result for Model 1 with $\pi$ and $\mu$ estimated from
    nonparametric models.}
  \label{fig:model1-nonparametric}
\end{figure}

\begin{figure}[H]
  \centering
  \includegraphics[width=\textwidth]{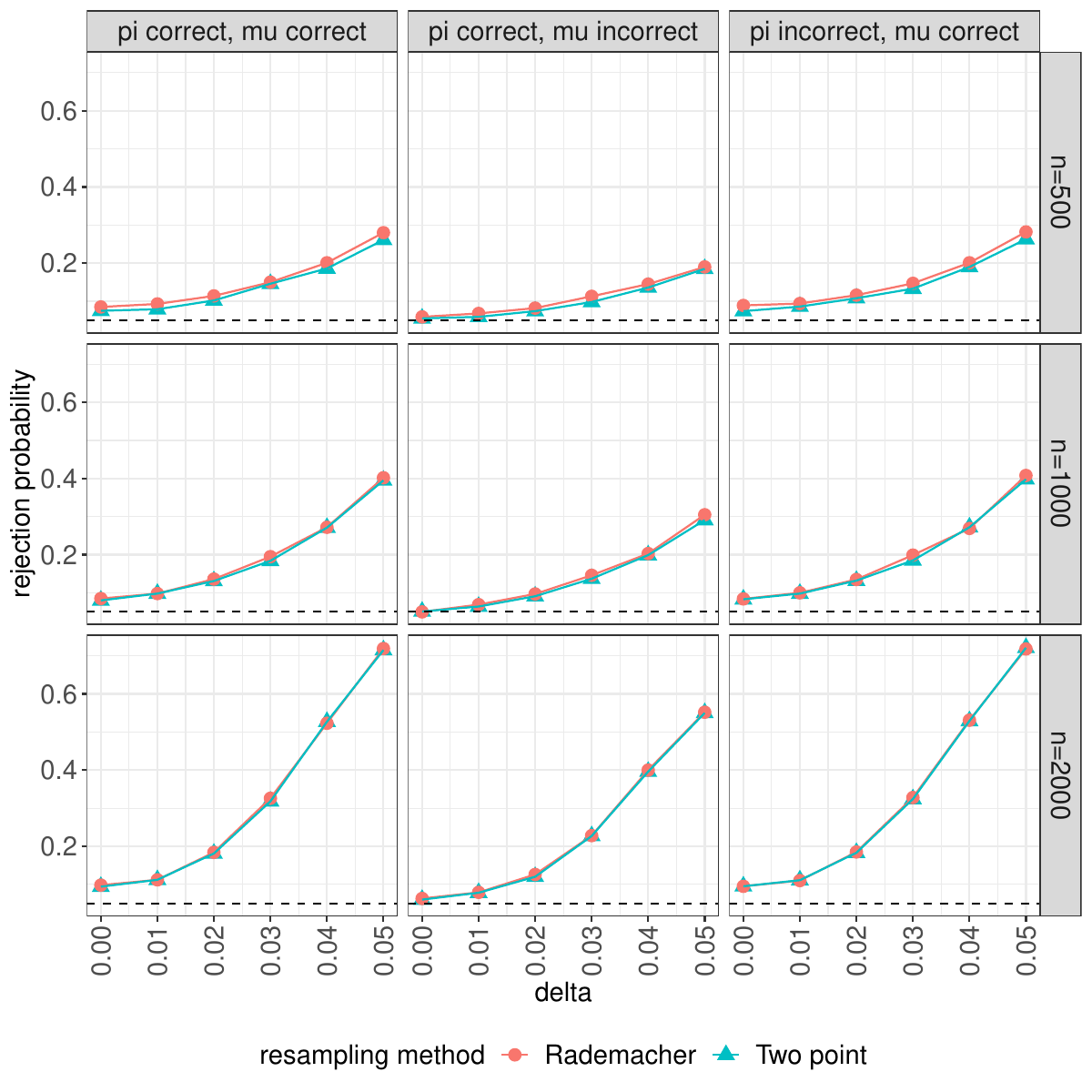}
  \caption{Simulation result for testing effect modifier with $\pi$ and $\mu$
    estimated from parametric models}
  \label{fig:effect-modifier-parametric}
\end{figure}

\begin{figure}[H]
  \centering
  \includegraphics[width=\textwidth]{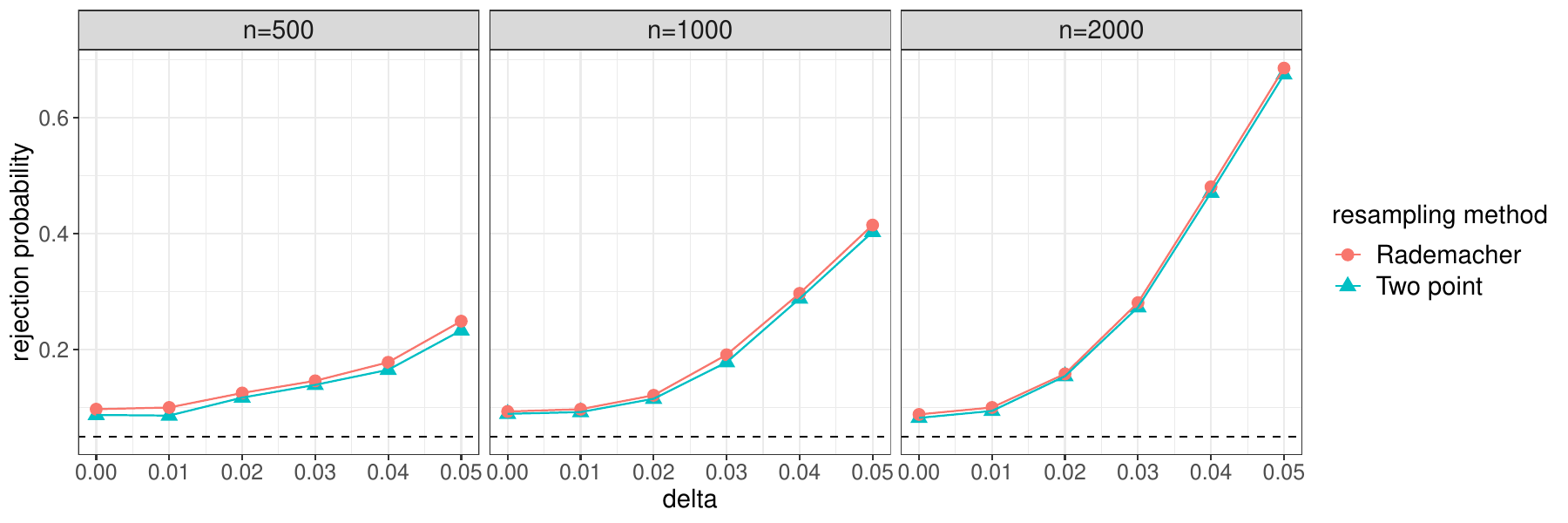}
  \caption{Simulation result for testing effect modifier with $\pi$ and $\mu$
    estimated from nonparametric models.}
  \label{fig:effect-modifier-nonparametric}
\end{figure}
\begin{figure}[H]
  \centering
  \includegraphics[width=\textwidth]{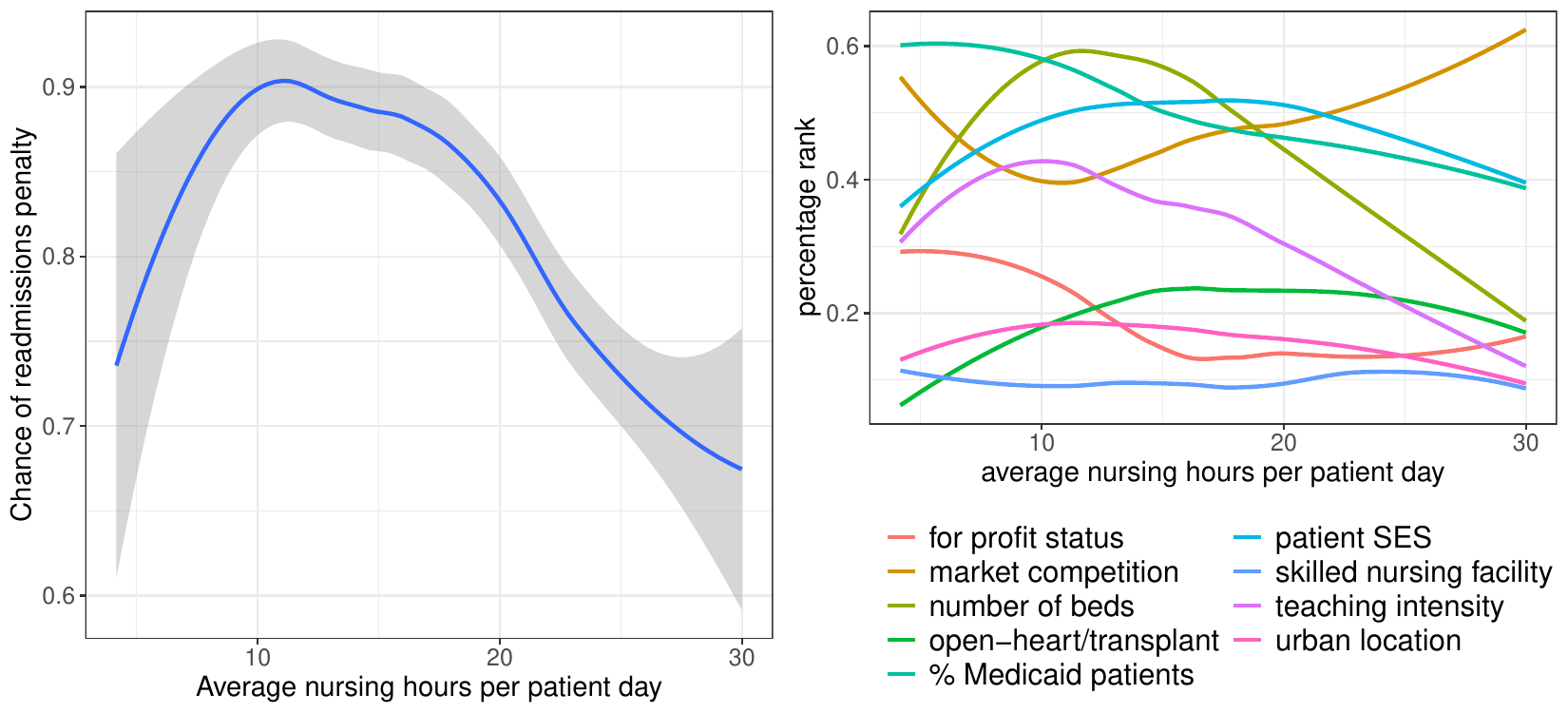}
  \caption{Left: Unadjusted loess fit of outcome against average nursing
    hours. Right: Average covariate values as a function of exposure,after
    transforming to percentiles to display on common scale.}
  \label{fig:observed-data}
\end{figure}
\begin{figure}[H]
  \centering
  \includegraphics[width=\textwidth]{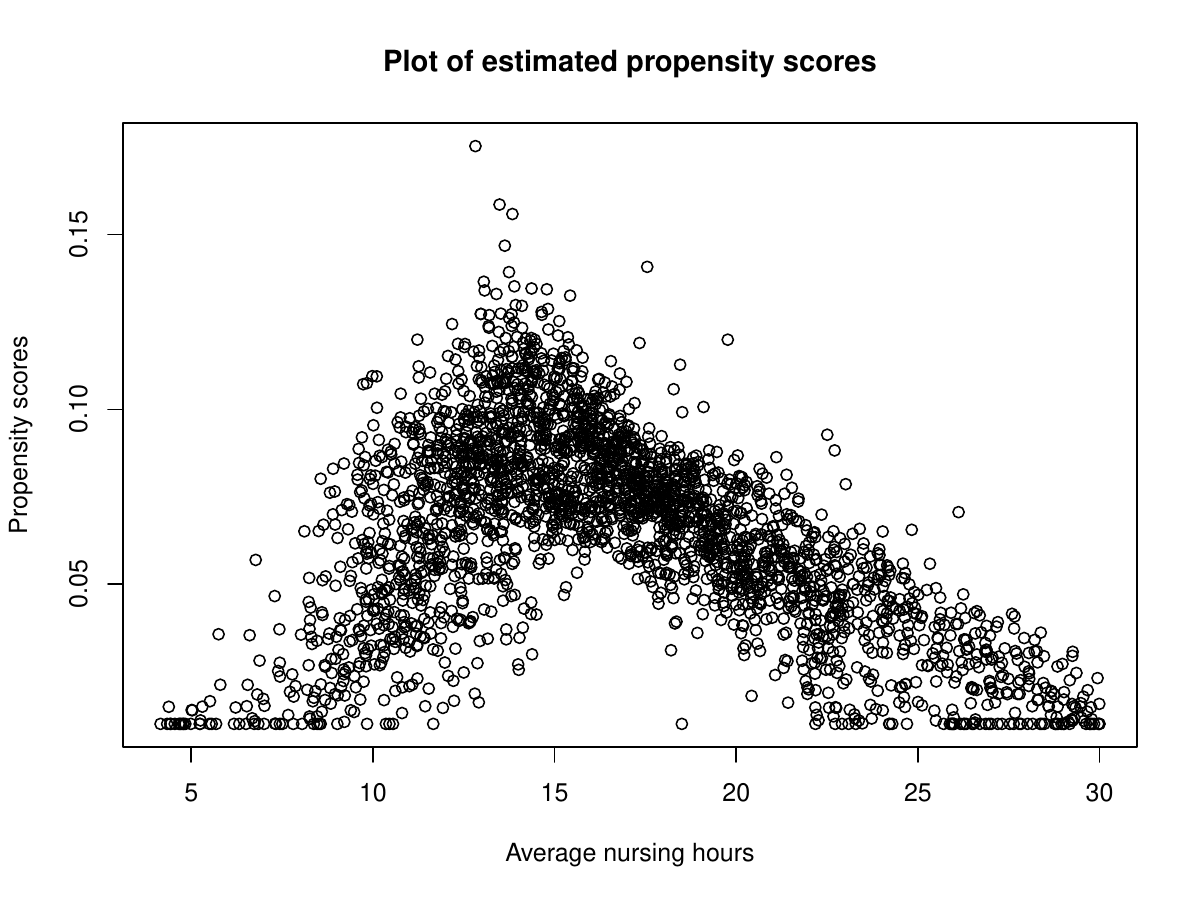}
  \caption{Plot of estimated propensity scores $\wh \pi(a|\bs{l})$ (truncated 
    below by 0.01) against average nursing hours   \label{fig:ps-plot}}
\end{figure}

\section{Cross-fitted Test}
\label{sec:cross-fitting}
In this section, we develop and study a cross-fitted version of the test we presented in the main paper,
and use simulation  studies to show how the dimensionality of the confounder
vector $\bs{L}$ affects the performance of the two tests.

\subsection{Test procedure with cross-fitting}
\label{subsec:cross-fitting-test}
Suppose we randomly partition the
index set $\{1,\ldots,n\}$ into $V$ disjoint sets
$\mc{V}_{n,1},\ldots,\mc{V}_{n,V}$ with cardinalities $N_1,\ldots N_V$, where
$V\in \{2,\ldots,\lfloor n\rfloor\}$ is the number of partitions and without
loss of generality, here we assume all the partitions have equal size, i.e.,
$N_1=N_2=\cdots=N_V$. For each $v\in \{1,\ldots,V\}$, we define $\mc{T}_{n,v} =
\{\bs{Z}_i:i\notin \mc{V}_{n,v}\}$ as the training set for fold $v$.

Recall our original test statistics $T_n$ can be written as
\begin{align}
  \label{eq:test-stat}
  T_n=n\sqrt{h}\int_{\mc{A}}\lp\widehat{\theta}_h(a)-\PP_n\widehat{\xi}\rp^2w(a)\,da,
\end{align}
where
\begin{align}
  \label{eq:pseudo-outcome}  \widehat{\xi}(\bs{Z};\widehat{\pi},\widehat{\mu})=\frac{Y-\widehat{\mu}(\bs{L},A)}{\widehat{\pi}(A|\bs{L})}\int_{\mc{L}}\widehat{\pi}(A|\bs{l})\,d\PP_n(\bs{l})+\int_{\mc{L}}\widehat{\mu}(\bs{l},A)\,d\PP_n(\bs{l}).
\end{align}
 and $\widehat{\theta}_h(a)$ is the local linear estimator applied to
$\{(\widehat{\xi}(\bs{Z}_i;\widehat{\pi},\widehat{\mu}),A_i)\}_{i=1}^n$. Given
one splitting of the data, i.e., $\mc{V}_{n,v}$ and $\mc{T}_{n,v}$, we can
calculate the pseudo-comes as
\begin{align}
  \label{eq:cv-pseudo-outcome}
  \begin{split}
  \MoveEqLeft  \widehat{\xi}^{v}(\bs{Z}_i;\widehat{\pi}_{n,v},\widehat{\mu}_{n,v})\\&=\frac{Y_i-\widehat{\mu}_{n,v}(\bs{L}_i,A_i)}{\widehat{\pi}(A|\bs{L})}\int_{\mc{L}}\widehat{\pi}_{n,v}(A|\bs{l})\,d\PP_{\mc{V}_{n,v}}(\bs{l})+\int_{\mc{L}}\widehat{\mu}_{n,v}(\bs{l},A)\,d\PP_{\mc{V}_{n,v}}(\bs{l}),
  \end{split}
\end{align}
where $\widehat{\pi}_{n,v}$ and $\widehat{\mu}_{n,v}$ are estimated only using
the observations in $\mc{T}_{n,v}$, and $\PP_{\mc{V}_{n,v}}$ is the empirical
measure defined on $\{\mc{Z}_i:i\in \mc{V}_{n,i}\}$. Similarly, we can calculate
the test statistic restricted to  this splitting as
\begin{align}
  \label{eq:cv-test-statistic}
T_n^{v} =N_v\sqrt{h}\int_{\mc{A}}\lp\widehat{\theta}^v_h(a)-\PP_{\mc{V}_{n,v}}\widehat{\xi}^v\rp^2w(a)\,da.
\end{align}
Here $\widehat{\theta}^v_h(a)$ is the local linear estimator applied
to
$\{(\widehat{\xi}^{v}(\bs{Z}_i;\widehat{\pi}_{n,v},\widehat{\mu}_{n,v}),A_i)\}_{i\in\mc{V}_{n,v}}$,
and with slight abuse of notation, we let $\PP_{\mc{V}_{n,v}}\widehat{\xi}^v=1/N_{v}\sum_{i\in\mc{V}_{n,v}}\widehat{\xi}^{v}(\bs{Z}_i;\widehat{\pi}_{n,v},\widehat{\mu}_{n,v})$.
We can do this for each splitting and aggregate the results from all the splittings to get the
cross-fitted test statistic as
\begin{align}
T_n^{\circ} = \frac{1}{V}\sum_{v=1}^VT_n^{v}.
\end{align}

We also use a modified wild bootstrap procedure to estimate the
distribution of the cross-fitted test statistic. We perform
wild bootstrap for each splitting separately. 
The detailed test procedure with cross-fitting is summarized as
follows.
\begin{enumerate}
\item Randomly partition the index set $\{1,\ldots, n\}$ into $V$
  disjoint set $\mc{V}_{n,1},\ldots,\mc{V}_{n,V}$ with equal
  cardinalities $N_V$ (assume $n$ is a multiplier of $V$).
\item For each split $\mc{T}_{n,v}, \mc{V}_{n,v}$,
  \begin{enumerate}
  \item Estimate  the nuisance functions $(\pi_0, \mu_0)$ with
    $(\widehat{\pi}_{n,v}, \widehat{\mu}_{n,v})$ using observations
    from $\mc{T}_{n,v}$.
    \item Calculate the pseudo outcomes
      $\widehat{\xi}^{v}(\bs{Z};\widehat{\pi}_{n,v},\widehat{\mu}_{n,v})$
      by \eqref{eq:cv-pseudo-outcome} and construct the local linear
      estimator $\widehat{\theta}^v_h$ using
      $\{(\widehat{\xi}^{v}(\bs{Z}_i;\widehat{\pi}_{n,v},\widehat{\mu}_{n,v}),A_i)\}_{i\in\mc{V}_{n,v}}$.
     \item Calculate $T_n^{v}$ using \eqref{eq:cv-test-statistic}.
    \item To generate wild bootstrap samples for this splitting,
      \begin{enumerate}
      \item For each $i \in \mc{V}_{n,v}$, calculate
        $\hat{\varepsilon}_{i}^v =
        \widehat{\xi}^{v}(\bs{Z}_i;\widehat{\pi}_{n,v},\widehat{\mu}_{n,v})-\widehat{\theta}_{h}^v(A_i)$. Generate
        $\varepsilon_i^{\ast,v,b}\sim \hat{F}_i$ and use
        $(\xi_{i}^{\ast,v,b}=\PP_{\mc{V}_{n,v}}\widehat{\xi}^v+\varepsilon_i^{\ast,v,b},A_i)$
        as bootstrap observations.
        \item Calculate
          $T_n^{v,\ast,b}$ as
          \begin{align}
            T_n^{v,\ast,b}=N_v\sqrt{h}\int_{\mc{A}}\lp\widehat{\theta}_h^{v,b}(a)-\frac{1}{N_v}\sum_{i\in\mc{V}_{n,v}}\xi_{i}^{\ast,v}\rp^2 w(a)\,da,
          \end{align}
          where $\widehat{\theta}_h^{v,b}(a)$ is the local linear
          estimator applied to
          $\{(\xi_{i}^{\ast,v,b},A_i)\}_{i\in \mc{V}_{n,v}}$.
          
        \item Repeat i. and ii. for $B$ times where $B$ is the desired number
          of boostrap samples.        
        \end{enumerate}
      \end{enumerate}
      \item Calculate $T_n^{o}=\sum_{v=1}^VT_n^v/V$,
        $T_n^{\ast,b}=\sum_{v=1}^VT_n^{v,\ast,b}/V$ for $b=1,\ldots
        B$. Let $\hat{t}_{n,1-\alpha}^{\ast, o}$ denote the $1-\alpha$
        quantile of $\{T_n^{\ast,b}\}_{b=1}^B$. Reject the null
        hypothesis if $T_n^{o}>\hat{t}_{n,1-\alpha}^{\ast, o}$.
\end{enumerate}

\subsection{Simulation studies for the effect of data dimensionality}
\label{subsec:cross-fitting-simulation}
We first compare the performance of the test with cross-fitting with
our original non-cross-fitted test procedure. We use the same data generating models,
Model 1 and Model 2, to perform the simulation. Both $\pi$ and $\mu$
are estimated with Super Learners. The results are shown
in Figures~\ref{fig:cv-model1} and \ref{fig:cv-model2}. We can see
under both two data models, both two versions of tests have type I
error probability converging to the desired level as we increase the
sample size. However, the cross-fitted test shows uniformly lower
power under alternatives. 

\begin{figure}[htbp]
  \centering
  \includegraphics[width=\textwidth]{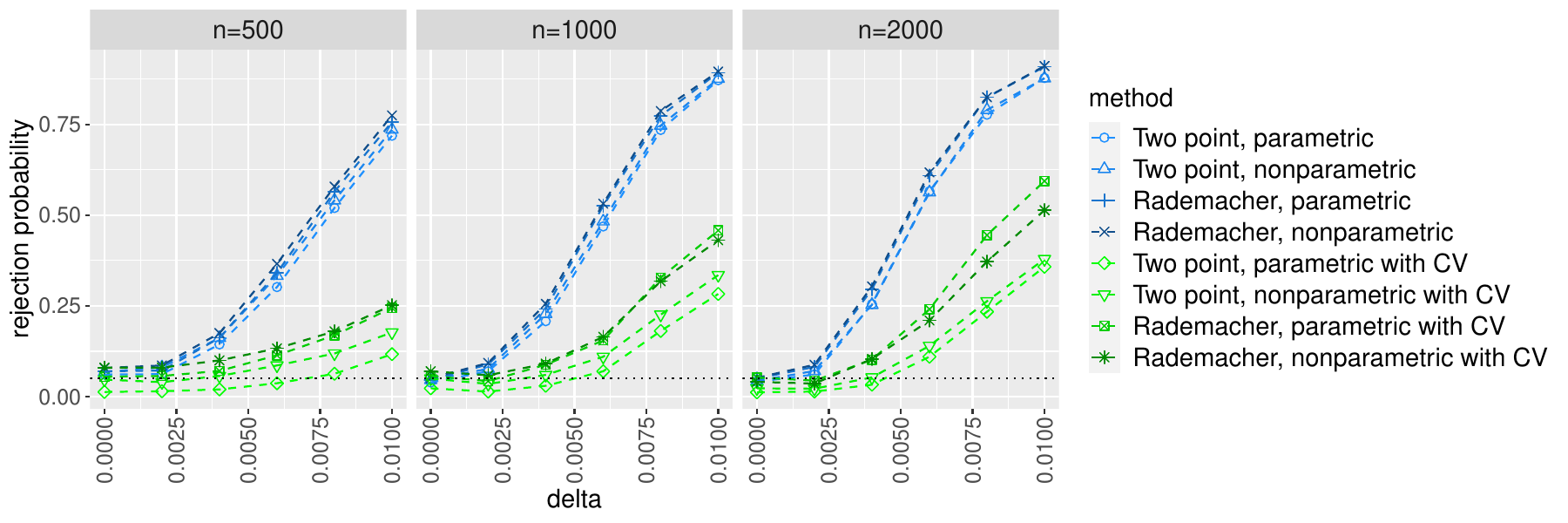}
  \caption{Simulation result for comparing non-cross-fitted test and
    cross-fitted test under model 1 \label{fig:cv-model1}}
\end{figure}

\begin{figure}[htbp]
  \centering
  \includegraphics[width=\textwidth]{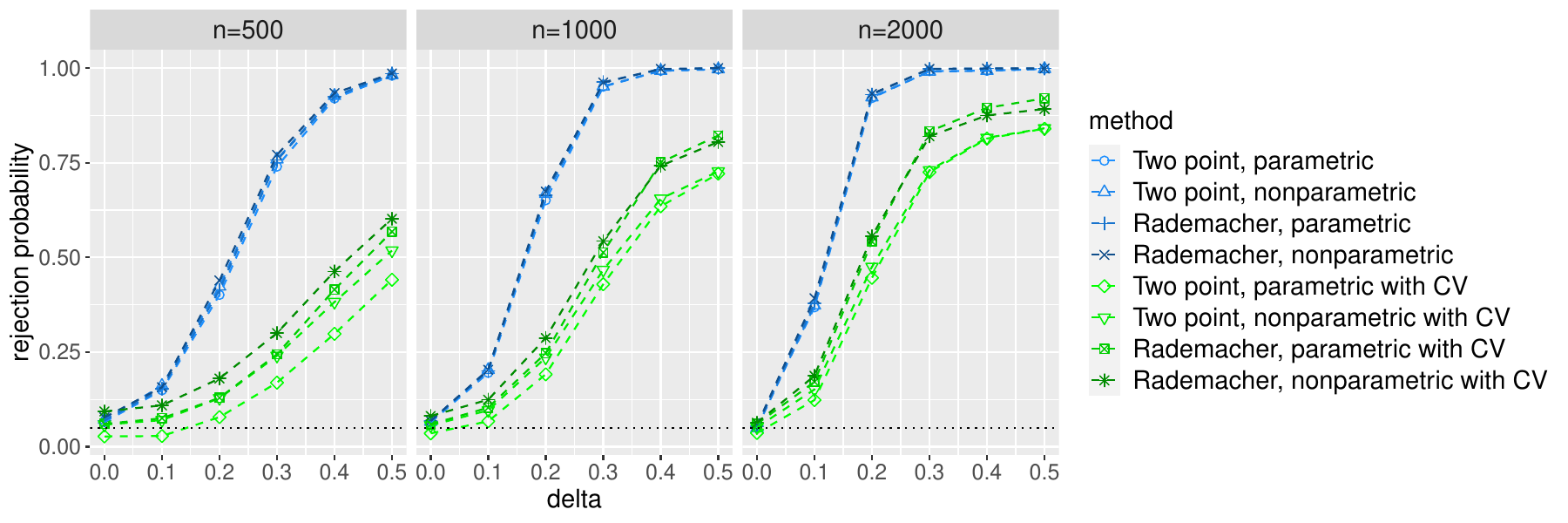}
  \caption{Simulation result for comparing non-cross-fitted test and
    cross-fitted test under model 2 \label{fig:cv-model2}}
\end{figure}

\begin{figure}[tbp]
  \centering
  \includegraphics[width=\textwidth]{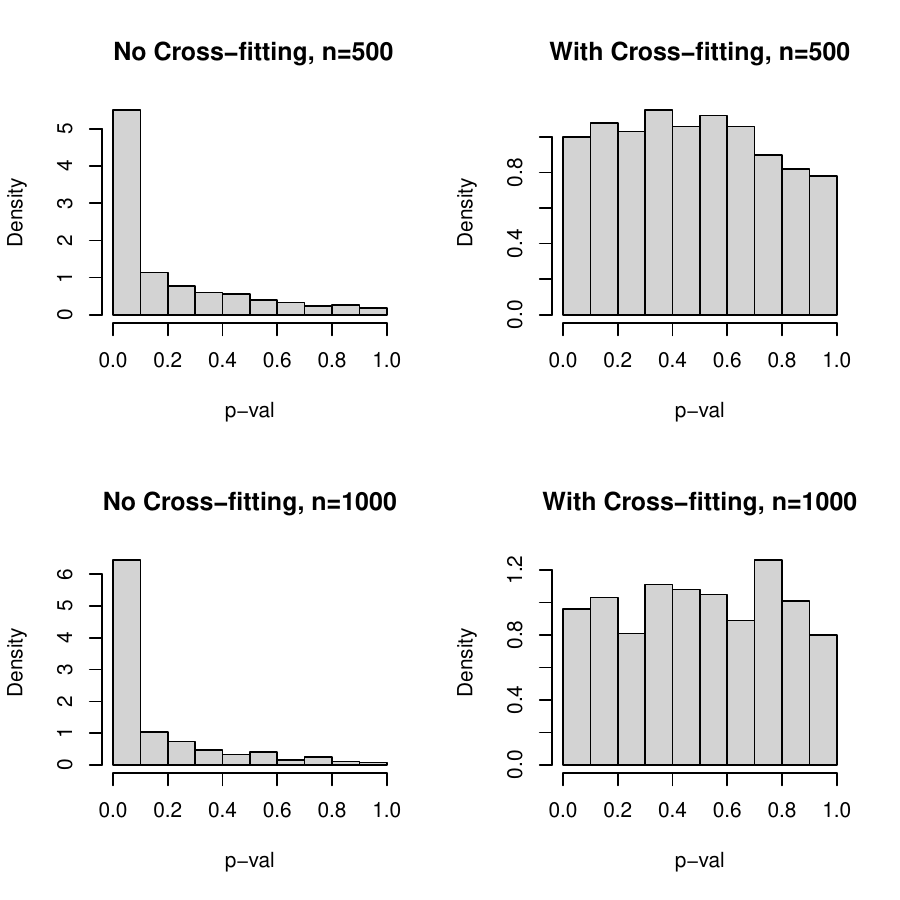}
  \label{fig:cv-hd}
  \caption{Simulation result for comparing p-values of non-cross-fitted test and
    cross-fitted test under Model~\eqref{eq:model-hd}. Type I error
    probabilities for $\alpha=0.05$ are: No Cross-fitting, n=500 (0.448);
  Cross-fitting, n=500 (0.056); No Cross-fitting, n=1000 (0.535);
  Cross-fitting, n=1000 (0.047).}
\end{figure}

Then we compare the two versions of tests with high dimensional
data. The data generating model comes from  \citet{Colangelo:2020tt}
with some slight modifications, where we let 
\begin{align}
  \label{eq:model-hd}
  \begin{split}
  \bs{L}&=(L_1,\ldots, L_{100})'\sim N(0,\Sigma),\\
    A&=\Phi(3\bs{L}'\beta)+0.75\nu,\\Y&=\gamma(1.2A+A^2+AL_1)+1.2\bs{L}'\tilde{\beta}+\epsilon
  \end{split}
\end{align}
where $\beta_j=1/j^2, \text{diag}(\Sigma)=1$, the $(i,j)$-entry
$\Sigma_{ij}=0.5$ for $|i-j|=1$ and 0  for $|i-j|>1$, for
$i,j=1,\ldots, 100$. $\Phi$ is the CDF of $N(0,1)$ and
$\tilde{\beta}=(\beta_{100},\beta_{99},\ldots, \beta_1)$. We set
$\gamma=0$ to show the effect of high dimensional data on type I error
probability of the test and when $\gamma>0$, the treatment effect is nonconstant. We use sample sizes 500 and 1000, and use
Random Forest to estimate the nuisance functions $\pi$ and $\mu$.  The
results are displayed in Figure~\ref{fig:cv-hd}, where we plot the histograms for
p-values from each replication and calculate the type I error
probabilities for $\alpha=0.05$ in the caption. We can see, the
distribution of p-values of non-cross-fitted test is very skewed and
far away from uniform distribution. Moreover,  increasing the sample size to
1000 does not show a great improvement.  On the other side,  the p-values from the
cross-fitted test are quite uniformly distributed. The calculations
also shows the non-cross-fitted test failed to maintain the desired
level of type I error probability and the cross-fitted test achieved
the desired level.

In summary, the cross-fitted test has lower power, especially with low
dimensional data. But the non-cross-fitted test tends to fail to maintain the
type I error probability under high dimensional data. So in practice,
we need to decide which version to use according the dimensionality of
our real world problem.

\section{U- and V-process results}
\label{sec:v-process-results}

\noindent
In analyzing the remainder terms in Theorems
\ref{thm:null-normality} and
\ref{thm:bootstrap-consistency-null}, 
we use the theory of U- or V-processes, which we thus discuss in this section.
Let $f \colon \mc{Z}^r \to \RR$ be permutation symmetric in its arguments.  A {\it U-statistic} based on the `kernel' $f$, and our i.i.d.\ sample $W_1, \ldots, W_n \in \mc{Z}$,
is ${n \choose r }^{-1} \sum_\beta f(W_{\beta_1}, \ldots, W_{\beta_r})$ where ${n \choose r} = n! / (n-r)! r!$ is the binomial coefficient (and $i!$ means the $i$th factorial) and the sum is over all ${n \choose r }$ combinations of $r$ non-repeated/distinct elements out of $n$ data points \citep[Chapter 12]{vanderVaart:1998dr}.  The kernel and the U-statistic have degree (sometimes ``order'') $r$.  If instead we sum over all size-$r$-subsets of the $n$ data points, $n^{-r} \sum_{i_1,\ldots,i_r}^n f(W_{i_1}, \ldots, W_{i_r})$ we get a degree $r$ {\it V-statistic}.  U- and V-statistics often have the same asymptotic distributions.  In our present context, we will need to allow $f$ to range over a function class, and so we arrive at so-called {\it U-} or {\it V-processes.}  For us V-processes arise as remainder terms and we wish to show they are asymptotically negligible; we will apply  maximal inequalities to do so.

To start, consider an order 2
U-process 
\begin{equation}
  \label{eq:1115}
  U_n(f):= n^{-3/2} \sum_{1 \le i < j \le n} f(W_i,W_j)  
\end{equation}
(with slight laziness in the normalization since here we will only be interested in order of magnitude)
for $f$ varying over some class $\mc F$ of centered functions, meaning $\PP f(W_1,W_2)=0$.  We assume $f(w_1,w_2) = f(w_2,w_1)$.  (Note: some of the functions $f$ we need to consider will not be symmetric,
but they can and will be symmetrized by considering $f(w_1,w_2) + f(w_2,w_1)$.)  A degenerate process is one such that $\PP f(w, W) = 0$ for almost every $w$ (and all $f \in \mc F$).  Otherwise the process is non-degenerate.  Like U-statistics, U-processes can be decomposed via a H\'ajek or Hoeffding decomposition into an i.i.d.\ process and a degenerate U-process.
For any $f$, define
$(\projf)(w_1,w_2) := f(w_1,w_2) - \PP f(w_1, W_2) - \PP f(W_1,w_2) + \PP f(W_1,W_2)$
(which satisfies $\PP \projf(w,W) =0$).
    Then
    $U_n(f) = U_n(\projf  ) + (n-1)n^{-3/2}(\PP_n \otimes \PP)(f)$
    where $\PP_n \otimes \PP$ is the product measure of $\PP_n$ (the empirical
    measure of the $W_i$'s) and $\PP$ on $\mc Z \times \mc Z$.
    \begin{mylongform}
      \begin{longform}
        (Note: for any
        function $g$ on $\mc W$,
        $\sum_{1 \le i < j \le n} g(W_i) + g(W_j) = (n-1) \sum_{i=1}^n g(W_i)$.)
      \end{longform}
    \end{mylongform}
    We can combine maximal inequalities for i.i.d.\ empirical processes with a maximal inequality for degenerate U-statistics from
    \cite{Nolan:jw}
    to yield a maximal inequality for the entire process.
    This is done in
     Proposition~\ref{lem:100} below.
    The former term generally dominates.
    We modify the result so it fits the details of our setting.
    The following result is
    a    form of 
    Theorem 6 of \cite{Nolan:jw} for (degenerate) U-processes
    (combined with 
    Theorem 2.14.1 of \cite{vanderVaart:1996tf}).
    We do not discuss measurability difficulties here.
    For a class $\mc F$ of symmetric functions (i.e., $f(w_1,w_2) = f(w_2,w_1)$),
    we let $\PP \mc F$ denote $\{ \PP f(\cdot, W) : f \in \mc F\}$.
    Recall the definitions of $J(\cdot, \cdot, L_2)$ and $J_2(\cdot, \cdot, L_2)$
    given in
    \eqref{eq:uniform-entropy-2-integral}.


    \begin{proposition}
      \label{lem:100}
      Assume $\mc F$ is a class of (measurable) functions on a (measure) space
      $\mc W \times \mc W$, with (measurable) envelope $F$.  Assume $f\in \mc F$
      satisfies $f(w_1,w_2) = f(w_2,w_1)$ and $\PP f(W_1,W_2) = 0$.  Assume
      $W_1, \ldots, W_n$ are i.i.d.\ and $U_n$ is defined by \eqref{eq:1115}.  Let
      $F_1(w) $
      be an envelope for $\PP \mc F$.  Then for a
      universal constant $C >0$,
      \begin{equation}
        \PP \| U_n \|_{\mc F} \le C   J_2(1, \mc F, L_2) \sqrt{ \PP F(W_1,W_2)^2 } n^{-1/2}
        +
        C J(1, \PP \mc F, L_2) \sqrt{\PP F_1(W)^2}.
      \end{equation}
    \end{proposition}


    \begin{proof}[Proof of Proposition~\ref{lem:100}]
  For any $f$ defined on $\mc{W}\times\mc{W}$, define
  $(\Pi f)(w_1,w_2):=f(w_1,w_2)-\PP f(w_1,W_2)-\PP f(W_1,w2)+\PP f(W_1,W_2)$
  (which satisfies $\PP\Pi f(w,W)=0$).
  The result follows from the decomposition $U_n(f) = U_n(\projf ) + (n-1)n^{-3/2}(\PP_n \otimes \PP)(f)$ from above.  We apply Theorem 6 of \cite{Nolan:jw} to the first summand which is a degenerate U process.  That theorem states that
  \begin{equation}
    \label{eq:Un-NP-thm6-bound}
    E U_n (\Pi f) \le
    n^{-3/2}
    C E( \tau_n \int_0^{\theta_n/\tau_n} 1+ \log N(\epsilon \tau_n  , \mc F, L_2(T_n)) d \epsilon)
  \end{equation}
  where
  $\tau_n = (T_n F^2)^{1/2}$,  $\theta_n = \sup_{f \in \mc F} (T_n f^2)^{1/2}/4$,
  and where
  $T_n$ is a measure defined on page 782 of \cite{Nolan:jw} which places mass $1$ at $4 n(n-1)$ pairs of data points $(W_i,W_j)$  where we always have $i \ne j$. (Note that  the covering numbers are unchanged under a rescaling of the measure $T_n$.)
  If, in the integral on the
  right side of \eqref{eq:Un-NP-thm6-bound},
  we replace $T_n$ by a generic probability measure (since $N$ is invariant under rescaling of the measure) $Q$ and take a sup over $Q$, upper bound $\theta_n/\tau_n$ by $1$, then the integral becomes $J_2(1,\mc F)$, which can be factored out of the expectation.  The expectation is then just $E \tau_n$;  by Jensen's inequality
  $E (T_n F^2)^{1/2} \le (E T_n F^2)^{1/2} = (4n(n-1))^{1/2} (\PP F^2(W_1,W_2))^{1/2}$, so we
  see the degenerate U-process is bounded by
  $C J_2(1, \mc F) \sqrt{ \PP F(W_1,W_2)^2 } n^{-1/2}.$

  We apply Theorem 2.14.1 of
  \cite{vanderVaart:1996tf} to see that $(n-1)n^{-3/2}(\PP_n \otimes \PP)(f)$
  is upper bounded by $C J(1, \PP \mc F) \sqrt{\PP F_1(W)^2}$, which
  completes the proof.
\end{proof}

In addition to the above result, which applies nicely to order $2$
U-processes, we also rely on results from
\cite{Arcones_Gine_1993} which apply to higher order U-processes
(see also \cite{Sherman_1994}).
We provide a maximal inequality derived from the results of \cite{Arcones_Gine_1993} 
in Proposition~\ref{prop:1}.
We consider U-processes of order $3$ in Lemma~\ref{lem:4}.

Let us start by considering  the term $D_3$ from \eqref{eq:6}.
We can write $D_3 $ as
\begin{equation}
  \label{eq:D3-decomp-main}
  \begin{split}
    \MoveEqLeft \PP_n\{\xi(\bs{Z};\bar{\pi},\bar{\mu})-\widehat{\xi}(\bs{Z};\widehat{\pi},\widehat{\mu})\}\\&=\PP_n\{\xi(\bs{Z};\bar{\pi},\bar{\mu})-\xi(\bs{Z};\widehat{\pi},\widehat{\mu})\}+\PP_n\{\xi(\bs{Z};\widehat{\pi},\widehat{\mu})-\widehat{\xi}(\bs{Z};\widehat{\pi},\widehat{\mu})\}.
  \end{split}
\end{equation}
The second term on right side of \eqref{eq:D3-decomp-main} can be written as an order $2$ V-process, as it equals
\begin{equation*}
  \begin{split}
    \MoveEqLeft \frac{1}{n} \sum_{i=1}^n \ls\frac{Y_i-\widehat{\mu}(\bs{L}_i,A_i)}{\widehat{\pi}(A_i|\bs{L}_i)} \lb\frac{1}{n}\sum_{j=1}^n\widehat{\pi}(A_i|\bs{L}_j)-\int\widehat{\pi}(A_i|\bs{l})\,dP(\bs{l}) \rb\right.\\    &\qquad\left.+\lb\frac{1}{n}\sum_{j=1}^n\widehat{\mu}(\bs{L}_j,A_i)-\int\widehat{\mu}(\bs{l},A_i)\,dP(\bs{l})\rb\rs.
  \end{split}
\end{equation*}
This yields a V-process (indexed by $\pi \in \mc{F}_\pi,$  $\mu \in \mc{F}_\mu$) based on the non-symmetric kernel
\begin{equation*}
  h_1(w_1, w_2)
  \equiv   h_{1, \pi, \mu}(w_1, w_2) 
  := \frac{y_1 - \mu(l_1,a_1)}{\pi(a_1|l_1)}
  \wt \pi(a_1 | l_2) + \wt \mu(l_2,a_1),
\end{equation*}
where 
we let tilde $ \wt{ \cdot } $ operate on any $\mu, \pi$ to
yield $\wt{\pi}(a_1|l_2) := \pi(a_1|l_2) - \PP \pi(a_1|L),$ and 
$\wt\mu(l_2,a_1) := \mu(l_2,a_1) - \PP \mu(L,a_1)$.  In
Lemma~\ref{lem:3}  we show that
this V-process (i.e., the second term on right side of \eqref{eq:D3-decomp-main})
is  $O_p(n^{-1/2})$.  (This is indeed the order that would arise for a degenerate V-statistic (with fixed kernel) of order 2.)
Similarly, V-process terms arise from the other remainder terms.  The outline and explanation of where they arise is given in the proof outline in the main document \cite{drtest-main}, after the statement of Theorem~\ref{thm:null-normality}.  Terms with order $ \ge 3$ are handled by the following proposition.
It uses results from the proof of Theorem~5.2 of \cite{Arcones_Gine_1993}.  Recall that an order $m$ kernel $H(z_1, \ldots, z_m)$ for a U-statistic is maximally degenerate (for measure $P$) if integrating over any one of the $m$ arguments yields an identically zero function.
The definition of $P$-measurable is given in  Definition 2.3.3, page 110, of  \cite{vanderVaart:1996tf}.
Let $U_m^n f :=
\frac{(n-m)!}{n!}
\sum_{i_1 \ne i_2 ... \ne i_m} f(X_{i_1}, \ldots, X_{i_m})$, where $n!$ is the factorial of $n$, based on a $P$-i.i.d.\ sample of $X_i \in \mc{Z}$.
Let $\| f \|_n := (U^n_m f^2)^{1/2}$ and recall the definition of $J_m(\delta, \mc F, L_2) \equiv J_m(\delta, \mc F)$ given in
\eqref{eq:uniform-entropy-2-integral}.
\begin{proposition}
  \label{prop:1}
  Let $\mc{F} $ be a $P$-measurable class of maximally degenerate functions  on a measurable product space $\mc{Z}^m$ with envelope $F$ satisfying $P^mF^2 < \infty$.
  Then 
  \begin{equation*}
    \EE \| n^{m/2} U_m^n f \|_{\mc F} \lesssim
    \EE ( J_m(\theta_n, \mc F) \| F\|_n )
    \lesssim J_m(1, \mc F) (P^mF^{2})^{1/2}
  \end{equation*}
  where   $\theta_n := \| \|f \|_n \|_{\mc F} / \| F \|_{n}$.
\end{proposition}
\begin{proof}
  By the first
  lines of the proof of Theorem~5.2 of \cite{Arcones_Gine_1993}, we have
  \begin{equation*}
    \EE \| n^{m/2} U_m^n f \|_{\mc F_\delta} \lesssim
    \EE \int_0^\infty \lp \log N(\epsilon, \mc{F}, L_2(U_m^n)) \rp^{m/2} d\epsilon
  \end{equation*}
  where for $\delta > 0$, 
  $\mc{F}_\delta:= \{ f-g : f, g \in \mc{F}, e_{P,m}(f,g) \le \delta n^{-m/2(m+1)} \}$
  and $ e_{P,m}(f,g) := \| f- g \|_{L^2(P^m)}$.

  Now replace $\mc{F} $ by $\mc{F} \cup \{0\}$ and so
  \begin{equation*}
    E \| n^{m/2} U_m^n f \|_{\mc{F}} \lesssim \int_0^{\| \|f\|_n \|_{\mc F}} \lp \log 1+ N( \epsilon, \mc{F} , L_2(U_m^n)) \rp^{m/2} d\epsilon,
  \end{equation*}
  by Propositions 2.1 and 2.6 in
  \cite{Arcones_Gine_1993} and the fact that $0 \in \mc{F} \cup \{0\}$ to replace $\mc{F}_\delta$ by $\mc{F} \cup \{0\}$  (see e.g.\ Corollary 2.2.8 of
  \cite{vanderVaart:1996tf}).
  Then do a change of variable to see the previous expression equals
  \begin{equation*}
    \EE
    \lp
    \|F \|_n \int_0^{\theta_n} \log^{m/2} 1 + N(\epsilon \|F\|_n, \mc{F}, L_2(U^n_m))
    \, d\epsilon 
    \rp 
  \end{equation*}
  where
  $\theta_n := \| \|f \|_n \|_{\mc F} / \| F \|_{n}$.
  This gives the first inequality of the  proposition. 
  And then the previous expression is bounded above by
  \begin{equation*}
    E \| F \|_n \sup_Q \int_0^1 (\log^{m/2} 1 + N( \epsilon \|Q \|_2, \mc {F} ,
    L_2(Q))) \, d\epsilon.
  \end{equation*}
  By Jensen's inequality, $E \| F \|_n \le (P^m F^2 )^{1/2}$ which gives the second inequality of the proposition. 
\end{proof}

\newpage

\section{Proof of Theorem~\ref{thm:bootstrap-consistency-null}}

The rough idea of the proof is that the oracle bootstrap is consistent, and then the bootstrap remainder terms (which we show to be asymptotically negligible) can either be dominated by, or analyzed in somewhat analogous fashion to, various remainder terms that arose in the analysis of $T_n$ (not bootstrap); sometimes relying on the fact that symmetrized (multiplier) terms are of the same order of magnitude as the non-symmetrized terms.  Each bootstrap error term has three components, from the three summands in $\wh \xi_i^* = \delta_i ( \wh \xi_i - \PP_n \wh \xi) + \PP_n \wh \xi$. 
So the details are somewhat lengthy; we break the argument up into 6 steps again.  The lengthiest computation is (again) in Step 6. 

\begin{proof}[Proof of Theorem~\ref{thm:bootstrap-consistency-null}]
  Without loss of generality we take $w(a) \equiv 1$, which does not substantively modify the proof. 
  Recall the setup: we let
  \begin{equation*}
    T^*_n := n \sqrt{h} \int_{\mc A} \lp \wh \theta^*(a) - \PP_n^* \wh \xi^* \rp^2 da,
  \end{equation*}
  where $\PP_n^*$ is the `empirical' measure of $\lb (\delta_i, \bs Z_i) \rb_{i=1}^n$, with $\delta_i$ being i.i.d.\ Rademacher variables independent of $\{ \bs Z_i \}_i$.  We let
  \begin{align*}
    \wh \epsilon_i & := \widehat{\xi}(\bs{Z}_i;\widehat{\pi},\widehat{\mu}) -
                     \sum_{j=1}^n\widehat{\xi}(\bs{Z}_j;\widehat{\pi},\widehat{\mu})/n
                     = \wh \xi_i - \PP_n \wh \xi,
     \\
    \wt \epsilon_i & :=   \xi(\bs Z_i; \overline \pi, \overline \mu) - \sum_{j=1}^n \xi(\bs Z_j; \overline \pi, \overline \mu) / n
                     = \wt \xi_i - \PP_n \wt \xi,
  \end{align*}
  and $\wh \epsilon_i^* = \delta_i \wh \epsilon_i $, and we also let
  $\wt \epsilon_i^* := \delta_i \wt \epsilon_i$.
   We let
  \begin{equation*}
    \wh \xi_i^*  := \wh \epsilon_i^* + \PP_n \wh \xi(\bs Z),
    \qquad     \text{ and } \qquad
    \wt \xi_i^*  := \wt \epsilon_i^* + \PP_n \wt \xi(\bs Z).
  \end{equation*}
  For notational ease let $W_{ha}(A) \equiv W_h(A-a) := g_{ha}^T \wh D_{ha}^{-1} g_{ha}(A) K_{ha}(A).$
  We let
  \begin{align*}
    \wh \theta^*(a) & := \PP_n^* W_h(A-a) (\delta ( \wh \xi - \PP_n \wh \xi) + \PP_n \wh \xi)
                      = \PP_n^* W_{ha}(A) \wh \xi_i, \\
    \wt \theta^*(a) & := \PP_n^* W_h(A-a) (\delta ( \wt \xi - \PP_n \wt \xi) + \PP_n \wt \xi)
                      = \PP_n^*  W_{ha}(A) \wt \xi_i.
  \end{align*}
  We can now decompose $T_n^*$ (as we decomposed $T_n$ in the proof of Theorem~\ref{thm:null-normality}), writing
  \begin{equation*}
    \wh \theta^*(a) - \PP_n \wh \xi^*
    = D_1^*(a) + D_2^*(a) + D_3^*,
  \end{equation*}
  with
  \begin{equation*}
    D_1^*(a) := \wh \theta^*(a) - \wt \theta^*(a),
    \quad
    D_2^*(a) := \wt \theta^*(a) - \PP_n^* \wt \xi^*,
    \;    \text{ and }     \;
    D_3^* := \PP_n^* \wt \xi^* - \PP_n^* \wh \xi^*.  
  \end{equation*}
  We consider the terms $\int (D_1^*)^2(a) da$, $\int D_1^*(a) D_2^*(a) da$, $\int (D_2^*)^2(a) da$, $\int D_1^*(a) D_3^*(a) da$, $\int D_2^*(a) D_3^*(a) da$, and $(D_3^*)^2$. Except for the $(D_2^*)^2$ term, we will show the rest are $o_p( 1 / \sqrt{n \sqrt{h}})$ (unconditionally, since no conditional argument is needed for negligible remainder terms by the definition of convergence in probability to $0$ in the (Dudley) metric $d(\cdot,\cdot)$).

  \bigskip
  \noindent \textbf{Step 1 ($(D_2^*)^2$).} \;   By
  Theorem 2 of \citet{Hardle:1993ih} (in combination with the proof of Theorem 2.1 of
  \citet{alcala1999goodness}), we have $d( \mc{L}^*( n \sqrt{h}  \int_{\mc A} (D_2^*(a))^2 da), \mc{L}( N(b_h, V))) \to_p 0$ as $n \to \infty$.

  \begin{mylongform}
    \begin{longform}

      I guess technically they prove this for the centering at $\wh \theta$ not at $\PP_n \wh \xi$.  But under the null the proof will go through.  Under the alternative I guess less clear. 
      
    \end{longform}
  \end{mylongform}

  \bigskip
  \noindent \textbf{Step 2 ($(D_3^*)^2$).}  \; We have
  $D_3^* = \PP_n^* ( \wh \xi^* - \wt \xi^*)$ which equals
  \begin{equation*}
    (\PP_n \wh \xi - \PP_n \wt \xi)
    + \PP_n^* \delta ( \wh \xi - \wt \xi +  \PP_n (\wh \xi - \wt \xi)).
  \end{equation*}
  Note that $\PP_n \wh \xi - \PP_n \wt \xi = D_3$, which was shown in the proof of Theorem~\ref{thm:null-normality} to be $o_p( 1 / \sqrt{n\sqrt{h}})$.  This also shows that $\PP_n^* \delta (\PP_n \wh \xi - \PP_n \wt \xi ) = o_p( 1 / \sqrt{n\sqrt{h}}) O_p(n^{-1/2})
  = o_p( 1 / \sqrt{n\sqrt{h}})$.

  It remains to analyze $\PP_n^* \delta( \wh \xi - \wt \xi)$, the symmetrized version of $D_3$.
  As in the analysis of $D_3$, we decompose this into an empirical process and a V-process.  The symmetrized V-process that arises is seen to be $O_p(n^{-1/2})$ because, by
    \cite{Arcones_Gine_1993} (see page 1509 and the argument there),
  it is of the same order of magnitude as the non-symmetrized V-process studied (and seen to be $O_p(n^{-1/2})$) in
  Lemma~\ref{lem:3}.

  The empirical process term is then broken into two terms, where we write
  \begin{equation*}
    \begin{split}
      \MoveEqLeft 
      \PP_n^* \delta (\xi(\bs Z; \overline \eta) - \xi(\bs Z; \wh \eta)) 
      =
      \inv{n}  \sum_{i=1}^n \delta_i  \ls
      (\xi(\bs Z; \overline \eta) - \xi(\bs Z; \wh \eta)) -
      \PP (\xi(\bs Z; \overline \eta) - \xi(\bs Z; \wh \eta)) \rs \\
      & \hspace{3.2cm} +
      \inv{n}  \sum_{i=1}^n \delta_i  
      \PP (\xi(\bs Z; \overline \eta) - \xi(\bs Z; \wh \eta))       
    \end{split}
  \end{equation*}
  (recall $\eta$ represents the two nuisance parameters).
  The first term is centered and so Lemma 2.3.6 of \cite{vanderVaart:1996tf} (which gives moment bounds of  a Rademacher symmetrized mean zero process in terms of the unsymmetrized process)
  applies, which, together with the analysis of $D_3$ in the proof of
  Theorem~\ref{thm:null-normality} (see Lemma~\ref{lem:step2-1}), shows that the first sum on the right in the display above is
  $O_p(n^{-1/2} + s_n^\infty r_n^\infty)$.  The remaining term,
  $\sum_{i=1}^n n^{-1} \delta_i \PP (\xi(\bs Z; \overline \eta) - \xi(\bs
  Z; \wh \eta))$,  is 
  $o_p( r_n^\infty s_n^\infty ) O_p(n^{-1/2}) = o_p( (n h^{1/2})^{-1/2}) O_p(n^{-1/2})$, 
  again
  by (the proof of) Lemma~\ref{lem:step2-1}.

  \bigskip   \noindent \textbf{Step 3 ($(D_1^*)^2$).} \;  We can decompose $D_1^*(a) = \wh \theta^*(a) - \wt \theta^*(a)$ into the sum of the following three summands:
  \begin{align}
    \PP_n^* W_{ha}(A) \delta (\wh \xi - \wt \xi), \label{eq:9}\\
    - \PP_n^* (W_{ha}(A) \delta  \PP_n( \wh \xi - \wt \xi )) ,\label{eq:10} \\
    \PP_n^* W_{ha}(A) ( \PP_n( \wh \xi - \wt \xi )) = \PP_n( \wh \xi - \wt \xi ) . \label{eq:11}
  \end{align}

  The term in \eqref{eq:11} is negligible from Step 2 of the proof of Theorem~\ref{thm:null-normality}.  The term \eqref{eq:10} is similarly negligible.
  
  Finally we consider \eqref{eq:9}.  Like in the proof of
  Theorem~\ref{thm:null-normality}
  we break this into a
  (standard) empirical process term and a V-process term. For the former
  term, we apply Lemma 2.3.6 of \cite{vanderVaart:1996tf}, which bounds the
  expectation of the empirical process implied by \eqref{eq:9} by the
  expectation of the empirical process implied by
  $\PP_n W_{ha}(A) (\wh \xi - \wt \xi)$, which was shown to be negligible in
  Step 3 of the proof of
  Theorem~\ref{thm:null-normality}. 
  Analogously, for the V-process term, we use results of
  \cite{Arcones_Gine_1993} (see page 1509 and the argument there), together
  with the argument in Step 3 of the proof of
  Theorem~\ref{thm:null-normality} 
  for the analogous
  V-process (without the $\delta$ multiplier) to see that term is also
  negligible.  Thus, we conclude that  \eqref{eq:9} is negligible.

  \bigskip \noindent\textbf{Step 4 ($D_2^* D_3^*$).} \; From Part 2, we have that $D_3^* = o_p( 1 / \sqrt{n \sqrt{h}})$, so to analyze $D_3^* \int_{\mc A} D_2^*(a) da$ it remains to analyze $\int_{\mc A} D_2^*(a) da$.  This reduces to analyzing $\int_{\mc A} (\wt \theta^*(a) - \PP \wt \xi) da$.  By the proof of Lemma~\ref{lem:d2d3}, $\int ( \wt \theta^*(a) - \PP( \wt \xi)) da = n^{-1} \sum_{i=1}^n \delta_i \lp \varpi_0^{-1}(A_i) + O(h) \rp (\wt \xi_i - \PP( \wt \xi)).$ Then by the (classical) Central Limit Theorem, we conclude that $\int_{\mc A} D_2^*(a)da = O_p(n^{-1/2})$, and
  thus $D_3^* \int_{\mc A} D_2^*(a) da = o_p(1/ n h^{1/4})$.
  
  \bigskip \noindent
  \textbf{Step 5 ($D_1^* D_3^*$).} \;  We can see that $D_3^* \int_{\mc A} D_1^*(a) da = o_p( 1 / \sqrt{n \sqrt{h}})$ by the results of Step 2 and Step 3. 

  \bigskip
  \noindent\textbf{Step 6 ($D_1^* D_2^*$).} \;  We now analyze $\int_{\mc A} D_1^*(a) D_2^*(a) da.$  Write $\wt \theta_h^*(a) - \PP_n^* \wt \xi^* = \wt \theta_h^* - \PP \wt \xi + \PP \wt \xi - \PP_n^* \wt \xi^*$.  We have $\PP_n^* \wt \xi^* = n^{-1} \sum_{i=1}^n ( \delta_i ( \wt \xi_i^* - \wt \theta_h(A_i)) + \PP_n \wt \xi)$ where $\PP_n \wt \xi \to \PP \wt \xi$ almost surely.  Let $\EE^* $ and $\text{Var}^*$ denote mean and variance conditional on the data (so, averaging over $\{ \delta_i \}_i$).  Then
  \begin{align*}
    \EE^* ( n^{-1}  \sum_{i=1}^n \delta_i ( \wt \xi_i -  \PP_n \wt \xi )) & = 0,\\
    \text{Var}^* \lp n^{-1}  \sum_{i=1}^n \delta_i ( \wt \xi_i - \PP_n \wt \xi ) \rp & = n^{-2} \sum_{i=1}^n ( \wt \xi_i - \PP_n \wt \xi )^2.
  \end{align*}
  We can see that $n^{-2} \sum_{i=1}^n ( \wt \xi_i - \PP_n \wt \xi)^2 \to_p \Var(\wt \xi ( \bs Z))$ as $n \to \infty$ under $H_0$.  Thus, $ \PP_n^* \wt \xi^* - \PP \wt \xi$ reduces to $\PP_n \wt \xi - \PP \wt \xi = D_3$ which by the proof of Theorem~\ref{thm:null-normality} is $o_p( 1 / \sqrt{n \sqrt{h}})$.
  
  Thus as in Step 5 we consider
  \begin{equation}
    \label{eq:16}
    \int_{\mc A} ( \wh \theta_h^*(a) - \wt \theta_h^*(a)) ( \wt \theta^*_h(a) - \PP \wt \xi) da,  
  \end{equation}
  which, after ignoring the $\int_{\mc A} (\PP_n (\wh \xi - \wt \xi)) ( \wt \theta^*_h(a) - \PP \wt \xi) da= D_3 \int_{\mc A} D_2^*(a)da$ summand which has already been shown to be negligible (in Step  4), equals
  \begin{equation}
    \label{eq:15}
    \int_{\mc A} n^{-1} \sum_{i=1}^n W_{ha}(A_i) \delta_i ( \wh \xi_i - \wt \xi_i +
    \PP_n( \wh \xi - \wt \xi) ) 
    ( \wt \theta_h^*(a) - \PP \wt \xi) da. 
  \end{equation}
  The term
  \begin{equation*}
    \int_{\mc A} n^{-1} \sum_{i=1}^n W_{ha}(A_i) \delta_i ( \wh \xi_i - \wt \xi_i) ( \wt \theta^*_h(a) - \PP \wt \xi ) da
  \end{equation*}
  is again decomposed (analogously to the decompositions in \eqref{eq:D3-decomp},
  \eqref{eq:102-extra});  we get two V-processes, of second and third order, as in Step 6 of
the proof  of Theorem~\ref{thm:null-normality}. The two terms are handled in a similar fashion as  previously, using Proposition~\ref{prop:1}.

  
  Then the remaining term,
  \begin{equation*}
    \int_{\mc A} n^{-1} \sum_{i=1}^n W_{ha}(A_i) \delta_i 
    (\PP_n( \wh \xi - \wt \xi))
    ( \wt \theta^*_h(a) - \PP \wt \xi) da     ,
  \end{equation*}
  can be studied in analogous fashion to the study of $\int_{\mc A} D_2(a) da$ in Lemma~\ref{lem:d2d3}, except that a V-statistic arises instead of an average, as follows.  We write the term in the above display as $ D_3 \int_{\mc A} n^{-1} \sum_{i=1}^n W_{ha}(A_i) \delta_i ( \wt \theta^*_h(a) - \PP \wt \xi) da $, which is
  \begin{equation*}
    D_3
    n^{-2}    \sum_{i=1}^n \int_{\mc A} W_{ha}(A_i) \delta_i
    \lp \sum_{j=1}^n W_{ha}(A_j) (\delta_j \wt \epsilon_j + \PP \wt \xi) - \PP \wt \xi \rp
    \, da.
  \end{equation*}
  By the property of local polynomial equivalent kernels (see (3.12) page 63
  \cite{Fan:1996tk}) that says that they preserve identically polynomials, we have
  $ \sum_{j=1}^n W_{ha}(A_j) \PP \wt \xi = \PP \wt \xi$ and so the above display reduces to
  \begin{align*}
    D_3
\lp     n^{-2}    \sum_{i=1}^n \sum_{j=1}^n \delta_i \delta_j \wt \epsilon_j \int_{\mc A} W_{ha}(A_i) 
    W_{ha}(A_j) 
    \, da \rp.
  \end{align*}
  The integral $ \int_{\mc A} W_{ha}(A_i) 
    W_{ha}(A_j) 
    \, da$ is of order $1/h$ (use the representation of the first order equivalent kernel, page 63
    \cite{Fan:1996tk}, and do a change of variables, say $u = (A_i - a)/h$).  This is sufficient to conclude that the V-statistic (not process) term inside parentheses in the above display is order $1 / nh$ \citep[Chapter 12]{vanderVaart:1998dr}.  Multiplying by $D_3 = O_p( r_n^\infty s_n^\infty) =
    o_p( (n h^{1/2})^{-1/2})$,     we see that the entire term is negligible. 
  \begin{mylongform}
    \begin{longform}
      One can then consider $\sum_j \delta_j \wt \epsilon_j V(A_i,A_j)/h$ with $V$ (the integral) as a order $1/h$ kernel and compute variance and bias to reduce a sqrt h away, i guess.
    \end{longform}
  \end{mylongform}
  This was the last term in Step 6 and so we have completed the proof.
\end{proof}

\section{Proof of Theorem~\ref{thm:bootstrap-consistency-null-2}}

Here we provide (only) the modifications to the proof of Theorem \ref{thm:bootstrap-consistency-null} needed to prove Theorem \ref{thm:bootstrap-consistency-null-2}.  The difference is in the definition of the $\epsilon$ variables, which are now defined to be
  \begin{align*}
    \wh \epsilon_i & := \widehat{\xi}(\bs{Z}_i;\widehat{\pi},\widehat{\mu}) -
                    \wh \theta_h(A_i)
    \\
    \wt \epsilon_i & :=   \xi(\bs Z_i; \overline \pi, \overline \mu) -
                     \wt \theta_h(A_i).
  \end{align*}

  \begin{proof}[Proof of Theorem~\ref{thm:bootstrap-consistency-null-2}]
    Again the proof is divided  into 6 steps. 
    
    \bigskip
    \noindent \textbf{Step 1 ($(D_2^*)^2$).} \; As previously, by Theorem 2 of \citet{Hardle:1993ih} (in combination with the proof of Theorem 2.1 of \citet{alcala1999goodness}), we have $d( \mc{L}^*( n \sqrt{h} \int_{\mc A} (D_2^*(a))^2 da), \mc{L}( N(b_h, V))) \to_p 0$ as $n \to \infty$.

  \bigskip
  \noindent \textbf{Step 2 ($(D_3^*)^2$).}  \; We have
  $D_3^* = \PP_n^* ( \wh \xi^* - \wt \xi^*)$ which equals
  \begin{equation*}
    (\PP_n \wh \xi(\bs Z) - \PP_n \wt \xi(\bs Z))
    + \PP_n^* \delta ( \wh \xi(\bs Z) - \wt \xi(\bs Z) +  \wh \theta_h(A) - \wt \theta_h(A)).
  \end{equation*}
  The term $\PP_n^* \delta (\wh \theta_h(A) - \wt \theta_h(A) )$ is the only one that is not analyzed and shown to be negligible in the proof of
  Theorem~\ref{thm:bootstrap-consistency-null}.
  Recall that (by definition)  $D_1(A) =  \wh \theta_h(A) - \wt \theta_h(A)$ and we can use the decomposition we have used previously for $D_1$, namely,
  \begin{equation}
    \label{eq:5}
    \begin{split}
      D_1(A_i)
      = \inv{n} \sum_{j=1}^n W_h(A_i-A_j) (\wt \xi_j - \wh \xi_j )
      &=
      d_{1,1}(A_i) + d_{1,2}(A_i) \\
      &      =
      d_{1,1}(A_i) +
      R_{n,1, A_i} + R_{n,2, A_i}
    \end{split}
  \end{equation}
  (recall the definitions of $d_{1,1}$, $d_{1,2}$ in \eqref{eq:13} and the definitions of $R_{n,i, a}$ in \eqref{eq:17}).  (Again, because of the integral over $\mc{A}$, we cannot analyze the terms separately.)  In Lemmas \ref{lem:5} and \ref{lem:step3-1} we have shown $\sup_{a \in \mc{A}}$ bounds for $d_{1,1}(a) $ and $R_{n,2,a}$, and those thus show that the corresponding sums here, $\PP_n^* \delta d_{1,1}(A)$ and $\PP_n^* \delta R_{n,2,A}$, are asymptotically negligible.  For $R_{n,1,A_i}$, the treatment now is slightly different than in Lemma \ref{lem:step3-1}.  Previously $R_{n,1, a}$ was treated as an empirical (order 1 V-) process whereas now we treat $\PP_n^* \delta R_{n,1,A}$ as an order 2 (degenerate) V-process.  The envelope and entropy calculations from the proof of Lemma \ref{lem:step3-1} still apply almost verbatim (the presence of $\delta$ changes neither), and then we apply Proposition~\ref{lem:100} (or Proposition~\ref{prop:1}).  We thus conclude that $D_3^* = o_p( 1 / \sqrt{n \sqrt{h}})$.

  \bigskip   \noindent \textbf{Step 3 ($(D_1^*)^2$).} \;  We consider $\int_{\mc{A}} D_1^*(a)^2 da$; we can decompose $D_1^*(a) = \wh \theta^*(a) - \wt \theta^*(a)$ into the sum of the following three summands:
  \begin{align}
    \PP_n^* W_{ha}(A) \delta (\wh \xi - \wt \xi), \label{eq:9}\\
    - \PP_n^* (W_{ha}(A) \delta  (\wh \theta_h(A) - \wt \theta_h(A))) ,\label{eq:10} \\
    \PP_n^* W_{ha}(A) (\wh \theta(A) - \wt \theta(A)), \label{eq:11}
  \end{align}
  and need to analyze their squared integrals (which suffices by Cauchy-Schwarz). 
  The term \eqref{eq:10} is the one that changed from the previous proof. The integral of that term squared is 
  \begin{equation}
    \label{eq:23}
    \int_{\mc A}
    \ls
    n^{-2}    \sum_{i=1}^n
    \big(
    \sum_{j=1}^n W_h(A_i-A_j) (\wt \xi_j - \wh \xi_j )
    \big)
    W_{ha}(A_i) \delta_i
    \rs^2
    \, da
    .    
  \end{equation}
  The term \eqref{eq:23} is asymptotically negligible (is $o_p( 1 / n \sqrt{h})$) which follows a somewhat similar recipe as used in the previous step.  We again use a decomposition as in \eqref{eq:5} for $D_1(A) = \wh \theta_h(A) - \wt \theta_h(A)$.  Since $(b+c)^2 \le \max(4b^2, 4c^2)$, it suffices to bound each of
  \begin{align}
    & \int_{\mc A} \ls \PP_n^* (  W_{ha}(A) \delta d_{1,1}(A)) \rs^2 \, da
    \label{eq:T1} \\
    &    \int_{\mc A} \ls \PP_n^* (  W_{ha}(A) \delta R_{n,1,A}) \rs^2 \, da ,
    \label{eq:T2} \\
    &    \int_{\mc A} \ls \PP_n^* (  W_{ha}(A) \delta R_{n,2,A}) \rs^2 \, da .
      \label{eq:T3} 
  \end{align}
  The handling of these three terms is somewhat analogous to what was done in the previous step.  Again, by the proofs of Lemmas \ref{lem:5} and \ref{lem:step3-1}, we have $\sup_{a \in \mc{A}}$ bounds for $|d_{1,1}(a)|$ and $|R_{n,2,a}|$ of order $o_p( 1 / \sqrt{n \sqrt{h}})$ which allows us to see that \eqref{eq:T1} and \eqref{eq:T3} are also $ O_p(1/ nh) o_p( 1 / n \sqrt{h}) = o_p( 1 / n \sqrt{h})$.

  For the term \eqref{eq:T2}, we slightly extend the proof of Lemma \ref{lem:step3-1} (to include the $W_{ha}(A) \delta$ terms in the corresponding function class and treat the term as an order 2 (degenerate) V-process) to see that the term inside the integrand has finite second moment and thus (by interchanging the integral and expectation) that the entire term is negligible.

  \bigskip \noindent\textbf{Step 4 ($D_2^* D_3^*$).} \; No new argument is needed here (since we rely in part on the result from Part 2 of the current proof, and because $D_2^*$ is defined via oracle values).

  \bigskip \noindent
  \textbf{Step 5 ($D_1^* D_3^*$).} \;  Again, we can see that $D_3^* \int_{\mc A} D_1^*(a) da = o_p( 1 / \sqrt{n \sqrt{h}})$ by the results of Step 2 and Step 3. 

  \bigskip
  \noindent\textbf{Step 6 ($D_1^* D_2^*$).} \;  We now analyze $\int_{\mc A} D_1^*(a) D_2^*(a) da.$
  The analysis introduces a fourth order V-process and so we rely on the finiteness of the $J_4$ entropy term.
  As in Step 6 of the proof of   Theorem~\ref{thm:bootstrap-consistency-null}
  we consider
  \begin{equation}
    \label{eq:16}
    \int_{\mc A} ( \wh \theta_h^*(a) - \wt \theta_h^*(a)) ( \wt \theta^*_h(a) - \PP \wt \xi) da,  
  \end{equation}
  which is again reduced to 
  \begin{equation}
    \label{eq:15}
    \int_{\mc A} n^{-1} \sum_{i=1}^n W_{ha}(A_i) \delta_i ( \wh \xi_i - \wt \xi_i - (\wh \theta_h(A_i) - \wt \theta_h(A_i) )) ( \wt \theta_h^*(a) - \PP \wt \xi) da. 
  \end{equation}
  The only  term not already handled in
  the proof of Theorem~\ref{thm:bootstrap-consistency-null} is 
  \begin{equation*}
    \int_{\mc A} n^{-1} \sum_{i=1}^n W_{ha}(A_i) \delta_i ( \wt \theta_h(A_i) - \wh \theta_h(A_i))
    ( \wt \theta^*_h(a) - \PP \wt \xi) da     ,
  \end{equation*}
  which (somewhat in parallel to \eqref{eq:23})
  is reduced to (after decomposing $\wt \theta^*_h(a) - \PP\wt \xi = \PP_n^* \ls W_{ha}(A) ( \delta ( \wt \xi - \wt \theta_h(a)) + \wt \theta_h(a)) \rs - \PP \wt \xi$ and ignoring the smaller $\wt \theta_h(a) - \PP \wt \xi$ term) the expression
  \begin{equation}
    \label{eq:18}
    n^{-3}    \sum_{i=1}^n
    \ls
    \big(
    \sum_j W_h(A_i-A_j) (\wt \xi_j - \wh \xi_j )
    \big)
    \int_{\mc A}  W_{ha}(A_i) \delta_i
    \big(
    \sum_k W_{ha}(A_k) \delta_k \wt \epsilon_k
    \big)
    \, da \rs
    .    
  \end{equation}
  We again use the decomposition given in \eqref{eq:5}.
  Previously, $d_{1,1}$ yielded V-processes (one for each of $\pi,\mu$), $R_{n,1,a}$ corresponded to an empirical process expression and $R_{n,2,a}$ corresponded to a `second order remainder' type of term (handled in different ways previously).  From
   \eqref{eq:5} we decompose \eqref{eq:18} into a sum of three terms,
  \begin{align}
    &    n^{-2}    \sum_{i=1}^n
    \ls
    d_{1,1}(A_i)
    \int_{\mc A}  W_{ha}(A_i) \delta_i
    \big(
    \sum_k W_{ha}(A_k) \delta_k \wt \epsilon_k
    \big)
    \, da \rs, \label{eq:19} \\
    & n^{-2}    \sum_{i=1}^n
    \ls
    R_{n,1, A_i}
    \int_{\mc A}  W_{ha}(A_i) \delta_i
    \big(
    \sum_k W_{ha}(A_k) \delta_k \wt \epsilon_k
    \big)
    \, da \rs, \qquad \text{ and }  \label{eq:20} \\
    & n^{-2}    \sum_{i=1}^n
    \ls
    R_{n, 2 , A_i}
    \int_{\mc A}  W_{ha}(A_i) \delta_i
    \big(
    \sum_k W_{ha}(A_k) \delta_k \wt \epsilon_k
    \big)
    \, da \rs. \label{eq:21}
  \end{align}
  The terms \eqref{eq:20} and \eqref{eq:21} are handled  in a fashion analogous to the handling of the ``$R_{n,1,a},$'' and ``$R_{n,2,a}$'' type terms previously (see Lemmas \ref{lem:step3-1} and \ref{lem:step5-decompose-term2})
  although with one more summation so yielding V-processes of order 3 and 2 with modified/extended function classes from the ones considered previously.

  The term \eqref{eq:19} (as in previous analyses of the ``$d_{1,1}$ term''), is a sum of 
  V-processes for $\mu$ and for $\pi$ which are handled similarly.  Consider the case of $\mu$ (with the $\pi$ term being analogous).  Let $\bs W_i= (\delta_i, \bs Z_i)$ and then define the V-process `kernel'
  $H \equiv H_\mu$ by
  \begin{equation*}
    H(\bs W_1, \ldots , \bs W_4)
    := \int_{\mc A} \wt \mu(\bs L_4, A_2) \delta_1 W_{ha}(A_1) W_h(A_1-A_2) \delta_3
    W_{ha}(A_3) \wt \epsilon_3 da,
  \end{equation*}
  where recall $\wt \epsilon := \xi(\bs Z; \overline \eta) - \PP \xi(\bs Z; \overline \eta)$.
  (Recall
  we let  tilde $ \wt{ \cdot  } $  operate on  any $\mu, \pi$ to yield
  $\wt{\pi}(a_1|l_2) := \pi(a_1|l_2) - \PP \pi(a_1|L),$
  and similarly $\wt\mu(l_2,a_1) := \mu(l_2,a_1) -  \PP \mu(L,a_1)$.)

  Recall the discussion of V- or U-statistics and -processes from Section~\ref{sec:v-process-results};
  We do not repeat that discussion here, except to recall the intuition that a V-statistic of order $k$ that is maximally degenerate (i.e., degenerate of degree $k-1$ meaning that averaging over any 1 argument yields $0$) and with a kernel not changing with $n$ can be expected to be of order $n^{-k/2}$ \citep[Chapter 12]{vanderVaart:1998dr}.
  V-statistics or V-processes have a so-called Hoeffding decomposition, and
  the least degenerate term generally governs the size of a V-statistic.  Now we let $G_{1, S} \equiv G_{1,S; \mu}$, for a set $S \subset \{1,2,3,4\}$ be the function given by averaging over independent copies of the $\{1,2,3,4\} \setminus S$ variables in $H$, where the (slightly redundant) $1$ indicates $S$ has size $1$.  So, e.g., $G_{1,\{1\}}(w_1) := \PP^3 H(w_1, \bs W_2, \bs W_3 , \bs W_4)$.  Then it is quick to check that
  \begin{equation*}
    G_{1, \{i\}} \equiv 0 \text{ for } i \in \{1,2,3,4\}.
  \end{equation*}
  Similarly, because averaging over any of the 1,3,4 coordinates yields $0$ (because of say $\delta_1, \delta_3, $ $\wt \mu(\bs L_4, A_2)$ (where $\wt \mu$ is centered on its first argument), $G_{2, \{i_1, i_2\}} \equiv 0$ (for any $i_1 \ne i_2$).
  Then $G_{3, \{i_1,i_2,i_3\}}$ is only nonzero for $\{i_1,i_2, i_3 \} = \{1,3,4\}$ in which case
  \begin{equation*}
    \begin{split}
      \MoveEqLeft
      G_{3 , \{1,3,4 \}}(w_1, w_3, w_4 ) \\
      & = \delta_3 \wt \epsilon_3 \delta_1 ( \PP \wt \mu(l_4, A_2) W_{ha}(A_2))
      \int_{\mc A} W_{h}(a-a_1) W_{h}(a-a_3) da.
    \end{split}
  \end{equation*}
  Thus, the only nondegenerate terms to consider arise from the (order 3, maximally degenerate) kernel $G_{3, \{1,3,4 \}}(w_1, w_3, w_4)$
  and the (order 4, maximally degenerate) kernel $(w_1,w_2,w_3,w_4) \mapsto H(w_1,w_2,w_3,w_4) - G_{3, \{1,3,4 \}}(w_1, w_3, w_4).$
  Both are maximally degenerate (averaging over any coordinate yields the zero function).
  (Typically, if the kernel $H$ did not depend on $n$ (via $h \equiv h_n$), the orders of magnitude would thus be $O_p(n^{-3/2})$ and $O_p(n^{-4/2})$; as we will see, indeed the first term is of larger size.)

  We can see that (with a change of variables) $\PP \wt \mu( l_4, A_2) W_h(a_1-A_2) \le 2 \| \varpi_0 \|_\infty \| \mu \|_\infty$ where $\| \mu \|_\infty$ is uniformly bounded by assumption, and where the right side does not depend on $h$.  The term $ \int_{\mc A} W_{h}(a-a_1) W_{h}(a-a_3) da$, after change of variables, is seen to be of order $O(1/h)$.  Thus we have an envelope for the class of functions $ G_{3 , \{1,3,4 \}}(w_1, w_3, w_4 )$ of size $O(1/h)$.  The class of functions $\{ H \}$ has an envelope of size $O(1/h^2)$ (after just doing one change of variables). 

  We now note that $J_4( 1, \mc{H}, L_2) < \infty$ where $\mc{H} := \{ H_\mu : \mu \in \mc{F}_\mu \}$; by Lemma~\ref{lem:integral-decreases-entropy} (and Lemma~\ref{lem:andrews-add}) this will also show that $J_3(1, \{ G_{3, \{ 1,3,4\}; \mu} : \mu \in \mc{F}_\mu \}, L_2) < J_4 \lp 1, \{ G_{3, \{ 1,3,4\}; \mu} : \mu \in \mc{F}_\mu \}, L_2 \rp < \infty$.
  That $J_4( 1, \mc{H}, L_2) < \infty$ follows from the assumption that $J_4( 1, \mc{F}_\mu, L_2) < \infty$ (i.e., by    Assumption~\ref{assm:EA-item2}$_4$) and from the preservation arguments we have made previously (see e.g., the proof of Lemma~\ref{lem:4}).
Thus the larger term is $O_p( 1 / hn^{3/2})$ (and the smaller bounded by $O_p( 1/ h^2n^2)$).
\end{proof}




\bibliographystyle{plainnat}
\bibliography{full}


\end{document}